\documentclass[aps,prx,superscriptaddress,notitlepage, amsfonts,longbibliography, twocolumn]{revtex4-2}

\usepackage[dvipdfmx]{graphicx}

\usepackage{amsmath,amssymb,amsthm,mathtools,mathrsfs}
\usepackage{srcltx}

\usepackage{bm,braket,comment}
\usepackage{amsmath, amssymb,amsthm,mathtools}
\usepackage{xcolor}
\usepackage{longtable}
\usepackage{enumerate}
\usepackage{ulem}
\usepackage{multirow}
\usepackage{here}
\usepackage{comment}
\usepackage{url}

\newtheorem{theorem}{Theorem}
\newtheorem{definition}{Definition}
\newtheorem{lemma}{Lemma}

\newcommand{\bout}[1]

\usepackage{hyperref}
\hypersetup{colorlinks=true,linkcolor=blue,citecolor=blue,urlcolor=blue}

\allowdisplaybreaks[4]
\begin{document}

\title{Characterization of randomness in quantum circuits of continuous gate sets}

\author{Yosuke Mitsuhashi}
\email{mitsuhashi@noneq.t.u-tokyo.ac.jp}
\affiliation{Department of Basic Science, University of Tokyo, 3-8-1 Komaba, Meguro-ku, Tokyo 153-8902, Japan}

\author{Ryotaro Suzuki}
\email{ryotaro.suzuki@fu-berlin.de}
\affiliation{Dahlem Center for Complex Quantum Systems, Freie Universität Berlin, Berlin 14195, Germany}

\author{Tomohiro Soejima}
\email{tomohiro\_soejima@g.harvard.edu}
\affiliation{Department of Physics, Harvard University, Cambridge, MA 02138, USA}

\author{Nobuyuki Yoshioka}
\email{nyoshioka@ap.t.u-tokyo.ac.jp}
\affiliation{Department of Applied Physics, University of Tokyo, 7-3-1 Hongo, Bunkyo-ku, Tokyo 113-8656, Japan}
\affiliation{Quantum Computing Center, RIKEN Cluster for Pioneering Research (CPR), Wako-shi, Saitama 351-0198, Japan}
\affiliation{JST, PRESTO, 4-1-8 Honcho, Kawaguchi, Saitama, 332-0012, Japan}

\begin{abstract}
In the accompanying paper of arXiv:2408.13472, we have established the method of characterizing the maximal order of asymptotic unitary designs generated by symmetric local random circuits, and have explicitly specified the order in the cases of $\mathbb{Z}_2$, $\mathrm{U}(1)$, and $\mathrm{SU}(2)$ symmetries. 
Here, we provide full details on the derivation of the main theorems for general symmetry and for concrete symmetries. 
Furthermore, we consider a general framework where we have access to a finite set of connected compact unitary subgroups, which includes symmetric local unitary gate sets. 
\end{abstract}

\maketitle

\onecolumngrid

\let\oldaddcontentsline\addcontentsline
\renewcommand{\addcontentsline}[3]{}

\section{Introduction}

In quantum mechanics, symmetry plays a fundamental role in both constraining and enriching phenomena in wide range of spatial and dynamical sense.
An early seminal example is represented by Noether's theorem~\cite{Noether1918}, which states that a global symmetry of a system results in a constrained dynamics that preserves  a conserved charge.
Symmetry also plays a crucial role in enriching physics, as represented in spontaneous symmetry breaking~\cite{nambu1960axial, nambu1961dynamical, goldstone1961field, nambu2009nobel} and deconfined quantum criticality~\cite{senthil2004deconfined, senthil2004quantum, sandvik2007evidence}.
A prominent application in quantum information science is the protection of quantum memory by quantum error correction~\cite{shor1995scheme, steane1996error, calderbank1996good, nielsen2000quantum}, in addition to the Eastin-Knill theorem that in turn puts restriction on a single error-correcting code to perform universal quantum computation~\cite{eastin2009restrictions}.

Driven by the capability of quantum circuit models to capture various statistical physics phenomena,  there is a surging interest in the interplay between symmetry and locality in quantum circuits.
One primary example is the symmetry-protected topological order in quantum phases of matter, in which the presence/absence of  constant-depth local symmetric quantum circuit is crucial for the definition~\cite{wen1995topological, kane2005quantum, hasan2010colloquium, sato2017topological, gu209tensor, pollman2010entanglement, chen2012science}.
Also, the interplay has been of interest to the statistical physics community which employs the quantum circuit model to describe discretized time evolution of local Hamiltonian to discover novel symmetry enriched phases in both static and dynamical ways~\cite{bao2021symmetry, lavasani2021measurement, morral-yepes2023detecting, hauser2024continuous}.

The interplay between the symmetry and locality has shed light on a primary problem in quantum information science--- the  universality of symmetric local quantum circuits. Here the universality refers to the ability of a given set of local quantum gates to express arbitrary global unitary, and its practical significance is highlighted in the Solovay-Kitaev theorem which states that $\epsilon$-close approximation of arbitrary unitary can be constructed from polylogarithmic number of universal gate sets~\cite{kitaev1997quantum, dawson2006solovay}.
While the fundamental theory of quantum computing has established that universality can be achieved with a finite set of locally universal  unitaries~\cite{divincenzo1995universal, lloyd1995almost},
surprisingly it was shown recently that the representability of {\it symmetric} local circuit is restricted, i.e., does not satisfy universality~\cite{marvian2022restriction}.
It was later pointed out that some local circuits under symmetry constraints satisfy a property called the {\it semi-universality}~\cite{marvian2023theory, marvian2024rotationally, hulse2024framework}, a weaker version of universality which ignores the tunability of relative phases between symmetry sectors~\cite{kempe2001encoded}.

The discovery of such a qualitative difference has further invoked question in terms of quantitative characterization, concretely in terms of the symmetric version of unitary $t$-design. Note that unitary design,  representing a set of unitaries that reproduces the Haar measure up to the $t$th moment~\cite{dankert2009exact}, has been a standard tool to understand the condition to perform various tasks in quantum information science including quantum advantage~\cite{boixo2018characterizing, arute2019quantuma}, quantum tomography~\cite{huang2020predicting}, randomized benchmarking~\cite{emerson2005scalable}, optimal quantum communication capacity~\cite{dupuis2014one}, and chaotic dynamics~\cite{roberts2017chaos}. 
While it is known that accumulation of non-symmetric local circuits allows us to generate unitary designs up to arbitrary order~\cite{brandao2016local, haferkamp2022random, mittal2023local, belkin2024approximate, schuster2024random, laracuente2024approximate}, in symmetric cases the expressibility of the symmetric local unitaries remains unestablished.
While there are existing attempts to characterize the design under U(1) and SU(d) symmetries~\cite{li2023designs, hearth2025unitary}, we are lacking of integrated theory that provides the exact number of maximal order $t$ achievable with symmetric local quantum circuits.

In an accompanying letter~\cite{mitsuhashi2025unitary}, we establish the method for general symmetric local quantum circuits that characterizes its expressibility in terms of symmetric unitary design. 
We have concretely shown that the necessary and sufficient condition of forming an asymptotic symmetric $t$-design is given by the nonexistence of a nontrivial integer solution of a certain linear equation specified by the symmetry and locality of the circuit.
In this manuscript, we provide the full details of the derivation of the main theorems. 
The equivalence between asymptotic unitary designs and the nonexistence of a nontrivial integer solution can be intuitively understood as follows:
Since we consider the situation where our accessible gate set is semi-universal, the difference between the expressibility of accessible gate set and that of the whole symmetric gate set appears only in the relative phases.
Therefore, the distribution of a random circuit generated with some gate set is an asymptotic unitary $t$-design if and only if whenever we are given the sum of $t$ relative phases, we can estimate the component of them, which can be equivalently expressed as the nonexistence of nontrivial integer solutions of a certain set of equations.

As for technical perspective, the core idea is to show the equivalence between the nonexistence of a nontrivial integer solution and the coincidence of the commutant of the $t$-fold allowed gate set and that of the $t$-fold symmetric unitary operators, which means that symmetric local quantum circuits are asymptotic unitary $t$-designs. 
When we prove the coincidence of the two commutants, we show the coincidence of the algebras of the $t$-fold allowed gate set of the $t$-fold symmetric unitary operators. 
On the other hand, when we prove the converse part, we explicitly construct an operator that commutes with all the $t$-fold allowed gates, but not with all the $t$-fold symmetric unitaries.

The remainder of this paper is organized as follows.
In Sec.~\ref{SMsec:preliminaries}, we introduce the preliminaries. 
In Sec.~\ref{SMsec:main_results}, we present a theorem about the explicit order of unitary designs, which is applicable to general symmetries and general continuous gate sets. 
We also present the detailed results for the $\mathrm{Z}_2$, $\mathrm{U}(1)$, and $\mathrm{SU}(2)$ symmetries. 
Then, in Sec.~\ref{SMsec:proof_general}, we present the proof of the general theorem. 
This is followed by Sec.~\ref{sec:conclusion} which gives the conclusion and discussion.
For the completeness of our work, in Appendix~\ref{SMsec:theorem_specific}, we present the proof of the theorems about the concrete symmetries $\mathbb{Z}_2$, $\mathrm{U}(1)$, and $\mathrm{SU}(2)$. 
In Appendix~\ref{SMsec:technical}, we show technical lemmas used in the proofs of the main theorems.

{
\hypersetup{linkcolor=blue}
}


\section{Preliminaries} \label{SMsec:preliminaries}

The notations used in this paper are as follows: 
For a general Hilbert space $\mathcal{K}$, we denote the sets of all linear operators and all unitary operators on $\mathcal{K}$ by $\mathcal{L}(\mathcal{K})$ and $\mathcal{U}(\mathcal{K})$, respectively. 
As for the definition of the Lie algebra associated with a Lie group, we adopt the physical version, i.e., we define the Lie algebra as the tangent space at the identity divided by the imaginary unit $i$. 
We define $\mathbb{N}:=\{n\in\mathbb{Z}\ |\ n>0\}$ and $\mathbb{Z}_{\geq 0}:=\{n\in\mathbb{Z}\ |\ n\geq 0\}$. 
For the sake of convenience, we define the sum and product over the empty set as $0$ and $1$, respectively.

We consider a circuit consisting of $n$ qudits with a local dimension $d$, and we denote the associated Hilbert space by $\mathcal{H}$. 
For convenience, we denote the set of all linear operators and all unitary operators on the $n$ qudits by $\mathcal{L}_n$ and $\mathcal{U}_n$, respectively, which are the same as $\mathcal{L}(\mathcal{H})$ and $\mathcal{U}(\mathcal{H})$. 
In the following, we give the notations about symmetry and the construction of random circuits. 
First, we explain the symmetry condition. 
By using the pair of a group $G$ and its representation $R$ on $\mathcal{H}$, we say that an operator $O\in\mathcal{L}(\mathcal{H})$ is $(G, R)$-symmetric if $O$ commutes with $R(g)$ for all $g\in\mathcal{G}$. 
We denote the set of all $(G, R)$-symmetric linear operators and unitary operators by $\mathcal{L}_{n, G, R}$ and $\mathcal{U}_{n, G, R}$, i.e., 
\begin{align}
    &\mathcal{L}_{n, G, R}:=\{L\in\mathcal{L}_n\ |\ [L, R(g)]=0\ \forall g\in G\}, \\
    &\mathcal{U}_{n, G, R}:=\{U\in\mathcal{U}_n\ |\ [U, R(g)]=0\ \forall g\in G\}. 
\end{align}

As examples of representations on multiqudit systems, we can take the following representations of three groups $\mathbb{Z}_2$, $\mathrm{U}(1)$, and $\mathrm{SU}(2)$ on $n$ qubits: 
\begin{align}
	&R(g)=\left(\mathrm{Z}^g\right)^{\otimes n}\ \textrm{when}\ G=\mathbb{Z}_2=\{0, 1\}, \label{SMeq:Z2_representation}\\
	&R\left(e^{i\theta}\right)=\left(e^{i\theta \mathrm{Z}}\right)^{\otimes n}\ \textrm{when}\ G=\mathrm{U}(1), \label{SMeq:U1_representation}\\
	&R\left(e^{i(\theta_\mathrm{X} \mathrm{X}+\theta_\mathrm{Y} \mathrm{Y}+\theta_\mathrm{Z} \mathrm{Z})}\right)
    =\left(e^{i(\theta_\mathrm{X} \mathrm{X}+\theta_\mathrm{Y} \mathrm{Y}+\theta_\mathrm{Z} \mathrm{Z})}\right)^{\otimes n}\ \textrm{when}\ G=\mathrm{SU}(2), \label{SMeq:SU2_representation}
\end{align}
where $\mathrm{X}$, $\mathrm{Y}$, and $\mathrm{Z}$ are the Pauli operators. 
We note that these representations $R$ can be written as the tensor product of representation $T^{\otimes n}$ with a representation $T$ on a single qubit. \\

Next, we explain the construction of random circuits. 
We consider the case when the allowed gate set is expressed as a finite number of connected compact unitary subgroups of $\mathcal{U}_{n, G, R}$. 
We denote each connected compact subgroup by $\mathcal{S}^\gamma$ and the set of all possible $\gamma$ by $\Gamma$. 
By using these gate sets $\{\mathcal{S}^\gamma\}_{\gamma\in\Gamma}$, we consider the distribution 
\begin{align}
    \zeta_{\{\mathcal{S}^\gamma\}_{\gamma\in\Gamma}}:=\sum_{\gamma\in\Gamma} p^\gamma \mu_{\mathcal{S}^\gamma} \label{SMeq:zeta_def}
\end{align}
with the Haar measure $\mu_{\mathcal{S}^\gamma}$ on $\mathcal{S}^\gamma$ and $p^\gamma>0$ satisfying $\sum_{\gamma\in\Gamma} p^\gamma=1$. 
We note that the exact values of $p^\gamma$'s do not affect our results as long as $p^\gamma\neq 0$, as we explain later.

We note that this setup includes the random circuits consisting of symmetric and local gates as follows: 
We label $n$ qudits as $1$, $2$, ..., and $n$, and for a subset of $\{1, 2, ..., n\}$, we denote by $\mathcal{U}_{n, G, R}^\gamma$ the set of all unitary subgroup of $\mathcal{U}_{n, G, R}$ acting nontrivially on the qudits labeled by $\gamma$. 
For example, when we have access to all symmetric nearest-neighbor unitary operators in a one-dimensional chain with the open boundary condition, $\Gamma$ is given by $\{\{1, 2\}, \{2, 3\}, ..., \{n-1, n\}\}$, which is illustrated in Fig.~\ref{fig:circuit}.

\begin{figure} 
    \centering
    \includegraphics[width=90mm]{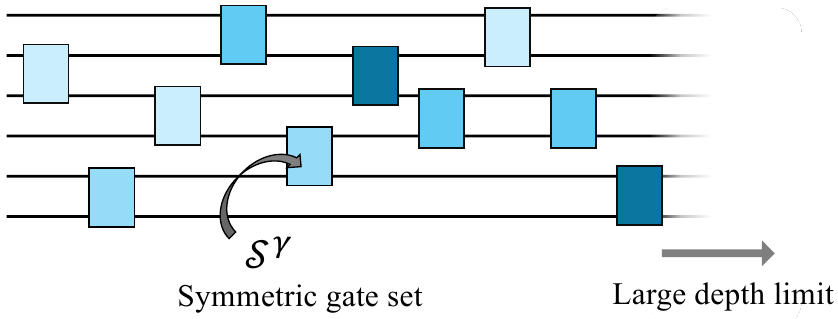}
    \caption{
    Example of symmetric random circuits. 
    We construct a random circuit by taking symmetric gate set $\mathcal{S}^\gamma$ with probability $p^\gamma$ and randomly drawing a unitary operator from the gate set. 
    This setup includes symmetric local random circuits when we consider the case when $\mathcal{S}^\gamma=\mathcal{U}_{n, G, R}^\gamma$ where $\gamma$ denotes the locality of the gate set.}
    \label{fig:circuit}
\end{figure}

In order to investigate randomness of the distribution $\zeta_{\{\mathcal{S}^\gamma\}_{\gamma\in\Gamma}}$ defined by Eq.~\eqref{SMeq:zeta_def}, we use \textit{asymptotic unitary designs} defined as follows:

\begin{definition} \label{SMdef:approx_sym_design}
    (Asymptotic symmetric unitary design.)
    Let $n, t\in\mathbb{N}$, $R$ be a unitary representation of a group $G$ on $\mathcal{H}$, and $\nu$ be a distribution on $\mathcal{U}_{n, G, R}$. 
    $\nu$ is an asymptotic $(G, R)$-symmetric unitary $t$-design if 
    \begin{align}
        \lim_{D\to\infty} (M_{t, \nu})^D=M_{t, \mu_{\mathcal{U}_{n, G, R}}}, \label{SMeq:approx_design_def}
    \end{align}
    with the normalized Haar measure $\mu_\mathcal{S}$ on a compact Lie subgroup $\mathcal{S}$ of $\mathcal{U}_{n, G, R}$ and the $t$th-order moment operator of $\nu$ defined by 
    \begin{align}
        M_{t, \nu}:=\int_{U\in\mathcal{U}_{n, G, R}} U^{\otimes t}\otimes U^{*\otimes t} d\nu(U). \label{SMeq:moment_operator_def}
    \end{align}
\end{definition}

This definition means that if a distribution $\nu$ is an asymptotic $(G, R)$-symmetric unitary $t$-design, the distribution of a circuit with infinite depth coincides with the Haar random distribution up to the $t$th moment. 
We use the term ``asymptotic unitary design'' because we only care about the asymptotic behavior of the distribution for deep circuits. 
We note that the distribution $\nu$ is an asymptotic unitary design if and only if for any $\epsilon>0$, there exists $D_0\in\mathbb{N}$ such that for any $D\geq D_0$, the $D$-fold convolution of $\nu$ is an $\epsilon$-approximate unitary design.  \\

In order to state the main theorem, we prepare the notion of semi-universality. 
If a gate set generates $\mathcal{U}_{n, G, R}$, it is called universal for $\mathcal{U}_{n, G, R}$. 
The semi-universality is a weaker version of the universality, defined as follows~\cite{kempe2001encoded}:

\begin{definition}
    (Semi-universality.)
    Let $n\in\mathbb{N}$, $R$ be a representation of a group $G$, and $\mathcal{X}$ be a subset $\mathcal{U}_{n, G, R}$. 
    $\mathcal{X}$ is semi-universal for $\mathcal{U}_{n, G, R}$ if 
    \begin{align}
        \Braket{\mathcal{X}}\cdot Z(\mathcal{U}_{n, G, R})=\mathcal{U}_{n, G, R}, \label{SMeq:semi_universality}
    \end{align}
    where $\braket{\mathcal{X}}$ is the group generated by the elements of $\mathcal{X}$, and $Z(\mathcal{U}_{n, G, R})$ is the center of $\mathcal{U}_{n, G, R}$, i.e., $Z(\mathcal{U}_{n, G, R}):=\{U\in\mathcal{U}_{n, G, R}\ |\ [U, V]=0\ \forall V\in\mathcal{U}_{n, G, R}\}$. 
\end{definition}

It is known in Refs.~\cite{marvian2023theory, marvian2024rotationally} that the $(G, R)$-symmetric $2$-local gate sets are semi-universal for $\mathcal{U}_{n, G, R}$ for $\mathbb{Z}_2$, $\mathrm{U}(1)$, and $\mathrm{SU}(2)$ symmetries given by Eqs.~\eqref{SMeq:Z2_representation}, \eqref{SMeq:U1_representation}, and \eqref{SMeq:SU2_representation} as long as $\Gamma$ is \textit{inseparable}.  We say that $\Gamma$ is inseparable in $\{1, 2, ..., n\}$ if there is no pair of nontrivial subsets $C_1$ and $C_2$ of $\{1, 2, ..., n\}$ that satisfy $C_1\cap C_2=\varnothing$, $C_1\cup C_2=\{1, 2, ..., n\}$, and $\gamma\subset C_1$ or $\gamma\subset C_2$ for all $\gamma\in\Gamma$. \\

In order to present the condition for semi-universality more directly, we introduce the decomposition of symmetric operators. 
Every unitary representation $R$ can be decomposed into irreducible representations, i.e., we can take an isomorphism 
\begin{align}
	\mathcal{H}\cong\bigoplus_{\lambda\in\Lambda} \mathbb{C}^{r_\lambda}\otimes\mathbb{C}^{m_\lambda} 
\end{align}
such that 
\begin{align}
	R(g)=\sum_{\lambda\in\Lambda} F_\lambda(R_\lambda(g)\otimes I)F_\lambda^\dag\ \forall g\in G, 
\end{align}
where $\Lambda$ is the set of all labels $\lambda$ for inequivalent irreducible representations appearing in $R$, $R_\lambda$'s are irreducible representations of $G$ on $\mathbb{C}^{r_\lambda}$, $m_\lambda$ is the multiplicity of the representation $R_\lambda$, and $F_\lambda$ is the isometry from $\mathbb{C}^{r_\lambda}\otimes\mathbb{C}^{m_\lambda}$ to $\mathcal{H}$. 
By using Schur's lemma, every $(G, R)$-symmetric operator $A$ can be written as 
\begin{align}
	A=\sum_{\lambda\in\Lambda} F_\lambda(I\otimes A_\lambda)F_\lambda^\dag \label{SMeq:sym_op_decomp}
\end{align} 
with some $A_\lambda$'s acting on $\mathbb{C}^{m_\lambda}$, which are uniquely determined for a $(G, R)$-symmetric operator $A\in\mathcal{L}_{n, G, R}$ on $n$ qudits. 
By using this decomposition, $Z(\mathcal{U}_{n, G, R})$ can be explicitly written as 
\begin{align}
	Z\left(\mathcal{U}_{n, G, R}\right) 
    =\left\{\sum_{\lambda\in\Lambda} F_\lambda(I\otimes e^{i\theta_\lambda}I)F_\lambda^\dag\ \middle|\ \theta_\lambda\in\mathbb{R}\ \forall \lambda\in\Lambda\right\}. \label{SMeq:rel_phase_group}
\end{align}

We note that the semi-universality of $\mathcal{X}$ can be equivalently expressed as 
\begin{align}
    \Braket{\mathcal{X}}
    \supset\left\{\sum_{\lambda\in\Lambda} F_\lambda(I\otimes U_\lambda)F_\lambda^\dag\ \middle|\ U_\lambda\in\mathrm{SU}(m_\lambda)\ \forall \lambda\in\Lambda\right\}. \label{SMeq:semi_universality_equivalence}
\end{align}
It is trivial to see that Eq.~\eqref{SMeq:semi_universality_equivalence} implies Eq.~\eqref{SMeq:semi_universality} by noting Eq.~\eqref{SMeq:rel_phase_group}. 
The proof of the converse is as follows: 
We suppose that $\mathcal{X}$ satisfies Eq.~\eqref{SMeq:semi_universality}. 
We take arbitrary $U\in\mathcal{U}_{n, G, R}$ in the form of $\sum_{\lambda\in\Lambda} F_\lambda(I\otimes U_\lambda)F_\lambda^\dag$ with some $U_\lambda\in\mathrm{SU}(m_\lambda)$. 
For each $\lambda\in\Lambda$, since $\mathrm{SU}(m_\lambda)$ is a simple Lie group, we can take $U'_\lambda, U''_\lambda\in\mathrm{SU}(m_\lambda)$ satisfying ${U'_\lambda}^{-1}{U''_\lambda}^{-1}U'_\lambda U''_\lambda=U_\lambda$. 
We define $U':=\sum_{\lambda\in\Lambda} F_\lambda(I\otimes U'_\lambda)F_\lambda^\dag$ and $U'':=\sum_{\lambda\in\Lambda} F_\lambda(I\otimes U''_\lambda)F_\lambda^\dag$. 
Since $U'$ and $U''$ satisfy $U', U''\in\mathcal{U}_{n, G, R}$, Eq.~\eqref{SMeq:semi_universality} implies that we can take $V', V''\in\braket{\mathcal{X}}$ and $W', W''\in Z(\mathcal{U}_{n, G, R})$ such that $U'=V'W'$ and $U''=V''W''$. 
Then, by noting that $W'$ and $W''$ commute with $V', W', V''$, and $W''$, we have $U={U'}^{-1}{U''}^{-1}U'U''={V'}^{-1}{V''}^{-1}V'V''\in\braket{\mathcal{X}}$.\\

\begin{figure}
    \centering 
    \includegraphics[width=120mm]{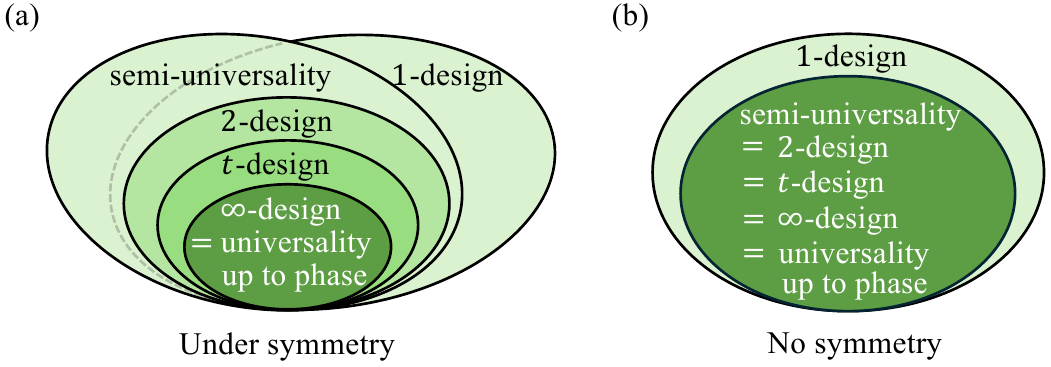}
    \caption{Hierarchy of asymptotic unitary designs of random circuits and its relation with the (semi-)universality of the gate sets consisting of the circuits. 
    (a) In the presence of symmetry, there is a rich structure of the classes of asymptotic unitary designs. 
    Semi-universal gate sets are necessary to construct random circuits with $2$-designs. 
    Our main result is to establish a method to characterize the maximum order of unitary designs of the distribution for symmetric random circuits composed of semi-universal gate sets. 
    (b) The relation between the universality and designs becomes rather trivial without symmetry.}
    \label{fig:Venn_diagram}
\end{figure}

In the following, we explain the relation between the (semi-)universality of the gate sets $\{\mathcal{S}^\gamma\}_{\gamma\in\Gamma}$ and asymptotic unitary designs of the distribution $\zeta_{\{\mathcal{S}^\gamma\}_{\gamma\in\Gamma}}$ generated by the gate set, which is shown in Fig.~\ref{fig:Venn_diagram}. 
First, if the distribution $\zeta_{\{\mathcal{S}^\gamma\}_{\gamma\in\Gamma}}$ defined by Eq.~\eqref{SMeq:zeta_def} is an asymptotic $(G, R)$-symmetric unitary $2$-design, then the gate set $\bigcup_{\gamma\in\Gamma} \mathcal{S}^\gamma$ is semi-universal for $\mathcal{U}_{n, G, R}$. 
This can be proven by the combination of Theorem~16 in Ref.~\cite{zeier2015squares} and Lemma~\ref{SMlem:design_commutant}. 
When $\zeta_{\{\mathcal{S}^\gamma\}_{\gamma\in\Gamma}}$ is an asymptotic $(G, R)$-symmetric unitary $2$-design, the commutant of $\{U^{\otimes 2}\ |\ \exists \gamma\in\Gamma \textrm{ s.t. } U\in\mathcal{S}^\gamma\}$ coincides with that of $\{U^{\otimes 2}\ |\ U\in\mathcal{U}_{n, G, R}\}$ by Lemma~\ref{SMlem:design_commutant}. 
This implies that the commutant of $\{A\otimes I+I\otimes A\ |\ \exists \gamma\in\Gamma \textrm{ s.t. } A\in\mathfrak{s}^\gamma\}$ coincides with that of $\{A\otimes I+I\otimes A\ |\ A\in\mathfrak{u}_{n, G, R}\}$, where $\mathfrak{s}^\gamma$ and $\mathfrak{u}_{n, G, R}$ are the Lie algebras of $\mathcal{S}^\gamma$ and $\mathcal{U}_{n, G, R}$, respectively. 
Then, Theorem~16 in Ref.~\cite{zeier2015squares} implies that $\bigcup_{\gamma\in\Gamma} \mathfrak{s}^\gamma$ generates $\mathfrak{u}_{n, G, R}$ up to $Z(\mathfrak{u}_{n, G, R})$ in the sense of Lie algebra, which implies the semi-universality of $\bigcup_{\gamma\in\Gamma} \mathcal{S}^\gamma$ for $\mathcal{U}_{n, G, R}$. 
We note that this statement holds only for gate sets consisting of connected compact groups, and not for discrete gate sets such as the Clifford group.

Second, if the gate set $\bigcup_{\gamma\in\Gamma} \mathcal{S}^\gamma$ is universal for $\mathcal{U}_{n, G, R}$ up to the global phase, then the distribution $\zeta_{\{\mathcal{S}^\gamma\}_{\gamma\in\Gamma}}$ is an asymptotic $(G, R)$-symmetric unitary $t$-design for all $t\in\mathbb{N}$, which directly follows from the same argument in the non-symmetric case~\cite{brandao2016local}. 
Moreover, the converse is also true, as we explain below Theorem~\ref{SMthm:general}. 
Since the semi-universality reduces to the universality up to the global phase in the non-symmetric case, i.e., $G=\{I\}$, every class of distributions forming unitary $t$-designs coincides to the class of distributions forming unitary $\infty$-designs. 
Thus, the inclusion relation in symmetric cases (Fig.~\ref{fig:Venn_diagram} (a)) becomes much simpler in the non-symmetric case (Fig.~\ref{fig:Venn_diagram} (b)).

Finally, we show that there are no nontrivial inclusion relations other than those stated above. 
Concretely, the semi-universality of the gate set $\bigcup_{\gamma\in\Gamma} \mathcal{S}^\gamma$ does not imply that the distribution $\zeta_{\{\mathcal{S}^\gamma\}_{\gamma\in\Gamma}}$ is an asymptotic $(G, R)$-symmetric unitary $2$-design, and not even a $1$-design. 
For example, when $n=1$, $G=\mathbb{Z}_2$ and $R(g)=\mathrm{Z}^g$ for $g\in\mathbb{Z}_2=\{0, 1\}$, the gate set consisting only of the identity is semi-universal for $\mathcal{U}_{n, G, R}$, but is not an asymptotic unitary $1$-design for $\mathcal{U}_{n, G, R}$. 
We also note that the distribution $\zeta_{\{\mathcal{S}^\gamma\}_{\gamma\in\Gamma}}$ being an asymptotic $(G, R)$-symmetric unitary $1$-design does not imply the semi-universality of the gate set $\bigcup_{\gamma\in\Gamma} \mathcal{S}^\gamma$ for $\mathcal{U}_{n, G, R}$, 
and that the combination of these two conditions does not imply that the distribution $\zeta_{\{\mathcal{S}^\gamma\}_{\gamma\in\Gamma}}$ is an asymptotic $(G, R)$-symmetric unitary $2$-design. 
For the proofs of the two statements above, we set $n=2$, $G=\mathbb{Z}_2$, and $R(g)=(\mathrm{Z}^g)^{\otimes 2}$ for $g\in\{0, 1\}$, and define the following four gate sets: 
\begin{align}
    &\mathcal{S}^1:=\left\{e^{i\theta \mathrm{X}\otimes \mathrm{X}}\right\}_{\theta\in\mathbb{R}}, \\
    &\mathcal{S}^2:=\left\{e^{i\theta \mathrm{Z}\otimes \mathrm{Z}}\right\}_{\theta\in\mathbb{R}}, \\
    &\mathcal{S}^3:=\left\{e^{i\theta \mathrm{Z}\otimes \mathrm{I}}\right\}_{\theta\in\mathbb{R}}, \\
    &\mathcal{S}^4:=\left\{e^{i\theta \mathrm{I}\otimes \mathrm{Z}}\right\}_{\theta\in\mathbb{R}},  
\end{align}
where $\mathrm{I}$ is the identity operator on a single qubit. 
When $\Gamma=\{1, 2, 3\}$, the distribution $\zeta_{\{\mathcal{S}^\gamma\}_{\gamma\in\Gamma}}$ is an asymptotic $(G, R)$-symmetric unitary $1$-design by Lemma~\ref{SMlem:design_commutant}, but the gate set $\bigcup_{\gamma\in\Gamma} \mathcal{S}^\gamma$ is not semi-universal for $\mathcal{U}_{n, G, R}$, which can be confirmed by noting that $\braket{\bigcup_{\gamma\in\Gamma} \mathcal{S}^\gamma}=\{\sum_{\lambda\in\{0, 1\}} F_\lambda e^{i(-1)^\lambda\theta}UF_\lambda^\dag\ |\ \theta\in\mathbb{R}, U\in\mathrm{SU}(2)\}$, where $F_\lambda$ is defined by $F_0\ket{j}:=\ket{j}\otimes\ket{j}$ and $F_1\ket{j}:=\ket{j}\otimes\ket{1-j}$ for $j\in\{0, 1\}$. 
When $\Gamma=\{1, 3, 4\}$, the distribution $\zeta_{\{\mathcal{S}^\gamma\}_{\gamma\in\Gamma}}$ is an asymptotic $(G, R)$-symmetric unitary $1$-design by Lemma~\ref{SMlem:design_commutant}, and $\bigcup_{\gamma\in\Gamma} \mathcal{S}^\gamma$ is semi-universal for $\mathcal{U}_{n, G, R}$, but $\zeta_{\{\mathcal{S}^\gamma\}_{\gamma\in\Gamma}}$ is not an asymptotic unitary $2$-design for $\mathcal{U}_{n, G, R}$ by Theorem~\ref{SMthm:general}. 
For the proof of the semi-universality, it is sufficient to confirm that the Lie algebras of $\mathcal{S}^\gamma$'s generate the Lie algebra $\mathfrak{u}_{n, G, R}$ of $\mathcal{U}_{n, G, R}$ up to $Z(\mathfrak{u}_{n, G, R})$ in the sense of Lie algebra.

\section{Main Results} \label{SMsec:main_results}

First, we present the general result about the maximal order of asymptotic unitary designs, which is applicable to general symmetries. 
The following theorem corresponds to Theorem~2 of Ref.~\cite{mitsuhashi2025unitary}.

\begin{theorem} \label{SMthm:general} 
    (General result.) 
    Let $n, t\in\mathbb{N}$, $R$ be a unitary representation of a group $G$ on the Hilbert space $\mathcal{H}$ of $n$ qudits, $\{\mathcal{S}^\gamma\}_{\gamma\in\Gamma}$ be the set of a finite number of connected compact subgroups of $\mathcal{U}_{n, G, R}$, and $\bigcup_{\gamma\in\Gamma} \mathcal{S}^\gamma$ be semi-universal. 
    Then, the distribution $\zeta_{\{\mathcal{S}^\gamma\}_{\gamma\in\Gamma}}$ defined by Eq.~\eqref{SMeq:zeta_def} is an asymptotic $(G, R)$-symmetric unitary $t$-design if and only if there do not exist nontrivial integer solutions $\bm{x}=(x_\lambda)_{\lambda\in\Lambda}\in\mathbb{Z}^\Lambda$ satisfying 
    \begin{align}
        &\sum_{\lambda\in\Lambda} m_\lambda |x_\lambda|\leq 2t, \label{SMeq:SMthm:general_cond1}\\
        &\sum_{\lambda\in\Lambda} m_\lambda x_\lambda=0, \label{SMeq:SMthm:general_cond2}\\
        &\sum_{\lambda\in\Lambda} v_\lambda x_\lambda=0\ \forall \bm{v}\in\mathcal{V}, \label{SMeq:SMthm:general_cond3}
    \end{align}
    where $\mathcal{V}:=\mathrm{span}_\mathbb{R}(\{\bm{f}(A)\ |\ \exists\gamma\in\Gamma \textrm{ s.t. } A\in\mathfrak{s}^\gamma\})$, $\mathfrak{s}^\gamma$ is the Lie algebra of $\mathcal{S}^\gamma$, and 
    \begin{align}
        \bm{f}(A)
        =(f_\lambda(A))_{\lambda\in\Lambda}
        :=(\mathrm{tr}(A_\lambda))_{\lambda\in\Lambda} \label{SMeq:trace_vec_def}
    \end{align}
    with $A_\lambda$ determined from $A$ by Eq.~\eqref{SMeq:sym_op_decomp}. 
    Especially when $\mathcal{S}^\gamma=\mathcal{U}_{n, G, R}^\gamma$ and $R=T^{\otimes n}$ with some representation $T$ of $G$ on a single qudit, Eqs.~\eqref{SMeq:SMthm:general_cond2} and \eqref{SMeq:SMthm:general_cond3} can be equivalently written as 
    \begin{align}
        \sum_{\lambda\in\Lambda} c_\lambda x_\lambda=0\ \forall \bm{c}\in\mathcal{C}, \label{SMeq:SMthm:general_cond4}
    \end{align}
    where $\mathcal{C}$ is defined by $\mathcal{C}:=\{\bm{f}(A\otimes\mathrm{I}^{\otimes n-k})\ |\ A\in\mathcal{L}_{k, G, T^{\otimes k}}\}$, and $k:=\max_{\gamma\in\Gamma} \#\gamma$. 
\end{theorem}

\vspace{5mm}

We give three remarks about this theorem. 
First, finding the condition on $t$ for the nonexistence of nontrivial integer solutions of Eqs~\eqref{SMeq:SMthm:general_cond1}, \eqref{SMeq:SMthm:general_cond2}, and \eqref{SMeq:SMthm:general_cond3} is equivalent to a simple integer optimization. 
In fact, the condition is explicitly expressed as 
\begin{align}
    t<\min_{\bm{x}\in(\widetilde{\mathcal{V}}^\perp\cap\mathbb{Z}^\Lambda)\backslash\{\bm{0}\}} \braket{\bm{m}, \bm{x}^+}, \label{SMeq:order_bound_general}
\end{align}
where 
\begin{align}
    &\bm{m}=(m_\lambda)_{\lambda\in\Lambda}, \\
    &\bm{x}^+=(x_\lambda^+)_{\lambda\in\Lambda}:=((|x_\lambda|+x_\lambda)/2)_{\lambda\in\Lambda}, \\
    &\widetilde{\mathcal{V}}:=\mathrm{span}_\mathbb{R}(\{\bm{m}\})+\mathcal{V}, \label{SMeq:relative_phase_space_extended}
\end{align}
and we use the standard inner product $\braket{\bm{a}, \bm{b}}:=\sum_{\lambda\in\Lambda} a_\lambda^* b_\lambda$ for $\bm{a}, \bm{b}\in\mathbb{C}^\Lambda$. 
This can be understood by noting that Eq.~\eqref{SMeq:SMthm:general_cond1}, \eqref{SMeq:SMthm:general_cond2} and \eqref{SMeq:SMthm:general_cond3} are equivalent to $\braket{\bm{m}, \bm{x}^+}\leq t$ and $\bm{x}\in\widetilde{\mathcal{V}}^\perp$. 
When $\mathcal{S}^\gamma=\mathcal{U}_{n, G, R}^\gamma$ and $R=T^{\otimes n}$, by taking a basis of $\mathcal{C}$, Eq.~\eqref{SMeq:SMthm:general_cond4} can be written as the set of $\mathrm{dim}(\mathcal{C})$ equations. 
We note that $\mathrm{dim}(\mathcal{C})$ is upper bounded by $\mathrm{dim}(\mathcal{L}_{k, G, T^{\otimes k}})$, which is independent of the qudit count $n$. 
Similarly to Eq.~\eqref{SMeq:order_bound_general}, the condition on $t$ can be written as 
\begin{align}
    t<\min_{\bm{x}\in(\mathcal{C}^\perp\cap\mathbb{Z}^\Lambda)\backslash\{\bm{0}\}} \braket{\bm{m}, \bm{x}^+}. 
\end{align}

Next, by using this theorem, we can confirm that $\zeta_{\{\mathcal{S}^\gamma\}_{\gamma\in\Gamma}}$ is an asymptotic unitary $t$-design for all $t\in\mathbb{N}$ if and only if $\bigcup_{\gamma\in\Gamma} \mathcal{S}^\gamma$ is universal for $\mathcal{U}_{n, G, R}$ up to the global phase. 
When $\bigcup_{\gamma\in\Gamma} \mathcal{S}^\gamma$ is universal for $\mathcal{U}_{n, G, R}$ up to the global phase, by Lemma~\ref{SMlem:design_universality}, Eq.~\eqref{SMeq:SMthm:general_cond2} and \eqref{SMeq:SMthm:general_cond3} have no nontrivial integer solution $\bm{x}$. 
Theorem~\ref{SMthm:general} thus implies that $\zeta_{\{\mathcal{S}^\gamma\}_{\gamma\in\Gamma}}$ is an asymptotic unitary $t$-design for all $t\in\mathbb{N}$. 
On the other hand, when $\bigcup_{\gamma\in\Gamma} \mathcal{S}^\gamma$ is not universal for $\mathcal{U}_{n, G, R}$ up to the global phase, we can take $\bm{d}\in (\widetilde{\mathcal{V}}^\perp\cap\mathbb{Z}^\Lambda)\backslash\{\bm{0}\}$ by Lemma~\ref{SMlem:design_universality}, and define $t_0:=\braket{\bm{m}, \bm{x}^+}$. 
Since $\bm{x}=\bm{d}$ is a nontrivial integer solution of Eqs.~\eqref{SMeq:SMthm:general_cond2} and \eqref{SMeq:SMthm:general_cond3}, any achievable order $t$ is smaller than $t_0$.

Finally, we can compute the tight upper bound on the achievable order $t$ by enumeration. 
When $\widetilde{\mathcal{V}}=\mathbb{R}^\Lambda$, there does not exist an upper bound on $t$. 
In the following, we consider the case of $\widetilde{\mathcal{V}}\neq\mathbb{R}^\Lambda$. 
In this case, by using the method above, we can take an upper bound $t_0$, which is not necessarily tight. 
Then, Eq.~\eqref{SMeq:order_bound_general} is equivalent to 
\begin{align}
    t<\min\left\{t_0, \min_{\bm{x}\in(\widetilde{\mathcal{V}}^\perp\cap\mathbb{Z}^\Lambda\cap\mathcal{F})\backslash\{\bm{0}\}}\braket{\bm{m}, \bm{x}^+}\right\} \label{SMeq:search_space_restriction}
\end{align}
with a bounded region $\mathcal{F}:=\{\bm{x}\in\mathbb{R}^\Lambda\ |\ |x_\lambda|<t_0/m_\lambda\ \forall\lambda\in\Lambda\}$. 
For the proof of Eq.~\eqref{SMeq:search_space_restriction}, it is sufficient to show that $\min_{\bm{x}\in\mathbb{Z}^\Lambda\backslash\mathcal{F}}\braket{\bm{m}, \bm{x}^+}\geq t_0$ by noting that $[(\widetilde{\mathcal{V}}^\perp\cap\mathbb{Z}^\Lambda\cap\mathcal{F})\backslash\{\bm{0}\}]\cup (\mathbb{Z}^\Lambda\backslash\mathcal{F})\supset(\widetilde{\mathcal{V}}^\perp\cap\mathbb{Z}^\Lambda)\backslash\{\bm{0}\}$. 
For arbitrary $\bm{x}\in\mathbb{Z}^\Lambda\backslash\mathcal{F}$, we can take some $\lambda\in\Lambda$ such that $|x_\lambda|\geq t_0/m_\lambda$. 
We can suppose that $x_\lambda>0$ without loss of generality, since $\bm{x}\in\mathbb{Z}^\Lambda\backslash\mathcal{F}$ implies $-\bm{x}\in\mathbb{Z}^\Lambda\backslash\mathcal{F}$. 
Then, we get $\braket{\bm{m}, \bm{x}^+}\geq m_\lambda x_\lambda\geq t_0$. \\

In the following theorems, we consider the cases when $G=\mathbb{Z}_2$, $\mathrm{U}(1)$ and $\mathrm{SU}(2)$, and $\mathcal{S}^\gamma$ is given by the set $\mathcal{U}_{n, G, R}^\gamma$ of all $(G, R)$-symmetric unitary operators acting on the qubits represented by $\gamma$, and the locality of the gate set satisfies $\max_{\gamma\in\Gamma} \#\gamma=k$.

First, we present the result for the $\mathbb{Z}_2$ symmetry, which corresponds to the first result of Theorem~1 of Ref.~\cite{mitsuhashi2025unitary}.

\begin{theorem} \label{SMthm:Z2}
    (Result for general locality $k$ and general qubit count $n$ under the $\mathbb{Z}_2$ symmetry.) 
	Let $n, k, t\in\mathbb{N}$ satisfy $k\geq 2$ and $n\geq k+1$, and $R$ be a unitary representation of $G=\mathbb{Z}_2$ on $n$ qubits defined by Eq.~\eqref{SMeq:Z2_representation}. 
    Then, the distribution of the $(G, R)$-symmetric $k$-local random circuit is an asymptotic $(G, R)$-symmetric unitary $t$-design if and only if $t<2^{n-1}$. 
\end{theorem}

We note that the condition of $t$ does not depend on the locality $k$, which is a feature different from the cases of $\mathrm{U}(1)$ and $\mathrm{SU}(2)$. 
We describe the proof of this theorem in Appendix~\ref{SMsubsec:Z2}. \\

Next, we present the result for the $\mathrm{U}(1)$ symmetry. 
For general locality $k$, we can give the maximal order of asymptotic unitary designs for sufficiently large $n$ in the following theorem, which corresponds to the second part of Theorem~1 of Ref.~\cite{mitsuhashi2025unitary}.

\begin{theorem} \label{SMthm:U1_general_locality}
    (Result for general locality $k$ and sufficiently large qubit count $n$ under the $\mathrm{U}(1)$ symmetry.) 
    Let $n, k, t\in\mathbb{N}$ satisfy $k\geq 2$ and $n\geq 2^k$, and $R$ be a unitary representation of $G=\mathrm{U}(1)$ on $n$ qubits defined by Eq.~\eqref{SMeq:U1_representation}. 
    Then, the distribution of the $(G, R)$-symmetric $k$-local random circuit is an asymptotic $(G, R)$-symmetric unitary $t$-design if and only if 
    \begin{align}
        t<\frac{2^{\lfloor k/2\rfloor}}{\lceil k/2\rceil!}\prod_{\alpha=1}^{\lceil k/2\rceil} (n-k+2\alpha-1). \label{SMeq:SMthm:U1_general_locality1}
    \end{align}
\end{theorem}

We note that the condition $n\geq 2^k$ is needed only for the proof of the ``if'' part, i.e., for any $n\geq k+1$, we can show that the distribution is not an asymptotic $(G, R)$-symmetric unitary $t$-design if $t$ does not satisfy the condition above. 
We present the proof in Appendix~\ref{SMsubsec:U1}.

While the theorem above specifies the maximal order of unitary designs of the $\mathrm{U}(1)$-symmetric local random circuits for sufficiently many qubits, it does not guarantee that the bound is the same in the case of few qubits. 
As a result complementary to Theorem~\ref{SMthm:U1_general_locality}, we show the result for small locality $k=2$, $3$, and $4$ in the following theorem.

\begin{theorem} \label{SMthm:U1_small_locality}
    (Result for small locality $k$ and general qubit count $n$ under the $\mathrm{U}(1)$ symmetry.) 
    Let $n, t\in\mathbb{N}$, $k=2$, $3$ or $4$, $n\geq k+1$, and $R$ be a unitary representation of $G=\mathrm{U}(1)$ on $n$ qubits defined by Eq.~\eqref{SMeq:U1_representation}. 
    Then, the distribution of the $(G, R)$-symmetric $k$-local random circuit is an asymptotic unitary $t$-design if and only if 
    \begin{align}
        \begin{cases}
            t<2(n-1) & (\textrm{when } k=2), \\
            t<n(n-2) & (\textrm{when } k=3), \\
            t<2(n-1)(n-3) & (\textrm{when } k=4). 
        \end{cases} \label{SMeq:U1_small_locality_order}
    \end{align}
\end{theorem}

This theorem means that for the locality $k=2$, $3$, and $4$, even in the case of few qubits, the maximal order of unitary designs is given by the same function of the number of qubits as in the many-qubit case. 
We note, however, that this does not hold for general locality. 
For example, when $n=7$ and $k=5$, we can confirm that the condition for $t$ is given by $t<64$, not by $t<70$. 
In the proof of this theorem, we first check the range of $n$ that satisfies the assumption in Lemma~\ref{SMlem:U1_nontrivial_nonexistence_general}. 
For other $n$, we check the condition for the existence of nontrivial integer solutions for the equations in Lemma~\ref{SMlem:U1_cond} one by one. 
We present the details in Appendix~\ref{SMsubsec:U1}.

Finally, we show the result for the case of $\mathrm{SU}(2)$ symmetry. 
We present the result for general locality $k$ for sufficiently large $n$ in the following theorem, which corresponds to the third part of Theorem~1 of Ref.~\cite{mitsuhashi2025unitary}.

\begin{theorem} \label{SMthm:SU2_general_locality}
    (Result for general locality $k$ and sufficiently large qubit count $n$ under the $\mathrm{SU}(2)$ symmetry.) 
    Let $n, k, t\in\mathbb{N}$ satisfy $k\geq 2$ and $n\geq 2^{2(\lfloor k/2\rfloor+1)}$, and $R$ be a unitary representation of $G=\mathrm{SU}(2)$ on $n$ qubits defined by Eq.~\eqref{SMeq:SU2_representation}. 
    Then, the distribution of the $(G, R)$-symmetric $k$-local random circuit is an asymptotic $(G, R)$-symmetric unitary $t$-design if and only if 
    \begin{align}
        t<\frac{2^{\lfloor k/2\rfloor}}{(\lfloor k/2\rfloor+1)!}\prod_{\alpha=1}^{\lfloor k/2\rfloor+1} (n-2\alpha+1). \label{SMeq:SMthm:SU2_general_locality1}
    \end{align}
\end{theorem}

Similarly to Theorem~\ref{SMthm:U1_general_locality}, the condition $n\geq 2^{2(\lfloor k/2\rfloor+1)}$ is needed only for the proof of the ``if'' part, and we can show that for any $k\geq 2(\lfloor k/2\rfloor+1)$, the distribution is not an asymptotic unitary $t$-design if $t$ does not satisfy the condition above without the assumption. 
We present the proof in Appendix~\ref{SMsubsec:SU2}.

While Theorem~\ref{SMthm:SU2_general_locality} gives the result for general locality $k$ and sufficiently large $n$, it does not hold for small $n$. 
As a complementary result, we focus on the small locality $k=2$, $3$, and $4$, and give the result for small $n$ in the following theorem.

\begin{theorem} \label{SMthm:SU2_small_locality}
	(Result for small locality $k$ and general qubit count $n$ under the $\mathrm{SU}(2)$ symmetry.) 
    Let $n, t\in\mathbb{N}$, $k=2$, $3$, or $4$, $n\geq k+1$, and $R$ be a unitary representation of $G=\mathrm{SU}(2)$ on $n$ qubits defined by Eq.~\eqref{SMeq:SU2_representation}. 
    Then, the distribution of the $(G, R)$-symmetric $k$-local random circuits is an asymptotic $(G, R)$-symmetric unitary $t$-designs if and only if 
    \begin{itemize}
    \item
    when $k=2$, 
    \begin{align}
		\begin{cases}
			t<\infty & \textrm{when}\ n=3, \\
			t<10 & \textrm{when}\ n=6, \\
			t<20 & \textrm{when}\ n=7, 8, \\
            t<(n-1)(n-3) & \textrm{when}\ n=4, 5 \textrm{ or } n\geq 9. \\
		\end{cases}
    \end{align}
    \item
    when $k=3$, 
    \begin{align}
		\begin{cases}
			t<10 & \textrm{when}\ n=6, \\
			t<20 & \textrm{when}\ n=7, 8, \\
            t<(n-1)(n-3) & \textrm{when}\ n=4, 5 \textrm{ or } n\geq 9. \\
		\end{cases}
    \end{align}
    \item 
    when $k=4$, 
    \begin{align}
        \begin{cases}
            t<\infty & \textrm{when }n=5, \\
            t<35 & \textrm{when }n=8, \\
            t<90 & \textrm{when }n=9, \\
            t<96 & \textrm{when }n=10, \\
            t<192 & \textrm{when }n=11, \\
            t<330 & \textrm{when }n=12, \\
            t<\displaystyle\frac{2}{3}(n-1)(n-3)(n-5) & \textrm{when }n=6, 7\textrm{ or } n\geq 13. 
        \end{cases}
    \end{align}
    \end{itemize}
\end{theorem}

In the proof of this theorem, we check the range of $n$ that satisfies the assumption in Lemma~\ref{SMlem:SU2_nontrivial_nonexistence_general}. 
For $n$ that does not satisfy the assumptions, we check the equations one by one, which we present in  Appendix~\ref{SMsubsec:SU2}.

\section{Proof of the general theorem (Theorem~\ref{SMthm:general})} \label{SMsec:proof_general}

In this section, we present the proof of Theorem~\ref{SMthm:general}. 
This proof consists of three parts. 
First, in Lemma~\ref{SMlem:design_commutant}, we rewrite the condition for forming unitary designs in terms of commutants. 
Next, in Lemma~\ref{SMlem:trivial_to_comm_eq}, we prove that the condition for the commutants is satisfied when Eqs.~\eqref{SMeq:SMthm:general_cond1}, \eqref{SMeq:SMthm:general_cond2}, and \eqref{SMeq:SMthm:general_cond3} have no nontrivial integer solution. 
Finally, in Lemma~\ref{SMlem:nontrivial_to_comm_neq}, we prove the converse part, i.e., we prove that the condition for the commutants is not satisfied  when Eqs.~\eqref{SMeq:SMthm:general_cond1}, \eqref{SMeq:SMthm:general_cond2}, and \eqref{SMeq:SMthm:general_cond3} have a nontrivial integer solution.

First, we show that the necessary and sufficient condition for forming unitary $t$-designs can be described as a property of commutants of $t$-fold operators. 
This is a standard technique to deal with unitary designs.

\begin{lemma} \label{SMlem:design_commutant}
	Let $t, n\in\mathbb{N}$, $R$ be a unitary representation of a group $G$ on $\mathcal{H}$, and $\{\mathcal{S}^\gamma\}_{\gamma\in\Gamma}$ be a finite set of connected compact Lie subgroups of $\mathcal{U}_{n, G, R}$. 
	Then, $\zeta_{\{\mathcal{S}^\gamma\}_{\gamma\in\Gamma}}$ is an asymptotic $(G, R)$-symmetric unitary $t$-design if and only if 
    \begin{align}
        \mathrm{Comm}\left(\Omega_t\left(\bigcup_{\gamma\in\Gamma} \mathcal{S}^\gamma\right)\right)=\mathrm{Comm}(\Omega_t(\mathcal{U}_{n, G, R})), \label{SMeq:SMlem:design_commutant1_2}
    \end{align}
	where $\mathrm{Comm}(\mathcal{X})$ is the set of operators commuting with all operators in $\mathcal{X}$, and 
	\begin{align}
		\Omega_t(U):=U^{\otimes t}. \label{SMeq:Omega_def}
	\end{align}     
\end{lemma}

\begin{proof}
    By the definition of the moment operator and the definition of $\zeta_{\{\mathcal{S}^\gamma\}_{\gamma\in\Gamma}}$, we have 
    \begin{align}
        M_{t, \zeta_{\{\mathcal{S}^\gamma\}_{\gamma\in\Gamma}}}
        =\sum_{\gamma\in\Gamma} p^\gamma M_{t, \mathcal{S}^\gamma}. \label{SMeq:SMlem:design_commutant1_3}
    \end{align}
    By Lemma~\ref{SMlem:moment_op_projection}, $M_{t, \mathcal{S}^\gamma}$ is Hermitian and positive for all $\gamma\in\Gamma$. 
    Thus $M_{t, \zeta_{\{\mathcal{S}^\gamma\}_{\gamma\in\Gamma}}}$ is also Hermitian and positive, and has the following spectral decomposition: 
    \begin{align}
        M_{t, \zeta_{\{\mathcal{S}^\gamma\}_{\gamma\in\Gamma}}} 
        =\sum_{h\in H} h\Pi_{\{\ket{\Psi}\in\mathcal{H}^{\otimes 2t}\ |\ M_{t, \zeta_{\{\mathcal{S}^\gamma\}_{\gamma\in\Gamma}}}\ket{\Psi}=h\ket{\Psi}\}}, 
    \end{align}
    where $H$ is the set of eigenvalues of $M_{t, \zeta_{\{\mathcal{S}^\gamma\}_{\gamma\in\Gamma}}}$, and $\Pi_\mathcal{K}$ is the projection operator onto $\mathcal{K}$. 
    Since $M_{t, \zeta_{\{\mathcal{S}^\gamma\}_{\gamma\in\Gamma}}}$ is a convex combination of projections, we have $H\subset [0, 1]$. 
    Then, we have 
    \begin{align}
        \lim_{D\to\infty} \left(M_{t, \zeta_{\{\mathcal{S}^\gamma\}_{\gamma\in\Gamma}}}\right)^D
        =\lim_{D\to\infty} \sum_{h\in H} h^D \Pi_{\{\ket{\Psi}\in\mathcal{H}^{\otimes 2t}\ |\ M_{t, \zeta_{\{\mathcal{S}^\gamma\}_{\gamma\in\Gamma}}}\ket{\Psi}=h\ket{\Psi}\}} 
        =\Pi_{\{\ket{\Psi}\in\mathcal{H}^{\otimes 2t}\ |\ M_{t, \zeta_{\{\mathcal{S}^\gamma\}_{\gamma\in\Gamma}}}\ket{\Psi}=\ket{\Psi}\}}. \label{SMeq:SMlem:design_commutant1_4}
    \end{align}
    We are going to show that 
    \begin{align}
        \{\ket{\Psi}\in\mathcal{H}^{\otimes 2t}\ |\ M_{t, \zeta_{\{\mathcal{S}^\gamma\}_{\gamma\in\Gamma}}}\ket{\Psi}=\ket{\Psi}\}=\bigcap_{\gamma\in\Gamma} E(\mathrm{Comm}(\Omega_t(\mathcal{S}^\gamma))), \label{SMeq:SMlem:design_commutant1_5}
    \end{align}
    where $E: \mathcal{L}(\mathcal{H}^{\otimes t})\to\mathcal{H}^{\otimes 2t}$ is defined by 
    \begin{align}
        &E(K):=(K\otimes I)\ket{\eta}\ \forall K\in\mathcal{L}(\mathcal{H}^{\otimes t}) \label{SMeq:E_def}
    \end{align}
    with 
    \begin{align}
        \ket{\eta}:=\frac{1}{\sqrt{d^{tn}}}\sum_{j=1}^{d^{tn}} \ket{j}\otimes\ket{j} \label{SMeq:eta_def}
    \end{align}
    and an orthonormal basis $\{\ket{j}\}_{j=1}^{d^{tn}}$ of $\mathcal{H}^{\otimes t}$. 
    For the proof of the inclusion relation $\{\ket{\Psi}\in\mathcal{H}^{\otimes 2t}\ |\ M_{t, \zeta_{\{\mathcal{S}^\gamma\}_{\gamma\in\Gamma}}}\ket{\Psi}=\ket{\Psi}\}\supset\bigcap_{\gamma\in\Gamma} E(\mathrm{Comm}(\Omega_t(\mathcal{S}^\gamma)))$, we take arbitrary $\ket{\Psi}\in\bigcap_{\gamma\in\Gamma} E(\mathrm{Comm}(\Omega_t(\mathcal{S}^\gamma)))$. 
    By Lemma~\ref{SMlem:moment_op_projection}, we have $M_{t, \mu_{\mathcal{S}^\gamma}}\ket{\Psi}=\ket{\Psi}$ for all $\gamma\in\Gamma$. 
    By Eq.~\eqref{SMeq:SMlem:design_commutant1_3}, we get $M_{t, \zeta_{\{\mathcal{S}^\gamma\}_{\gamma\in\Gamma}}}\ket{\Psi}=\ket{\Psi}$. 
    For the proof of the inverse inclusion relation, we take arbitrary $\ket{\Psi}\in\mathcal{H}^{\otimes 2t}$ satisfying $M_{t, \zeta_{\{\mathcal{S}^\gamma\}_{\gamma\in\Gamma}}}\ket{\Psi}=\ket{\Psi}$. 
    Then by Eq.~\eqref{SMeq:SMlem:design_commutant1_3}, we have 
    \begin{align}
        \sum_{\gamma\in\Gamma} p^\gamma\bra{\Psi}M_{t, \zeta_{\{\mathcal{S}^\gamma\}_{\gamma\in\Gamma}}}\ket{\Psi} 
        =\bra{\Psi}M_{t, \zeta_{\{\mathcal{S}^\gamma\}_{\gamma\in\Gamma}}}\ket{\Psi} 
        =1. \label{SMeq:SMlem:design_commutant1_6}
    \end{align}
    By noting that $\sum_{\gamma\in\Gamma} p^\gamma=1$, $p^\gamma>0$, and $\bra{\Psi}M_{t, \zeta_{\{\mathcal{S}^\gamma\}_{\gamma\in\Gamma}}}\ket{\Psi}\in [0, 1]$ for all $\gamma\in\Gamma$, Eq.~\eqref{SMeq:SMlem:design_commutant1_6} implies that $\bra{\Psi}M_{t, \zeta_{\{\mathcal{S}^\gamma\}_{\gamma\in\Gamma}}}\ket{\Psi}=1$ for all $\gamma\in\Gamma$, which implies that $\ket{\Psi}\in E(\mathrm{Comm}(\Omega_t(\mathcal{S}^\gamma)))$ by Lemma~\ref{SMlem:moment_op_projection}. 
    Since this holds for $\gamma\in\Gamma$, we get $\ket{\Psi}\in\bigcap_{\gamma\in\Gamma} E(\mathrm{Comm}(\Omega_t(\mathcal{S}^\gamma)))$. 
    Thus we have shown Eq.~\eqref{SMeq:SMlem:design_commutant1_5}. 
    Since $E$ is bijective, we have 
    \begin{align}
        \bigcap_{\gamma\in\Gamma} E(\mathrm{Comm}(\Omega_t(\mathcal{S}^\gamma))) 
        =E\left(\bigcap_{\gamma\in\Gamma} \mathrm{Comm}(\mathcal{S}^\gamma)\right) 
        =E\left(\mathrm{Comm}\left(\bigcup_{\gamma\in\Gamma} \mathcal{S}^\gamma\right)\right). \label{SMeq:SMlem:design_commutant1_7}
    \end{align}
    By Eqs.~\eqref{SMeq:SMlem:design_commutant1_4} and \eqref{SMeq:SMlem:design_commutant1_7}, we have 
    \begin{align}
        \lim_{D\to\infty} \left(M_{t, \zeta_{\{\mathcal{S}^\gamma\}_{\gamma\in\Gamma}}}\right)^D 
        =\Pi_{E(\mathrm{Comm}(\bigcup_{\gamma\in\Gamma} \mathcal{S}^\gamma))}. \label{SMeq:SMlem:design_commutant1_8}
    \end{align}
    By Lemma~\ref{SMlem:moment_op_projection}, we have 
    \begin{align}
        M_{t, \mu_{\mathcal{U}_{n, G, R}}} 
        =\Pi_{E(\mathrm{Comm}(\Omega_t(\mathcal{U}_{n, G, R})))}. \label{SMeq:SMlem:design_commutant1_9}
    \end{align}
    By Eqs.~\eqref{SMeq:SMlem:design_commutant1_8} and \eqref{SMeq:SMlem:design_commutant1_9}, the distribution $\zeta_{\{\mathcal{S}^\gamma\}_{\gamma\in\Gamma}}$ is an asymptotic $(G, R)$-symmetric unitary $t$-design if and only if $E(\mathrm{Comm}(\bigcup_{\gamma\in\Gamma} \mathcal{S}^\gamma))=E(\mathrm{Comm}(\Omega_t(\mathcal{U}_{n, G, R})))$, which is equivalent to Eq.~\eqref{SMeq:SMlem:design_commutant1_2} by the bijectivity of $E$. 
\end{proof}

Next, we show that the nonexistence of a nontrivial integer solution of the equations in Theorem~\ref{SMthm:general} implies the commutant relation presented in Lemma~\ref{SMlem:design_commutant}.

\begin{lemma} \label{SMlem:trivial_to_comm_eq}
    Let $n, t\in\mathbb{N}$, $R$ be a unitary representation of a group $G$, $\{\mathcal{S}^\gamma\}_{\gamma\in\Gamma}$ be a finite set of connected compact Lie subgroup of $\mathcal{U}_{n, G, R}$, $\bigcup_{\gamma\in\Gamma} \mathcal{S}^\gamma$ be semi-universal for $\mathcal{U}_{n, G, R}$ and Eqs.~\eqref{SMeq:SMthm:general_cond1} \eqref{SMeq:SMthm:general_cond2}, and \eqref{SMeq:SMthm:general_cond3} do not have a nontrivial integer solution $(x_\lambda)_{\lambda\in\Lambda}\in\mathbb{Z}^\Lambda$.  
    Then, 
    \begin{align}
        \mathrm{Comm}\left(\Omega_t\left(\bigcup_{\gamma\in\Gamma} \mathcal{S}^\gamma\right)\right)
        =\mathrm{Comm}(\Omega_t(\mathcal{U}_{n, G, R})), \label{SMeq:SMlem:trivial_to_comm_eq0}
    \end{align}
    where $\Omega_t$ is defined by Eq.~\eqref{SMeq:Omega_def}. 
\end{lemma}

\begin{proof}
    We prove this lemma in three steps.

    In the first step, we show that 
    \begin{align}
        \sum_{(\lambda, \alpha)\in\Xi} e^{iw_{\lambda, \alpha}}P_{\lambda, \alpha}\in\Braket{\bigcup_{\gamma\in\Gamma} \mathcal{S}^\gamma}\ \forall\bm{w}=(w_{\lambda, \alpha})_{(\lambda, \alpha)\in\Xi}\in\mathcal{W}, \label{SMeq:SMlem:trivial_to_comm_eq1_1}
    \end{align}
    where a set $\Xi$, a linear subspace $\mathcal{W}$ of $\mathbb{R}^\Xi$, and projections $P_{\lambda, \alpha}$ are defined by 
    \begin{align}
        &\Xi:=\{(\lambda, \alpha)\ |\ \lambda\in\Lambda, \alpha\in\{1, 2, ..., m_\lambda\}\}, \\
        &\mathcal{W}:=\bm{\Delta}^{-1}(\mathcal{V}), \\
        &\bm{\Delta}(\bm{w}):=(\Delta_\lambda(\bm{w}))_{\lambda\in\Lambda}\ \forall\bm{w}\in\mathbb{R}^\Xi, \label{SMeq:Delta_vector_def}\\
        &\Delta_\lambda(\bm{w}):=\sum_{\alpha=1}^{m_\lambda} w_{\lambda, \alpha}\ \forall\bm{w}\in\mathbb{R}^\Xi, \label{SMeq:Delta_element_def}\\
        &\mathcal{V}:=\bm{f}\left(\mathrm{span}\left(\bigcup_{\gamma\in\Gamma} \mathfrak{s}^\gamma\right)\right), \\
        &P_{\lambda, \alpha}:=F_\lambda(I\otimes \ket{\alpha}\bra{\alpha})F_\lambda^\dag, 
    \end{align}
    and $\ket{\alpha}$ is the $\alpha$th basis vector of $\mathbb{C}^{m_\lambda}$. 
    We take arbitrary $\bm{w}\in\mathcal{W}$. 
    By the definition of $\mathcal{W}$, there exists $A\in\mathrm{span}(\bigcup_{\gamma\in\Gamma} \mathfrak{s}^\gamma)$ such that 
    \begin{align}
        \left(\sum_{\alpha=1}^{m_\lambda} w_{\lambda, \alpha}\right)_{\lambda\in\Lambda}=\bm{f}(A). \label{SMeq:SMlem:trivial_to_comm_eq1_2}
    \end{align}
    Since $A\in\mathrm{span}(\bigcup_{\gamma\in\Gamma} \mathfrak{s}^\gamma)$, $A$ can be written as 
    \begin{align}
        A=\sum_{\gamma\in\Gamma} A^\gamma \label{SMeq:SMlem:trivial_to_comm_eq1_3}
    \end{align}
    with some $A^\gamma\in\mathfrak{s}^\gamma$. 
    By noting that $A^\gamma\in\mathfrak{u}_{n, G, R}$, $A^\gamma$ can be expressed as 
    \begin{align}
        A^\gamma=\sum_{\lambda\in\Lambda} F_\lambda(I\otimes A_\lambda^\gamma)F_\lambda^\dag \label{SMeq:SMlem:trivial_to_comm_eq1_4}
    \end{align}
    with some $A_\lambda^\gamma\in\mathcal{L}(\mathbb{C}^{m_\lambda})$, which implies that
    \begin{align}
        e^{-iA^\gamma}=\sum_{\lambda\in\Lambda} F_\lambda\left(I\otimes e^{-iA_\lambda^\gamma}\right)F_\lambda^\dag. \label{SMeq:SMlem:trivial_to_comm_eq1_5}
    \end{align}
    By the definition of $P_{\lambda, \alpha}$, we have 
    \begin{align}
        \sum_{(\lambda, \alpha)\in\Xi} e^{iw_{\lambda, \alpha}}P_{\lambda, \alpha} 
        =\sum_{\lambda\in\Lambda} F_\lambda\left[I\otimes \left(\sum_{\alpha=1}^{m_\lambda} e^{iw_{\lambda, \alpha}}\ket{\alpha}\bra{\alpha}\right)\right]F_\lambda^\dag. \label{SMeq:SMlem:trivial_to_comm_eq1_6}
    \end{align}
    By Eqs.~\eqref{SMeq:SMlem:trivial_to_comm_eq1_5} and \eqref{SMeq:SMlem:trivial_to_comm_eq1_6}, we get 
    \begin{align}
        \left(\sum_{(\lambda, \alpha)\in\Xi} e^{iw_{\lambda, \alpha}}P_{\lambda, \alpha}\right)\left(\prod_{\gamma\in\Gamma} e^{-iA^\gamma}\right)
        =\sum_{\lambda\in\Lambda} F_\lambda\left(I\otimes \left(\sum_{\alpha=1}^{m_\lambda} e^{iw_{\lambda, \alpha}}\ket{\alpha}\bra{\alpha}\right)\left(\prod_{\gamma\in\Gamma} e^{-iA_\lambda^\gamma}\right)\right)F_\lambda^\dag. \label{SMeq:SMlem:trivial_to_comm_eq1_7}
    \end{align}
    By plugging Eq.~\eqref{SMeq:SMlem:trivial_to_comm_eq1_4} into Eq.~\eqref{SMeq:SMlem:trivial_to_comm_eq1_3}, we get 
    \begin{align}
        A=\sum_{\lambda\in\Lambda} F_\lambda\left(I\otimes\sum_{\gamma\in\Gamma} A_\lambda^\gamma\right)F_\lambda^\dag, 
    \end{align}
    which implies that 
    \begin{align}
        f_\lambda(A) 
        =\mathrm{tr}\left(\sum_{\gamma\in\Gamma} A_\lambda^\gamma\right) 
        =\sum_{\gamma\in\Gamma} \mathrm{tr}(A_\lambda^\gamma). 
    \end{align}
    By Eqs.~\eqref{SMeq:SMlem:trivial_to_comm_eq1_2} and \eqref{SMeq:SMlem:trivial_to_comm_eq1_7}, we get 
    \begin{align}
        \sum_{\alpha=1}^{m_\lambda} w_{\lambda, \alpha} 
        =\sum_{\gamma\in\Gamma} \mathrm{tr}(A_\lambda^\gamma), 
    \end{align}
    which implies that 
    \begin{align}
        \mathrm{det}\left(\left(\sum_{\alpha=1}^{m_\lambda} e^{iw_{\lambda, \alpha}}\ket{\alpha}\bra{\alpha}\right)\left(\prod_{\gamma\in\Gamma} e^{-iA_\lambda^\gamma}\right)\right) 
        =e^{i\sum_{\alpha=1}^{m_\lambda} w_{\lambda, \alpha}}\prod_{\gamma\in\Gamma} e^{-i\mathrm{tr}(A_\lambda^\gamma)} 
        =e^{i(\sum_{\alpha=1}^{m_\lambda} w_{\lambda, \alpha}-\sum_{\gamma\in\Gamma} \mathrm{tr}(A_\lambda^\gamma))}
        =1. \label{SMeq:SMlem:trivial_to_comm_eq1_8}
    \end{align}
    This means that the operator of the l.h.s. of Eq.~\eqref{SMeq:SMlem:trivial_to_comm_eq1_7} is in the form of $\sum_{\lambda\in\Lambda} F_\lambda (I\otimes U_\lambda)F_\lambda^\dag$ with some $U_\lambda\in\mathrm{SU}(m_\lambda)$. 
    By using the semi-universality condition shown as Eq.~\eqref{SMeq:semi_universality_equivalence}, we have 
    \begin{align}
        \left(\sum_{(\lambda, \alpha)\in\Xi} e^{iw_{\lambda, \alpha}}P_{\lambda, \alpha}\right)\left(\prod_{\gamma\in\Gamma} e^{-iA^\gamma}\right)\in\Braket{\bigcup_{\gamma\in\Gamma} \mathcal{S}^\gamma}. \label{SMeq:SMlem:trivial_to_comm_eq1_9}
    \end{align}
    Since $e^{-iA^\gamma}\in\mathcal{S}^\gamma$, we have 
    \begin{align}
        \left(\prod_{\gamma\in\Gamma} e^{-iA^\gamma}\right)^{-1} 
        \in\Braket{\bigcup_{\gamma\in\Gamma} \mathcal{S}^\gamma}. \label{SMeq:SMlem:trivial_to_comm_eq1_10}
    \end{align}
    By multiplying Eq.~\eqref{SMeq:SMlem:trivial_to_comm_eq1_9} and Eq.~\eqref{SMeq:SMlem:trivial_to_comm_eq1_10}, we get Eq.~\eqref{SMeq:SMlem:trivial_to_comm_eq1_1}.

    In the second step, we show that 
    \begin{align}
        \mathscr{S}_t\left(\bigotimes_{l=1}^t P_{\lambda_l, \alpha_l}\right)
        \in\mathrm{Alg}\left(\Omega_t\left(\Braket{\bigcup_{\gamma\in\Gamma} \mathcal{S}^\gamma}\right)\right)\ \forall (\lambda_1, \alpha_1), ..., (\lambda_t, \alpha_t)\in\Xi, \label{SMeq:SMlem:trivial_to_comm_eq2_1}
    \end{align}
    where $\mathrm{Alg}(\cdot)$ is the generated algebra over $\mathbb{C}$, and $\mathscr{S}_t$ is defined by 
    \begin{align}
        \mathscr{S}_t(A)
        :=\frac{1}{t!}\sum_{\sigma\in\mathfrak{S}_t} V_\sigma A V_\sigma^\dag. \label{SMeq:symmetrizer_def}
    \end{align}
    with $V_\sigma$ defined by 
    \begin{align}
        V_\sigma\left(\bigotimes_{\alpha=1}^m \ket{\psi_\alpha}\right)
        =\bigotimes_{\alpha=1}^m \ket{\psi_{\sigma^{-1}(\alpha)}} \label{SMeq:inter_system_permutation}
    \end{align}
    for $\sigma\in\mathfrak{S}_m$. 
    We define $z_{\lambda, \alpha}:=\#\{s\in\{1, 2, ..., t\}\ |\ (\lambda_s, \alpha_s)=(\lambda, \alpha)\}$. 
    By using Lemma~\ref{SMlem:perm_symm_proj}, we have 
    \begin{align}
        \mathscr{S}_t\left(\bigotimes_{s=1}^t P_{\lambda_s, \alpha_s}\right)
        =\mathscr{S}_t\left(\bigotimes_{(\lambda, \alpha)\in\Xi} P_{\lambda, \alpha}^{\otimes z_{\lambda, \alpha}}\right), \label{SMeq:SMlem:trivial_to_comm_eq2_2}
    \end{align}
    where we note that we do not have to specify the order in $\Xi$ due to the property of $\mathscr{S}_t$. 
    By Eq.~\eqref{SMeq:SMlem:trivial_to_comm_eq2_2}, it is sufficient to show that $\mathscr{S}_t\left(\bigotimes_{(\lambda, \alpha)\in\Xi} P_{\lambda, \alpha}^{\otimes z_{\lambda, \alpha}}\right)\in\mathrm{Alg}\left(\Omega_t\left(\bigcup_{\gamma\in\Gamma} \mathcal{S}^\gamma\right)\right)$. 
    We take an arbitrary fixed basis of $\{\bm{q}_l\}_{l=1, ..., L}$ of $\mathcal{W}$ and arbitrary $\theta_1, ..., \theta_L\in\mathbb{R}$. 
    Since $\sum_{l=1}^L \theta_l\bm{q}_l\in\mathcal{W}$, by Eq.~\eqref{SMeq:SMlem:trivial_to_comm_eq1_1}, we have 
    \begin{align}
        \sum_{(\lambda, \alpha)\in\Xi} \exp\left(i\sum_{l=1}^L \theta_l q_{l, \lambda, \alpha}\right)P_{\lambda, \alpha}
        \in\Braket{\bigcup_{\gamma\in\Gamma} \mathcal{S}^\gamma}, 
    \end{align}
    which implies that 
    \begin{align}
        \left(\sum_{(\lambda, \alpha)\in\Xi} \exp\left(i\sum_{l=1}^L \theta_l q_{l, \lambda, \alpha}\right)P_{\lambda, \alpha}\right)^{\otimes t} 
        \in\Omega_t\left(\Braket{\bigcup_{\gamma\in\Gamma} \mathcal{S}^\gamma}\right). \label{SMeq:SMlem:trivial_to_comm_eq2_3}
    \end{align}
    By Lemma~\ref{SMlem:perm_symm_proj}, we have 
    \begin{align}
        \left(\sum_{(\lambda, \alpha)\in\Xi} \exp\left(i\sum_{l=1}^L \theta_l q_{l, \lambda, \alpha}\right) P_{\lambda ,\alpha}\right)^{\otimes t}
        =&\sum_{\bm{z}'\in\mathcal{Z}_t} \frac{t!}{\displaystyle\prod_{(\lambda, \alpha)\in\Xi} z'_{\lambda, \alpha}!} \mathscr{S}_t\left(\bigotimes_{(\lambda, \alpha)\in\Xi} \left(\exp\left(i\sum_{l=1}^L \theta_l q_{l, \lambda, \alpha}\right) P_{\lambda, \alpha}\right)^{\otimes z'_{\lambda, \alpha}}\right) \nonumber\\
        =&\sum_{\bm{z}'\in\mathcal{Z}_t} \frac{t!}{\displaystyle\prod_{(\lambda, \alpha)\in\Xi} z'_{\lambda, \alpha}!} \exp\left(i\sum_{l=1}^L \theta_l \sum_{(\lambda, \alpha)\in\Xi} z'_{\lambda, \alpha}q_{l, \lambda, \alpha}\right)\mathscr{S}_t\left(\bigotimes_{(\lambda, \alpha)\in\Xi} P_{\lambda, \alpha}^{\otimes z'_{\lambda, \alpha}}\right), \label{SMeq:SMlem:trivial_to_comm_eq2_4}
    \end{align}
    where $\mathcal{Z}_t$ is defined by 
    \begin{align}
        \mathcal{Z}_t:=\left\{\bm{z}'\in(\mathbb{Z}_{\geq 0})^\Xi\ \middle|\ \sum_{(\lambda, \alpha)\in\Xi} z'_{\lambda, \alpha}=t\right\}. \label{SMeq:Z_t_def}
    \end{align}
    Equation~\eqref{SMeq:SMlem:trivial_to_comm_eq2_4} implies that 
    \begin{align}
        &\lim_{\Theta\to\infty} \frac{1}{(2\Theta)^L}\int_{-\Theta}^\Theta d\theta_L
        \cdots \int_{-\Theta}^\Theta d\theta_1
        \exp\left(-i\sum_{l=1}^L \theta_l\sum_{(\lambda, \alpha)\in\Xi} q_{l, \lambda, \alpha} z_{\lambda, \alpha}\right) 
        \left(\sum_{(\lambda, \alpha)\in\Xi} \exp\left(i\sum_{l=1}^L \theta_l q_{l, \lambda, \alpha}\right) P_{\lambda ,\alpha}\right)^{\otimes t} \nonumber\\
        =&\sum_{\bm{z}'\in\mathcal{Z}_t} \frac{t!}{\displaystyle\prod_{(\lambda, \alpha)\in\Xi} z'_{\lambda, \alpha}!}
        \prod_{l=1}^L \left(\lim_{\Theta\to\infty} \frac{1}{2\Theta}\int_{-\Theta}^\Theta \exp\left(i\theta_l\left(\sum_{(\lambda, \alpha)\in\Xi} q_{l, \lambda, \alpha} z'_{\lambda, \alpha}-\sum_{(\lambda, \alpha)\in\Xi} q_{l, \lambda, \alpha} z_{\lambda, \alpha}\right)\right) d\theta_l\right)
        \mathscr{S}_t\left(\bigotimes_{(\lambda, \alpha)\in\Xi} P_{\lambda, \alpha}^{\otimes z'_{\lambda, \alpha}}\right) \nonumber\\
        =&\sum_{\bm{z}'\in\mathcal{Z}_t} \frac{t!}{\displaystyle\prod_{(\lambda, \alpha)\in\Xi} z'_{\lambda, \alpha}!}
        \prod_{l=1}^L 
        \delta_{\sum_{(\lambda, \alpha)\in\Xi} q_{l, \lambda, \alpha} z'_{\lambda, \alpha}, \sum_{(\lambda, \alpha)\in\Xi} q_{l, \lambda, \alpha}z_{\lambda, \alpha}}
        \mathscr{S}_t\left(\bigotimes_{(\lambda, \alpha)\in\Xi} P_{\lambda, \alpha}^{\otimes z'_{\lambda, \alpha}}\right). \label{SMeq:SMlem:trivial_to_comm_eq2_5}
    \end{align}
    Since there do not exist nontrivial $(x_\lambda)_{\lambda\in\Lambda}\in\mathbb{Z}^\Lambda$ satisfying Eqs.~\eqref{SMeq:SMthm:general_cond1}, \eqref{SMeq:SMthm:general_cond2}, and \eqref{SMeq:SMthm:general_cond3} by assumption, Lemma~\ref{SMlem:trivial_solution} implies that if $\sum_{(\lambda, \alpha)\in\Xi} z_{\lambda, \alpha}=\sum_{(\lambda, \alpha)\in\Xi} z'_{\lambda, \alpha}=t$ and $\sum_{(\lambda, \alpha)\in\Xi} w_{\lambda, \alpha} z_{\lambda, \alpha}=\sum_{(\lambda, \alpha)\in\Xi} w_{\lambda, \alpha} z'_{\lambda, \alpha}$ for all $(w_{\lambda, \alpha})\in\mathcal{W}$, then we have $\bm{z}=\bm{z}'$. 
    This can be rephrased as 
    \begin{align}
        \prod_{l=1}^L 
        \delta_{\sum_{(\lambda, \alpha)\in\Xi} q_{l, \lambda, \alpha} z'_{\lambda, \alpha}, \sum_{(\lambda, \alpha)\in\Xi} q_{l, \lambda, \alpha} z_{\lambda, \alpha}}
        =\prod_{(\lambda, \alpha)\in\Xi} 
        \delta_{z'_{\lambda, \alpha}, z_{\lambda, \alpha}}. \label{SMeq:SMlem:trivial_to_comm_eq2_6}
    \end{align}
    By Eqs.~\eqref{SMeq:SMlem:trivial_to_comm_eq2_5} and \eqref{SMeq:SMlem:trivial_to_comm_eq2_6}, we get 
    \begin{align}
        &\lim_{\Theta\to\infty} \frac{1}{(2\Theta)^L}\int_{-\Theta}^\Theta d\theta_L
        \cdots \int_{-\Theta}^\Theta d\theta_1
        \exp\left(-i\sum_{l=1}^L \theta_l\sum_{(\lambda, \alpha)\in\Xi} q_{l, \lambda, \alpha} z_{\lambda, \alpha}\right) 
        \left(\sum_{(\lambda, \alpha)\in\Xi} \exp\left(i\sum_{l=1}^L \theta_l q_{l, \lambda, \alpha}\right) P_{\lambda ,\alpha}\right)^{\otimes t} \nonumber\\
        =&\frac{t!}{\displaystyle\prod_{(\lambda, \alpha)\in\Xi} z_{\lambda, \alpha}!} \mathscr{S}_t\left(\bigotimes_{(\lambda, \alpha)\in\Xi} P_{\lambda, \alpha}^{\otimes z_{\lambda, \alpha}}\right). \label{SMeq:SMlem:trivial_to_comm_eq2_7}
    \end{align}
    By Eq.~\eqref{SMeq:SMlem:trivial_to_comm_eq2_3}, the l.h.s. of Eq.~\eqref{SMeq:SMlem:trivial_to_comm_eq2_7} is an element of $\mathrm{Alg}(\Omega_t(\braket{\bigcup_{\gamma\in\Gamma} \mathcal{S}^\gamma}))$. 
    Thus Eq.~\eqref{SMeq:SMlem:trivial_to_comm_eq2_7} implies that $\mathscr{S}_t\left(\bigotimes_{\lambda, \alpha} P_{\lambda, \alpha}^{\otimes z_{\lambda, \alpha}}\right)\in\mathrm{Alg}(\Omega_t(\braket{\bigcup_{\gamma\in\Gamma} \mathcal{S}^\gamma}))$. 
    By combining this with Eq.~\eqref{SMeq:SMlem:trivial_to_comm_eq2_2}, we get Eq.~\eqref{SMeq:SMlem:trivial_to_comm_eq2_1}.

    In the final step, we show Eq.~\eqref{SMeq:SMlem:trivial_to_comm_eq0}. 
    We take arbitrary $U\in\mathrm{Z}(\mathcal{U}_{n, G, R})$. 
    Then, $U$ can be written as 
    \begin{align}
        U 
        =\sum_{\lambda\in\Lambda} F_\lambda(I\otimes u_\lambda I)F_\lambda^\dag 
        =\sum_{(\lambda, \alpha)\in\Xi} u_\lambda P_{\lambda, \alpha} \label{SMeq:SMlem:trivial_to_comm_eq3_1}
    \end{align}
    with some $u_\lambda\in\mathbb{C}$. 
    By the definitions of $\Omega_t$ and $\mathscr{S}_t$ and Eq.~\eqref{SMeq:SMlem:trivial_to_comm_eq3_1}, we get 
    \begin{align}
        \Omega_t(U) 
        =&U^{\otimes t} \nonumber\\
        =&\mathscr{S}_t(U^{\otimes t}) \nonumber\\
        =&\mathscr{S}_t\left(\sum_{(\lambda_1, \alpha_1), ..., (\lambda_t, \alpha_t)\in\Xi} \bigotimes_{s=1}^t u_{\lambda_s} P_{\lambda_s, \alpha_s}\right) \nonumber\\
        =&\sum_{(\lambda_1, \alpha_1), ..., (\lambda_t, \alpha_t)\in\Xi} \left(\prod_{s=1}^t u_{\lambda_s}\right) \mathscr{S}_t\left(\bigotimes_{s=1}^t u_{\lambda_s} P_{\lambda_s, \alpha_s}\right) \nonumber\\
        \in&\mathrm{Alg}\left(\Omega_t\left(\Braket{\bigcup_{\gamma\in\Gamma} \mathcal{S}^\gamma}\right)\right). 
    \end{align}
    Since this holds for all $U\in Z(\mathcal{U}_{n, G, R})$, we have 
    \begin{align}
        \Omega_t(Z(\mathcal{U}_{n, G, R})) 
        \subset\mathrm{Alg}\left(\Omega_t\left(\Braket{\bigcup_{\gamma\in\Gamma} \mathcal{S}^\gamma}\right)\right). 
    \end{align}
    By taking the commutant of the both sides, we get 
    \begin{align}
        \mathrm{Comm}(\Omega_t(Z(\mathcal{U}_{n, G, R}))) 
        \subset\mathrm{Comm}\left(\mathrm{Alg}\left(\Omega_t\left(\Braket{\bigcup_{\gamma\in\Gamma} \mathcal{S}^\gamma}\right)\right)\right) 
        =\mathrm{Comm}\left(\Omega_t\left(\Braket{\bigcup_{\gamma\in\Gamma} \mathcal{S}^\gamma}\right)\right). \label{SMeq:SMlem:trivial_to_comm_eq3_2}
    \end{align}
    Since $\bigcup_{\gamma\in\Gamma} \mathcal{S}^\gamma$ is semi-universal for $\mathcal{U}_{n, G, R}$, we have 
    \begin{align}
        \Omega_t\left(\Braket{\bigcup_{\gamma\in\Gamma} \mathcal{S}^\gamma}\right)\cdot\Omega_t(Z(\mathcal{U}_{n, G, R})) 
        =\Omega_t(\mathcal{U}_{n, G, R}). 
    \end{align}
    By taking the commutant of this equation, we get 
    \begin{align}
        \mathrm{Comm}\left(\Omega_t\left(\Braket{\bigcup_{\gamma\in\Gamma} \mathcal{S}^\gamma}\right)\right)\cap\mathrm{Comm}\left(\Omega_t(Z(\mathcal{U}_{n, G, R}))\right) 
        =\mathrm{Comm}\left(\Omega_t(\mathcal{U}_{n, G, R})\right). \label{SMeq:SMlem:trivial_to_comm_eq3_3}
    \end{align}
    By Eqs.~\eqref{SMeq:SMlem:trivial_to_comm_eq3_2} and \eqref{SMeq:SMlem:trivial_to_comm_eq3_3}, we get Eq.~\eqref{SMeq:SMlem:trivial_to_comm_eq0}. 
\end{proof}

Finally, we show the converse of Lemma~\ref{SMlem:trivial_to_comm_eq}.

\begin{lemma} \label{SMlem:nontrivial_to_comm_neq}
	Let $n, t\in\mathbb{N}$, $\{\mathcal{S}^\gamma\}_{\gamma\in\Gamma}$ be a finite set of connected compact unitary subgroups of $\mathcal{U}_{n, G, R}$, and Eqs.~\eqref{SMeq:SMthm:general_cond1}, \eqref{SMeq:SMthm:general_cond2}, and \eqref{SMeq:SMthm:general_cond3} have a nontrivial integer solution $\bm{x}\in\mathbb{Z}^\Lambda$. 
	Then, 
    \begin{align}
        \mathrm{Comm}\left(\Omega_t\left(\bigcup_{\gamma\in\Gamma} \mathcal{S}^\gamma\right)\right) 
        \neq\mathrm{Comm}(\Omega_t(\mathcal{U}_{n, G, R})). \label{SMeq:SMlem:nontrivial_to_comm_neq0}
    \end{align}
\end{lemma}

\begin{proof}
    We prove this lemma in three steps.

    In the first step, we show that 
    \begin{align}
        \omega_t(A)\ket{\Phi(\bm{p})}=\left(\sum_{\lambda\in\Lambda} m_\lambda\mathrm{tr}(A_\lambda)\right)\ket{\Phi(\bm{p})}, \label{SMeq:SMlem:nontrivial_to_comm_neq1_1}
    \end{align}
    for all $A\in\mathfrak{u}_{n, G, R}$ and $\bm{p}\in(\mathbb{Z}_{\geq 0})^\Lambda$ satisfying $\sum_{\lambda\in\Lambda} m_\lambda p_\lambda=t$, where $A_\lambda$ is defined by Eq.~\eqref{SMeq:sym_op_decomp}, and 
    \begin{align}
        \ket{\Phi(\bm{p})}:=\bigotimes_{\lambda\in\Lambda} \left[F_\lambda^{\otimes m_\lambda}(\ket{\psi_\lambda}^{\otimes m_\lambda}\otimes \ket{\chi(\mathbb{C}^{m_\lambda})})\right]^{\otimes p_\lambda}, 
    \end{align}
    with arbitrarily chosen states $\ket{\psi_\lambda}\in\mathbb{C}^{r_\lambda}$ and 
    \begin{align}
        &\ket{\chi(\mathbb{C}^m)}:=\frac{1}{\sqrt{m!}}\sum_{\sigma\in\mathfrak{S}_m} \mathrm{sgn}(\sigma)\bigotimes_{\alpha=1}^m \ket{\sigma(\alpha)}, \label{SMeq:antisymmetric_state_def}\\
		&\omega_t(A):=\sum_{s=1}^t I^{\otimes s-1}\otimes A\otimes I^{\otimes t-s}. \label{SMeq:omega_def}
	\end{align} 
    By noting that $(O_1\otimes I+I\otimes O_2)(\ket{\phi_1}\otimes\ket{\phi_2})=(\alpha_1+\alpha_2)(\ket{\phi_1}\otimes\ket{\phi_2})$ when $O_j\ket{\phi_j}=\alpha_j\ket{\phi_j}$, for the proof of \eqref{SMeq:SMlem:nontrivial_to_comm_neq1_1}, it is sufficient to show that 
    \begin{align}
        \omega_{m_\lambda}(A)\left[F_\lambda^{\otimes m_\lambda}(\ket{\psi_\lambda}^{\otimes m_\lambda}\otimes \ket{\chi(\mathbb{C}^{m_\lambda})})\right]
        =\mathrm{tr}(A_\lambda)\left[F_\lambda^{\otimes m_\lambda}(\ket{\psi_\lambda}^{\otimes m_\lambda}\otimes \ket{\chi(\mathbb{C}^{m_\lambda})})\right] \label{SMeq:SMlem:nontrivial_to_comm_neq1_2}
    \end{align}
    for all $\lambda\in\Lambda$. 
    By the decomposition of $A$, we have 
    \begin{align}
        AF_\lambda
        =\sum_{\mu\in\Lambda} F_\mu(I\otimes A)F_\mu^\dag F_\lambda. \label{SMeq:SMlem:nontrivial_to_comm_neq1_3}
    \end{align}
    By the definition of $F_\lambda$'s, we have 
    \begin{align}
        F_\mu^\dag F_\lambda=
        \begin{cases}
			I & \textrm{if}\ \lambda=\mu \\
			0 & \textrm{if}\ \lambda\neq\mu. 
		\end{cases}
        \label{SMeq:SMlem:nontrivial_to_comm_neq1_4}
    \end{align}
    By plugging Eq.~\eqref{SMeq:SMlem:nontrivial_to_comm_neq1_4} into Eq.~\eqref{SMeq:SMlem:nontrivial_to_comm_neq1_3}, we get 
    \begin{align}
        AF_\lambda=F_\lambda(I\otimes A_\lambda). 
    \end{align}
    By the definition of $\omega_{m_\lambda}$ and this equation, we get 
    \begin{align}
        \omega_{m_\lambda}(A)F_\lambda^{\otimes m_\lambda}
        =&\sum_{s=1}^{m_\lambda} F_\lambda^{\otimes s-1}\otimes AF_\lambda\otimes F_\lambda^{\otimes m_\lambda-s} \nonumber\\
        =&\sum_{s=1}^{m_\lambda} F_\lambda^{\otimes s-1}\otimes F_\lambda(I\otimes A_\lambda)\otimes F_\lambda^{\otimes m_\lambda-s} \nonumber\\
        =&F_\lambda^{\otimes m_\lambda}\omega_{m_\lambda}(I\otimes A_\lambda) \nonumber\\
        =&F_\lambda^{\otimes m_\lambda}(I^{\otimes m_\lambda}\otimes \omega_{m_\lambda}(A_\lambda)), 
    \end{align}
    which implies that 
    \begin{align}
        \omega_{m_\lambda}(A)\left[F_\lambda^{\otimes m_\lambda}(\ket{\psi_\lambda}^{\otimes m_\lambda}\otimes \ket{\chi(\mathbb{C}^{m_\lambda})})\right]
        =F_\lambda^{\otimes m_\lambda}(\ket{\psi_\lambda}^{\otimes m_\lambda}\otimes\omega_{m_\lambda}(A_\lambda)\ket{\chi(\mathbb{C}^{m_\lambda})}). 
    \end{align}
    By applying Lemma~\ref{SMlem:antisymmetric_state} to the r.h.s. of this equation, we get Eq.~\eqref{SMeq:SMlem:nontrivial_to_comm_neq1_2}.

    In the second step, we show that 
    \begin{align}
        \mathrm{Comm}\left(\omega_t\left(\bigcup_{\gamma\in\Gamma} \mathfrak{s}^\gamma\right)\right)
        \neq\mathrm{Comm}(\omega_t(\mathfrak{u}_{n, G, R})). \label{SMeq:SMlem:nontrivial_to_comm_neq2_1}
    \end{align}
    For the proof of this, we construct an operator $O$ such that $O\in\mathrm{Comm}(\bigcup_{\gamma\in\Gamma} \mathfrak{s}^\gamma)$ and $O\not\in\mathrm{Comm}(\mathfrak{u}_{n, G, R})$. 
    By Lemma~\ref{SMlem:trivial_solution}, we can take two different vectors $\bm{y}, \bm{y}'\in(\mathbb{Z}_{\geq 0})^\Lambda$ satisfying $\sum_{\lambda\in\Lambda} m_\lambda y_\lambda=\sum_{\lambda\in\Lambda} m_\lambda y'_\lambda\leq t$ and $\sum_{\lambda\in\Lambda} y_\lambda v_\lambda=\sum_{\lambda\in\Lambda} y'_\lambda v_\lambda$ for all $\bm{v}\in \bm{f}(\mathrm{span}(\bigcup_{\lambda\in\Lambda} \mathfrak{s}^\gamma))$. 
    We define $O$ by 
    \begin{align}
        O:=\ket{\Phi(\bm{y})}\bra{\Phi(\bm{y}')}\otimes I^{\otimes u}, 
    \end{align}
    where $u:=t-\sum_{\lambda\in\Lambda} m_\lambda y_\lambda$. 
    By Eq.~\eqref{SMeq:SMlem:nontrivial_to_comm_neq1_1}, we have 
    \begin{align}
        [\omega_t(A), O]=\left(\sum_{\lambda\in\Lambda} y_\lambda\mathrm{tr}(A_\lambda)-\sum_{\lambda\in\Lambda} y'_\lambda\mathrm{tr}(A_\lambda)\right)O=0\ \forall A\in\bigcup_{\gamma\in\Gamma} \mathfrak{s}^\gamma, 
    \end{align}
    which means that $O\in\mathrm{Comm}(\bigcup_{\gamma\in\Gamma} \mathfrak{s}^\gamma)$. 
    Since we have $\bm{y}$ and $\bm{y}'$ are different, we can take $\kappa\in\Lambda$ such that $y_\kappa\neq y'_\kappa$. 
    We define $\widetilde{P}_\kappa:=\sum_{\alpha=1}^{m_\kappa} P_{\kappa, \alpha}$. 
    Then, we have $\widetilde{P}_\kappa\in\mathfrak{u}_{n, G, R}$, and Eq.~\eqref{SMeq:SMlem:nontrivial_to_comm_neq1_1} implies that 
    \begin{align}
        [\omega_t(\widetilde{P}_\kappa), O]
        =(y_\mu \mathrm{tr}(\widetilde{P}_\kappa)-y'_\mu \mathrm{tr}(\widetilde{P}_\kappa))O 
        =m_\lambda(y_\kappa-y'_\kappa)O
        \neq 0, 
    \end{align}
    which means that $O\not\in\mathrm{Comm}(\mathfrak{u}_{n, G, R})$. 
    Thus we have proven Eq.~\eqref{SMeq:SMlem:nontrivial_to_comm_neq0}.

    Finally, we show Eq.~\eqref{SMeq:SMlem:nontrivial_to_comm_neq0}. 
    We note that  
    \begin{align}
        &\mathrm{Comm}\left(\Omega_t\left(\bigcup_{\gamma\in\Gamma} \mathcal{S}^\gamma\right)\right) 
        =\mathrm{Comm}\left(\bigcup_{\gamma\in\Gamma} \Omega_t(\mathcal{S}^\gamma)\right) 
        =\bigcap_{\gamma\in\Gamma} \mathrm{Comm}\left(\Omega_t(\mathcal{S}^\gamma)\right), \label{SMeq:SMlem:nontrivial_to_comm_neq3_1}\\
        &\mathrm{Comm}\left(\omega_t\left(\bigcup_{\gamma\in\Gamma} \mathfrak{s}^\gamma\right)\right) 
        =\mathrm{Comm}\left(\bigcup_{\gamma\in\Gamma} \omega_t(\mathfrak{s}^\gamma)\right) 
        =\bigcap_{\gamma\in\Gamma} \mathrm{Comm}\left(\omega_t(\mathfrak{s}^\gamma)\right). \label{SMeq:SMlem:nontrivial_to_comm_neq3_2}
    \end{align}
    By Lemma~\ref{SMlem:Lie_group_algebra_comm}, we have 
    \begin{align}
        \mathrm{Comm}\left(\Omega_t(\mathcal{S}^\gamma)\right) 
        =\mathrm{Comm}\left(\omega_t(\mathfrak{s}^\gamma)\right). \label{SMeq:SMlem:nontrivial_to_comm_neq3_3}
    \end{align}
    By Eqs.~\eqref{SMeq:SMlem:nontrivial_to_comm_neq3_1}, \eqref{SMeq:SMlem:nontrivial_to_comm_neq3_2}, and \eqref{SMeq:SMlem:nontrivial_to_comm_neq3_3}, we get 
    \begin{align}
        \mathrm{Comm}\left(\Omega_t\left(\bigcup_{\gamma\in\Gamma} \mathcal{S}^\gamma\right)\right) 
        =\mathrm{Comm}\left(\omega_t\left(\bigcup_{\gamma\in\Gamma} \mathfrak{s}^\gamma\right)\right). \label{SMeq:SMlem:nontrivial_to_comm_neq3_4}
    \end{align}
    By using Lemma~\ref{SMlem:Lie_group_algebra_comm} again, we have 
    \begin{align}
        \mathrm{Comm}\left(\omega_t(\mathfrak{u}_{n, G, R})\right)
        =\mathrm{Comm}\left(\Omega_t(\mathcal{U}_{n, G, R})\right). \label{SMeq:SMlem:nontrivial_to_comm_neq3_5}
    \end{align}
    Equations~\eqref{SMeq:SMlem:nontrivial_to_comm_neq3_4} and \eqref{SMeq:SMlem:nontrivial_to_comm_neq3_5} imply the equivalence between Eq.~\eqref{SMeq:SMlem:nontrivial_to_comm_neq2_1} and Eq.~\eqref{SMeq:SMlem:nontrivial_to_comm_neq0}. 
\end{proof}

By combining the lemmas above, we get the proof of Theorem~\ref{SMthm:general} as follows: \\

\noindent
\textit{Proof of Theorem~\ref{SMthm:general}.}
The ``if'' part follows from the combination of Lemmas~\ref{SMlem:design_commutant} and \ref{SMlem:trivial_to_comm_eq}, and the ``only if'' part follows from the combination of Lemmas~\ref{SMlem:design_commutant} and \ref{SMlem:nontrivial_to_comm_neq}. 
When $R$ can be written as $T^{\otimes n}$ with a single-qudit representation $T$, Eq.~\eqref{SMeq:SMthm:general_cond2} is implied by Eq.~\eqref{SMeq:SMthm:general_cond3}, because $I\in\mathfrak{u}_{n, G, R}^\gamma$ with some $\gamma\in\Gamma$ (actually for all $\gamma\in\Gamma$) and $f_\lambda(I)=m_\lambda$. 
By Lemma~\ref{SMlem:simplified_condition}, Eq.~\eqref{SMeq:SMthm:general_cond3} is equivalent to Eq.~\eqref{SMeq:SMthm:general_cond4}. 
By combining these two statements, we can confirm that Eqs.~\eqref{SMeq:SMthm:general_cond2} and \eqref{SMeq:SMthm:general_cond3} are equivalent to Eq.~\eqref{SMeq:SMthm:general_cond4}. 
\hfill $\square$

\section{Conclusion and Discussion} \label{sec:conclusion}
In this work, we have proposed a general method for calculating the maximal $t$ such that the random circuits with a gate set of connected compact unitary subgroups form asymptotic symmetric unitary $t$-designs. 
In particular, we have explicitly identified the tight bound on the maximal achievable order of unitary designs of symmetric local random circuits in the cases of $\mathbb{Z}_2$, $\mathrm{U}(1)$, and $\mathrm{SU}(2)$ symmetries.
Although we have focused on the above symmetries, our method is general and useful for calculating the maximal order of design for other symmetries as long as the gate set satisfies the semi-universality. 
On the other hand, symmetric random circuits that do not satisfy the semi-university do not generate asymptotic symmetric unitary $2$-designs. 
We can therefore show the maximal order of designs of arbitrary symmetric random circuits, once we know if a given gate set satisfies the semi-universality. 
In this sense, we have fully characterized the randomness of symmetric local random circuits.

Although we have only considered the local random circuit where we apply one gate at each time step, the maximal order of design is the same for a random circuit with other architectures, such as the brick-wall architecture, as long as the circuit cannot be separated into two independent parts and the representation is the tensor product of a single-qudit representation.

It is an important open problem to derive the rate to generate an asymptotic symmetric unitary $t$-design in symmetric local random circuits. 
Without any symmetry, it has been shown recently that local random circuits are unitary $t$-designs if the circuit depth is linear in $t$ \cite{chen2024incompressibilityspectralgapsrandom}. 
It would be interesting to ask if the $t$-dependence on the convergence rate is the same under a symmetry. 
Moreover, while $n$-qubit local random circuits without any symmetry have been shown to form unitary $t$-designs with a logarithmic depth in $n$ \cite{schuster2024random}, the situation is completely different under a symmetry: it is observed that symmetric circuits require superlinear depth in the case of $\mathrm{U}(1)$ symmetry~\cite{hearth2025unitary} and $\mathrm{SU}(2)$ symmetry~\cite{li2023designs}. 
Therefore, it is desirable to characterize how the convergence rate depends on the qubit count $n$ under general symmetry.
In addition, we believe that our work will open up new directions for future research. 
In the proof of Lemma~\ref{SMlem:nontrivial_to_comm_neq}, we have found a conserved quantity on $t$-copy states which evolve under symmetric and local dynamics. 
To investigate the consequence of such conservation law for physical properties, such as thermalization and entanglement dynamics, would also be interesting.

\section*{Acknowledgements}
The authors wish to thank 
Iman Marvian, 
Hiroyasu Tajima, 
Janek Denzler,
and Zongping Gong
for insightful discussions. 
Y.M. is supported by JSPS KAKENHI Grant No. JP23KJ0421.
R.S. is supported by the BMBF (PhoQuant, Grant No. 13N16103).
This research is funded in part by the
Gordon and Betty Moore Foundation's EPiQS Initiative,
Grant GBMF8683 to T.S.
N.Y. wishes to thank JST PRESTO No. JPMJPR2119, JST ASPIRE Grant Number JPMJAP2316,
and the support from IBM Quantum.
This work was supported by JST Grant Number JPMJPF2221, JST ERATO Grant Number JPMJER2302, and JST CREST Grant Number JPMJCR23I4, Japan.

Note Added: During the preparation of this article, we became aware of independent work by 
Austin Hulse, Hanqing Liu, and Iman Marvian ~\cite{hulse2024unitary}, which studies similar questions and was posted on arXiv concurrently with the present paper.
Both have arrived at the same result on the maximal order of unitary designs under the $\mathrm{U}(1)$ and $\mathrm{SU}(2)$ symmetries. 
Reference~\cite{hulse2024unitary} has assumed conjectures about combinatorial identities, which are introduced as Eqs.~(86) and (120) of the version~1 of their manuscript for the proof of general $k$-local cases.
In our work, we have provided a proof that is independent of any conjectures.

\bibliography{bib.bib}

\begin{thebibliography}{53}%
\makeatletter
\providecommand \@ifxundefined [1]{%
 \@ifx{#1\undefined}
}%
\providecommand \@ifnum [1]{%
 \ifnum #1\expandafter \@firstoftwo
 \else \expandafter \@secondoftwo
 \fi
}%
\providecommand \@ifx [1]{%
 \ifx #1\expandafter \@firstoftwo
 \else \expandafter \@secondoftwo
 \fi
}%
\providecommand \natexlab [1]{#1}%
\providecommand \enquote  [1]{``#1''}%
\providecommand \bibnamefont  [1]{#1}%
\providecommand \bibfnamefont [1]{#1}%
\providecommand \citenamefont [1]{#1}%
\providecommand \href@noop [0]{\@secondoftwo}%
\providecommand \href [0]{\begingroup \@sanitize@url \@href}%
\providecommand \@href[1]{\@@startlink{#1}\@@href}%
\providecommand \@@href[1]{\endgroup#1\@@endlink}%
\providecommand \@sanitize@url [0]{\catcode `\\12\catcode `\$12\catcode `\&12\catcode `\#12\catcode `\^12\catcode `\_12\catcode `\%12\relax}%
\providecommand \@@startlink[1]{}%
\providecommand \@@endlink[0]{}%
\providecommand \url  [0]{\begingroup\@sanitize@url \@url }%
\providecommand \@url [1]{\endgroup\@href {#1}{\urlprefix }}%
\providecommand \urlprefix  [0]{URL }%
\providecommand \Eprint [0]{\href }%
\providecommand \doibase [0]{https://doi.org/}%
\providecommand \selectlanguage [0]{\@gobble}%
\providecommand \bibinfo  [0]{\@secondoftwo}%
\providecommand \bibfield  [0]{\@secondoftwo}%
\providecommand \translation [1]{[#1]}%
\providecommand \BibitemOpen [0]{}%
\providecommand \bibitemStop [0]{}%
\providecommand \bibitemNoStop [0]{.\EOS\space}%
\providecommand \EOS [0]{\spacefactor3000\relax}%
\providecommand \BibitemShut  [1]{\csname bibitem#1\endcsname}%
\let\auto@bib@innerbib\@empty
\bibitem [{\citenamefont {Noether}(1918)}]{Noether1918}%
  \BibitemOpen
  \bibfield  {author} {\bibinfo {author} {\bibfnamefont {E.}~\bibnamefont {Noether}},\ }\bibfield  {title} {\bibinfo {title} {Invariante variationsprobleme},\ }\href {http://eudml.org/doc/59024} {\bibfield  {journal} {\bibinfo  {journal} {Nachr. Ges. Wiss. Gottingen}\ }\textbf {\bibinfo {volume} {1918}},\ \bibinfo {pages} {235} (\bibinfo {year} {1918})}\BibitemShut {NoStop}%
\bibitem [{\citenamefont {Nambu}(1960)}]{nambu1960axial}%
  \BibitemOpen
  \bibfield  {author} {\bibinfo {author} {\bibfnamefont {Y.}~\bibnamefont {Nambu}},\ }\bibfield  {title} {\bibinfo {title} {Axial vector current conservation in weak interactions},\ }\href {https://doi.org/10.1103/PhysRevLett.4.380} {\bibfield  {journal} {\bibinfo  {journal} {Phys. Rev. Lett.}\ }\textbf {\bibinfo {volume} {4}},\ \bibinfo {pages} {380} (\bibinfo {year} {1960})}\BibitemShut {NoStop}%
\bibitem [{\citenamefont {Nambu}\ and\ \citenamefont {Jona-Lasinio}(1961)}]{nambu1961dynamical}%
  \BibitemOpen
  \bibfield  {author} {\bibinfo {author} {\bibfnamefont {Y.}~\bibnamefont {Nambu}}\ and\ \bibinfo {author} {\bibfnamefont {G.}~\bibnamefont {Jona-Lasinio}},\ }\bibfield  {title} {\bibinfo {title} {Dynamical model of elementary particles based on an analogy with superconductivity. i},\ }\href {https://doi.org/10.1103/PhysRev.122.345} {\bibfield  {journal} {\bibinfo  {journal} {Phys. Rev.}\ }\textbf {\bibinfo {volume} {122}},\ \bibinfo {pages} {345} (\bibinfo {year} {1961})}\BibitemShut {NoStop}%
\bibitem [{\citenamefont {Goldstone}(1961)}]{goldstone1961field}%
  \BibitemOpen
  \bibfield  {author} {\bibinfo {author} {\bibfnamefont {J.}~\bibnamefont {Goldstone}},\ }\bibfield  {title} {\bibinfo {title} {Field theories with superconductor solutions},\ }\href {https://doi.org/10.1007/BF02812722} {\bibfield  {journal} {\bibinfo  {journal} {Nuovo Cim.}\ }\textbf {\bibinfo {volume} {19}},\ \bibinfo {pages} {154} (\bibinfo {year} {1961})}\BibitemShut {NoStop}%
\bibitem [{\citenamefont {Nambu}(2009)}]{nambu2009nobel}%
  \BibitemOpen
  \bibfield  {author} {\bibinfo {author} {\bibfnamefont {Y.}~\bibnamefont {Nambu}},\ }\bibfield  {title} {\bibinfo {title} {Nobel lecture: Spontaneous symmetry breaking in particle physics: A case of cross fertilization},\ }\href {https://doi.org/10.1103/RevModPhys.81.1015} {\bibfield  {journal} {\bibinfo  {journal} {Rev. Mod. Phys.}\ }\textbf {\bibinfo {volume} {81}},\ \bibinfo {pages} {1015} (\bibinfo {year} {2009})}\BibitemShut {NoStop}%
\bibitem [{\citenamefont {Senthil}\ \emph {et~al.}(2004{\natexlab{a}})\citenamefont {Senthil}, \citenamefont {Vishwanath}, \citenamefont {Balents}, \citenamefont {Sachdev},\ and\ \citenamefont {Fisher}}]{senthil2004deconfined}%
  \BibitemOpen
  \bibfield  {author} {\bibinfo {author} {\bibfnamefont {T.}~\bibnamefont {Senthil}}, \bibinfo {author} {\bibfnamefont {A.}~\bibnamefont {Vishwanath}}, \bibinfo {author} {\bibfnamefont {L.}~\bibnamefont {Balents}}, \bibinfo {author} {\bibfnamefont {S.}~\bibnamefont {Sachdev}},\ and\ \bibinfo {author} {\bibfnamefont {M.~P.~A.}\ \bibnamefont {Fisher}},\ }\bibfield  {title} {\bibinfo {title} {Deconfined quantum critical points},\ }\href {https://doi.org/10.1126/science.1091806} {\bibfield  {journal} {\bibinfo  {journal} {Science}\ }\textbf {\bibinfo {volume} {303}},\ \bibinfo {pages} {1490} (\bibinfo {year} {2004}{\natexlab{a}})}\BibitemShut {NoStop}%
\bibitem [{\citenamefont {Senthil}\ \emph {et~al.}(2004{\natexlab{b}})\citenamefont {Senthil}, \citenamefont {Balents}, \citenamefont {Sachdev}, \citenamefont {Vishwanath},\ and\ \citenamefont {Fisher}}]{senthil2004quantum}%
  \BibitemOpen
  \bibfield  {author} {\bibinfo {author} {\bibfnamefont {T.}~\bibnamefont {Senthil}}, \bibinfo {author} {\bibfnamefont {L.}~\bibnamefont {Balents}}, \bibinfo {author} {\bibfnamefont {S.}~\bibnamefont {Sachdev}}, \bibinfo {author} {\bibfnamefont {A.}~\bibnamefont {Vishwanath}},\ and\ \bibinfo {author} {\bibfnamefont {M.~P.~A.}\ \bibnamefont {Fisher}},\ }\bibfield  {title} {\bibinfo {title} {Quantum criticality beyond the landau-ginzburg-wilson paradigm},\ }\href {https://doi.org/10.1103/PhysRevB.70.144407} {\bibfield  {journal} {\bibinfo  {journal} {Phys. Rev. B}\ }\textbf {\bibinfo {volume} {70}},\ \bibinfo {pages} {144407} (\bibinfo {year} {2004}{\natexlab{b}})}\BibitemShut {NoStop}%
\bibitem [{\citenamefont {Sandvik}(2007)}]{sandvik2007evidence}%
  \BibitemOpen
  \bibfield  {author} {\bibinfo {author} {\bibfnamefont {A.~W.}\ \bibnamefont {Sandvik}},\ }\bibfield  {title} {\bibinfo {title} {Evidence for deconfined quantum criticality in a two-dimensional heisenberg model with four-spin interactions},\ }\href {https://doi.org/10.1103/PhysRevLett.98.227202} {\bibfield  {journal} {\bibinfo  {journal} {Phys. Rev. Lett.}\ }\textbf {\bibinfo {volume} {98}},\ \bibinfo {pages} {227202} (\bibinfo {year} {2007})}\BibitemShut {NoStop}%
\bibitem [{\citenamefont {Shor}(1995)}]{shor1995scheme}%
  \BibitemOpen
  \bibfield  {author} {\bibinfo {author} {\bibfnamefont {P.~W.}\ \bibnamefont {Shor}},\ }\bibfield  {title} {\bibinfo {title} {Scheme for reducing decoherence in quantum computer memory},\ }\href {https://doi.org/10.1103/PhysRevA.52.R2493} {\bibfield  {journal} {\bibinfo  {journal} {Phys. Rev. A}\ }\textbf {\bibinfo {volume} {52}},\ \bibinfo {pages} {R2493} (\bibinfo {year} {1995})}\BibitemShut {NoStop}%
\bibitem [{\citenamefont {Steane}(1996)}]{steane1996error}%
  \BibitemOpen
  \bibfield  {author} {\bibinfo {author} {\bibfnamefont {A.~M.}\ \bibnamefont {Steane}},\ }\bibfield  {title} {\bibinfo {title} {Error correcting codes in quantum theory},\ }\href {https://doi.org/10.1103/PhysRevLett.77.793} {\bibfield  {journal} {\bibinfo  {journal} {Phys. Rev. Lett.}\ }\textbf {\bibinfo {volume} {77}},\ \bibinfo {pages} {793} (\bibinfo {year} {1996})}\BibitemShut {NoStop}%
\bibitem [{\citenamefont {Calderbank}\ and\ \citenamefont {Shor}(1996)}]{calderbank1996good}%
  \BibitemOpen
  \bibfield  {author} {\bibinfo {author} {\bibfnamefont {A.~R.}\ \bibnamefont {Calderbank}}\ and\ \bibinfo {author} {\bibfnamefont {P.~W.}\ \bibnamefont {Shor}},\ }\bibfield  {title} {\bibinfo {title} {Good quantum error-correcting codes exist},\ }\href {https://doi.org/10.1103/PhysRevA.54.1098} {\bibfield  {journal} {\bibinfo  {journal} {Phys. Rev. A}\ }\textbf {\bibinfo {volume} {54}},\ \bibinfo {pages} {1098} (\bibinfo {year} {1996})}\BibitemShut {NoStop}%
\bibitem [{nielsen2000quantum()}]{nielsen2000quantum}%
  \BibitemOpen
  \bibinfo {note} {M. A. Nielsen and I. Chuang, {\it Quantum Computation and Quantum Information}. (Cambridge University Press, Cambridge, 2000)}\BibitemShut {NoStop}%
\bibitem [{\citenamefont {Eastin}\ and\ \citenamefont {Knill}(2009)}]{eastin2009restrictions}%
  \BibitemOpen
  \bibfield  {author} {\bibinfo {author} {\bibfnamefont {B.}~\bibnamefont {Eastin}}\ and\ \bibinfo {author} {\bibfnamefont {E.}~\bibnamefont {Knill}},\ }\bibfield  {title} {\bibinfo {title} {Restrictions on transversal encoded quantum gate sets},\ }\href {https://doi.org/10.1103/PhysRevLett.102.110502} {\bibfield  {journal} {\bibinfo  {journal} {Phys. Rev. Lett.}\ }\textbf {\bibinfo {volume} {102}},\ \bibinfo {pages} {110502} (\bibinfo {year} {2009})}\BibitemShut {NoStop}%
\bibitem [{\citenamefont {Wen}(1995)}]{wen1995topological}%
  \BibitemOpen
  \bibfield  {author} {\bibinfo {author} {\bibfnamefont {X.-G.}\ \bibnamefont {Wen}},\ }\bibfield  {title} {\bibinfo {title} {Topological orders and edge excitations in fractional quantum hall states},\ }\href {https://doi.org/10.1080/00018739500101566} {\bibfield  {journal} {\bibinfo  {journal} {Advances in Physics}\ }\textbf {\bibinfo {volume} {44}},\ \bibinfo {pages} {405} (\bibinfo {year} {1995})}\BibitemShut {NoStop}%
\bibitem [{\citenamefont {Kane}\ and\ \citenamefont {Mele}(2005)}]{kane2005quantum}%
  \BibitemOpen
  \bibfield  {author} {\bibinfo {author} {\bibfnamefont {C.~L.}\ \bibnamefont {Kane}}\ and\ \bibinfo {author} {\bibfnamefont {E.~J.}\ \bibnamefont {Mele}},\ }\bibfield  {title} {\bibinfo {title} {Quantum spin hall effect in graphene},\ }\href {https://doi.org/10.1103/PhysRevLett.95.226801} {\bibfield  {journal} {\bibinfo  {journal} {Phys. Rev. Lett.}\ }\textbf {\bibinfo {volume} {95}},\ \bibinfo {pages} {226801} (\bibinfo {year} {2005})}\BibitemShut {NoStop}%
\bibitem [{\citenamefont {Hasan}\ and\ \citenamefont {Kane}(2010)}]{hasan2010colloquium}%
  \BibitemOpen
  \bibfield  {author} {\bibinfo {author} {\bibfnamefont {M.~Z.}\ \bibnamefont {Hasan}}\ and\ \bibinfo {author} {\bibfnamefont {C.~L.}\ \bibnamefont {Kane}},\ }\bibfield  {title} {\bibinfo {title} {Colloquium: Topological insulators},\ }\href {https://doi.org/10.1103/RevModPhys.82.3045} {\bibfield  {journal} {\bibinfo  {journal} {Rev. Mod. Phys.}\ }\textbf {\bibinfo {volume} {82}},\ \bibinfo {pages} {3045} (\bibinfo {year} {2010})}\BibitemShut {NoStop}%
\bibitem [{\citenamefont {Sato}\ and\ \citenamefont {Ando}(2017)}]{sato2017topological}%
  \BibitemOpen
  \bibfield  {author} {\bibinfo {author} {\bibfnamefont {M.}~\bibnamefont {Sato}}\ and\ \bibinfo {author} {\bibfnamefont {Y.}~\bibnamefont {Ando}},\ }\bibfield  {title} {\bibinfo {title} {Topological superconductors: a review},\ }\href {https://doi.org/10.1088/1361-6633/aa6ac7} {\bibfield  {journal} {\bibinfo  {journal} {Rep. Prog. Phys.}\ }\textbf {\bibinfo {volume} {80}},\ \bibinfo {pages} {076501} (\bibinfo {year} {2017})}\BibitemShut {NoStop}%
\bibitem [{\citenamefont {Gu}\ and\ \citenamefont {Wen}(2009)}]{gu209tensor}%
  \BibitemOpen
  \bibfield  {author} {\bibinfo {author} {\bibfnamefont {Z.-C.}\ \bibnamefont {Gu}}\ and\ \bibinfo {author} {\bibfnamefont {X.-G.}\ \bibnamefont {Wen}},\ }\bibfield  {title} {\bibinfo {title} {Tensor-entanglement-filtering renormalization approach and symmetry-protected topological order},\ }\href {https://doi.org/10.1103/PhysRevB.80.155131} {\bibfield  {journal} {\bibinfo  {journal} {Phys. Rev. B}\ }\textbf {\bibinfo {volume} {80}},\ \bibinfo {pages} {155131} (\bibinfo {year} {2009})}\BibitemShut {NoStop}%
\bibitem [{\citenamefont {Pollmann}\ \emph {et~al.}(2010)\citenamefont {Pollmann}, \citenamefont {Turner}, \citenamefont {Berg},\ and\ \citenamefont {Oshikawa}}]{pollman2010entanglement}%
  \BibitemOpen
  \bibfield  {author} {\bibinfo {author} {\bibfnamefont {F.}~\bibnamefont {Pollmann}}, \bibinfo {author} {\bibfnamefont {A.~M.}\ \bibnamefont {Turner}}, \bibinfo {author} {\bibfnamefont {E.}~\bibnamefont {Berg}},\ and\ \bibinfo {author} {\bibfnamefont {M.}~\bibnamefont {Oshikawa}},\ }\bibfield  {title} {\bibinfo {title} {Entanglement spectrum of a topological phase in one dimension},\ }\href {https://doi.org/10.1103/PhysRevB.81.064439} {\bibfield  {journal} {\bibinfo  {journal} {Phys. Rev. B}\ }\textbf {\bibinfo {volume} {81}},\ \bibinfo {pages} {064439} (\bibinfo {year} {2010})}\BibitemShut {NoStop}%
\bibitem [{\citenamefont {Chen}\ \emph {et~al.}(2012)\citenamefont {Chen}, \citenamefont {Gu}, \citenamefont {Liu},\ and\ \citenamefont {Wen}}]{chen2012science}%
  \BibitemOpen
  \bibfield  {author} {\bibinfo {author} {\bibfnamefont {X.}~\bibnamefont {Chen}}, \bibinfo {author} {\bibfnamefont {Z.-C.}\ \bibnamefont {Gu}}, \bibinfo {author} {\bibfnamefont {Z.-X.}\ \bibnamefont {Liu}},\ and\ \bibinfo {author} {\bibfnamefont {X.-G.}\ \bibnamefont {Wen}},\ }\bibfield  {title} {\bibinfo {title} {Symmetry-protected topological orders in interacting bosonic systems},\ }\href {https://doi.org/10.1126/science.1227224} {\bibfield  {journal} {\bibinfo  {journal} {Science}\ }\textbf {\bibinfo {volume} {338}},\ \bibinfo {pages} {1604} (\bibinfo {year} {2012})}\BibitemShut {NoStop}%
\bibitem [{\citenamefont {Bao}\ \emph {et~al.}(2021)\citenamefont {Bao}, \citenamefont {Choi},\ and\ \citenamefont {Altman}}]{bao2021symmetry}%
  \BibitemOpen
  \bibfield  {author} {\bibinfo {author} {\bibfnamefont {Y.}~\bibnamefont {Bao}}, \bibinfo {author} {\bibfnamefont {S.}~\bibnamefont {Choi}},\ and\ \bibinfo {author} {\bibfnamefont {E.}~\bibnamefont {Altman}},\ }\bibfield  {title} {\bibinfo {title} {Symmetry enriched phases of quantum circuits},\ }\href {https://doi.org/10.1016/j.aop.2021.168618} {\bibfield  {journal} {\bibinfo  {journal} {Annals of Physics}\ }\textbf {\bibinfo {volume} {435}},\ \bibinfo {pages} {168618} (\bibinfo {year} {2021})}\BibitemShut {NoStop}%
\bibitem [{\citenamefont {Lavasani}\ \emph {et~al.}(2021)\citenamefont {Lavasani}, \citenamefont {Alavirad},\ and\ \citenamefont {Barkeshli}}]{lavasani2021measurement}%
  \BibitemOpen
  \bibfield  {author} {\bibinfo {author} {\bibfnamefont {A.}~\bibnamefont {Lavasani}}, \bibinfo {author} {\bibfnamefont {Y.}~\bibnamefont {Alavirad}},\ and\ \bibinfo {author} {\bibfnamefont {M.}~\bibnamefont {Barkeshli}},\ }\bibfield  {title} {\bibinfo {title} {Measurement-induced topological entanglement transitions in symmetric random quantum circuits},\ }\href {https://doi.org/10.1038/s41567-020-01112-z} {\bibfield  {journal} {\bibinfo  {journal} {Nature Phys.}\ }\textbf {\bibinfo {volume} {17}},\ \bibinfo {pages} {342} (\bibinfo {year} {2021})}\BibitemShut {NoStop}%
\bibitem [{\citenamefont {Morral-Yepes}\ \emph {et~al.}(2023)\citenamefont {Morral-Yepes}, \citenamefont {Pollmann},\ and\ \citenamefont {Lovas}}]{morral-yepes2023detecting}%
  \BibitemOpen
  \bibfield  {author} {\bibinfo {author} {\bibfnamefont {R.}~\bibnamefont {Morral-Yepes}}, \bibinfo {author} {\bibfnamefont {F.}~\bibnamefont {Pollmann}},\ and\ \bibinfo {author} {\bibfnamefont {I.}~\bibnamefont {Lovas}},\ }\bibfield  {title} {\bibinfo {title} {Detecting and stabilizing measurement-induced symmetry-protected topological phases in generalized cluster models},\ }\href {https://doi.org/10.1103/PhysRevB.108.224304} {\bibfield  {journal} {\bibinfo  {journal} {Phys. Rev. B}\ }\textbf {\bibinfo {volume} {108}},\ \bibinfo {pages} {224304} (\bibinfo {year} {2023})}\BibitemShut {NoStop}%
\bibitem [{\citenamefont {Hauser}\ \emph {et~al.}(2024)\citenamefont {Hauser}, \citenamefont {Li}, \citenamefont {Vijay},\ and\ \citenamefont {Fisher}}]{hauser2024continuous}%
  \BibitemOpen
  \bibfield  {author} {\bibinfo {author} {\bibfnamefont {J.}~\bibnamefont {Hauser}}, \bibinfo {author} {\bibfnamefont {Y.}~\bibnamefont {Li}}, \bibinfo {author} {\bibfnamefont {S.}~\bibnamefont {Vijay}},\ and\ \bibinfo {author} {\bibfnamefont {M.~P.~A.}\ \bibnamefont {Fisher}},\ }\bibfield  {title} {\bibinfo {title} {Continuous symmetry breaking in adaptive quantum dynamics},\ }\href {https://doi.org/10.1103/PhysRevB.109.214305} {\bibfield  {journal} {\bibinfo  {journal} {Phys. Rev. B}\ }\textbf {\bibinfo {volume} {109}},\ \bibinfo {pages} {214305} (\bibinfo {year} {2024})}\BibitemShut {NoStop}%
\bibitem [{\citenamefont {Kitaev}(1997)}]{kitaev1997quantum}%
  \BibitemOpen
  \bibfield  {author} {\bibinfo {author} {\bibfnamefont {A.~Y.}\ \bibnamefont {Kitaev}},\ }\bibfield  {title} {\bibinfo {title} {Quantum computations: algorithms and error correction},\ }\href {https://doi.org/10.1070/RM1997v052n06ABEH002155} {\bibfield  {journal} {\bibinfo  {journal} {Russ. Math. Surv.}\ }\textbf {\bibinfo {volume} {52}},\ \bibinfo {pages} {1191} (\bibinfo {year} {1997})}\BibitemShut {NoStop}%
\bibitem [{\citenamefont {Dawson}\ and\ \citenamefont {Nielsen}(2006)}]{dawson2006solovay}%
  \BibitemOpen
  \bibfield  {author} {\bibinfo {author} {\bibfnamefont {C.~M.}\ \bibnamefont {Dawson}}\ and\ \bibinfo {author} {\bibfnamefont {M.~A.}\ \bibnamefont {Nielsen}},\ }\bibfield  {title} {\bibinfo {title} {The solovay-kitaev algorithm},\ }\href {https://doi.org/10.5555/2011679.2011685} {\bibfield  {journal} {\bibinfo  {journal} {Quantum Info. Comput.}\ }\textbf {\bibinfo {volume} {6}},\ \bibinfo {pages} {81} (\bibinfo {year} {2006})}\BibitemShut {NoStop}%
\bibitem [{\citenamefont {DiVincenzo}(1995)}]{divincenzo1995universal}%
  \BibitemOpen
  \bibfield  {author} {\bibinfo {author} {\bibfnamefont {D.~P.}\ \bibnamefont {DiVincenzo}},\ }\bibfield  {title} {\bibinfo {title} {Two-bit gates are universal for quantum computation},\ }\href {https://doi.org/10.1103/PhysRevA.51.1015} {\bibfield  {journal} {\bibinfo  {journal} {Phys. Rev. A}\ }\textbf {\bibinfo {volume} {51}},\ \bibinfo {pages} {1015} (\bibinfo {year} {1995})}\BibitemShut {NoStop}%
\bibitem [{\citenamefont {Lloyd}(1995)}]{lloyd1995almost}%
  \BibitemOpen
  \bibfield  {author} {\bibinfo {author} {\bibfnamefont {S.}~\bibnamefont {Lloyd}},\ }\bibfield  {title} {\bibinfo {title} {Almost any quantum logic gate is universal},\ }\href {https://doi.org/10.1103/PhysRevLett.75.346} {\bibfield  {journal} {\bibinfo  {journal} {Phys. Rev. Lett.}\ }\textbf {\bibinfo {volume} {75}},\ \bibinfo {pages} {346} (\bibinfo {year} {1995})}\BibitemShut {NoStop}%
\bibitem [{\citenamefont {Marvian}(2022)}]{marvian2022restriction}%
  \BibitemOpen
  \bibfield  {author} {\bibinfo {author} {\bibfnamefont {I.}~\bibnamefont {Marvian}},\ }\bibfield  {title} {\bibinfo {title} {Restrictions on realizable unitary operations imposed by symmetry and locality},\ }\href {https://doi.org/10.1038/s41567-021-01464-0} {\bibfield  {journal} {\bibinfo  {journal} {Nature Phys.}\ }\textbf {\bibinfo {volume} {18}},\ \bibinfo {pages} {283} (\bibinfo {year} {2022})}\BibitemShut {NoStop}%
\bibitem [{\citenamefont {Marvian}(2024)}]{marvian2023theory}%
  \BibitemOpen
  \bibfield  {author} {\bibinfo {author} {\bibfnamefont {I.}~\bibnamefont {Marvian}},\ }\bibfield  {title} {\bibinfo {title} {Theory of quantum circuits with abelian symmetries},\ }\href {https://doi.org/10.1103/PhysRevResearch.6.043292} {\bibfield  {journal} {\bibinfo  {journal} {Phys. Rev. Res.}\ }\textbf {\bibinfo {volume} {6}},\ \bibinfo {pages} {043292} (\bibinfo {year} {2024})}\BibitemShut {NoStop}%
\bibitem [{\citenamefont {Marvian}\ \emph {et~al.}(2024)\citenamefont {Marvian}, \citenamefont {Liu},\ and\ \citenamefont {Hulse}}]{marvian2024rotationally}%
  \BibitemOpen
  \bibfield  {author} {\bibinfo {author} {\bibfnamefont {I.}~\bibnamefont {Marvian}}, \bibinfo {author} {\bibfnamefont {H.}~\bibnamefont {Liu}},\ and\ \bibinfo {author} {\bibfnamefont {A.}~\bibnamefont {Hulse}},\ }\bibfield  {title} {\bibinfo {title} {Rotationally invariant circuits: Universality with the exchange interaction and two ancilla qubits},\ }\href {https://doi.org/10.1103/PhysRevLett.132.130201} {\bibfield  {journal} {\bibinfo  {journal} {Phys. Rev. Lett.}\ }\textbf {\bibinfo {volume} {132}},\ \bibinfo {pages} {130201} (\bibinfo {year} {2024})}\BibitemShut {NoStop}%
\bibitem [{\citenamefont {Hulse}\ \emph {et~al.}(2024{\natexlab{a}})\citenamefont {Hulse}, \citenamefont {Liu},\ and\ \citenamefont {Marvian}}]{hulse2024framework}%
  \BibitemOpen
  \bibfield  {author} {\bibinfo {author} {\bibfnamefont {A.}~\bibnamefont {Hulse}}, \bibinfo {author} {\bibfnamefont {H.}~\bibnamefont {Liu}},\ and\ \bibinfo {author} {\bibfnamefont {I.}~\bibnamefont {Marvian}},\ }\bibfield  {title} {\bibinfo {title} {A framework for semi-universality: Semi-universality of 3-qudit su(d)-invariant gates},\ }\href {https://doi.org/10.48550/arXiv.2407.21249} {\bibfield  {journal} {\bibinfo  {journal} {arXiv preprint arXiv:2407.21249}\ } (\bibinfo {year} {2024}{\natexlab{a}})}\BibitemShut {NoStop}%
\bibitem [{\citenamefont {Kempe}\ \emph {et~al.}(2001)\citenamefont {Kempe}, \citenamefont {Bacon}, \citenamefont {DiVincenzo},\ and\ \citenamefont {Whaley}}]{kempe2001encoded}%
  \BibitemOpen
  \bibfield  {author} {\bibinfo {author} {\bibfnamefont {J.}~\bibnamefont {Kempe}}, \bibinfo {author} {\bibfnamefont {D.}~\bibnamefont {Bacon}}, \bibinfo {author} {\bibfnamefont {D.~P.}\ \bibnamefont {DiVincenzo}},\ and\ \bibinfo {author} {\bibfnamefont {K.~B.}\ \bibnamefont {Whaley}},\ }\bibfield  {title} {\bibinfo {title} {Encoded universality from a single physical interaction},\ }\href {https://dl.acm.org/doi/10.5555/2016994.2017000} {\bibfield  {journal} {\bibinfo  {journal} {Quantum Info. Comput.}\ }\textbf {\bibinfo {volume} {1}},\ \bibinfo {pages} {33} (\bibinfo {year} {2001})}\BibitemShut {NoStop}%
\bibitem [{\citenamefont {Dankert}\ \emph {et~al.}(2009)\citenamefont {Dankert}, \citenamefont {Cleve}, \citenamefont {Emerson},\ and\ \citenamefont {Livine}}]{dankert2009exact}%
  \BibitemOpen
  \bibfield  {author} {\bibinfo {author} {\bibfnamefont {C.}~\bibnamefont {Dankert}}, \bibinfo {author} {\bibfnamefont {R.}~\bibnamefont {Cleve}}, \bibinfo {author} {\bibfnamefont {J.}~\bibnamefont {Emerson}},\ and\ \bibinfo {author} {\bibfnamefont {E.}~\bibnamefont {Livine}},\ }\bibfield  {title} {\bibinfo {title} {Exact and approximate unitary 2-designs and their application to fidelity estimation},\ }\href {https://doi.org/10.1103/PhysRevA.80.012304} {\bibfield  {journal} {\bibinfo  {journal} {Phys. Rev. A}\ }\textbf {\bibinfo {volume} {80}},\ \bibinfo {pages} {012304} (\bibinfo {year} {2009})}\BibitemShut {NoStop}%
\bibitem [{\citenamefont {Boixo}\ \emph {et~al.}(2018)\citenamefont {Boixo}, \citenamefont {Isakov}, \citenamefont {Smelyanskiy}, \citenamefont {Babbush}, \citenamefont {Ding}, \citenamefont {Jiang}, \citenamefont {Bremner}, \citenamefont {Martinis},\ and\ \citenamefont {Neven}}]{boixo2018characterizing}%
  \BibitemOpen
  \bibfield  {author} {\bibinfo {author} {\bibfnamefont {S.}~\bibnamefont {Boixo}}, \bibinfo {author} {\bibfnamefont {S.~V.}\ \bibnamefont {Isakov}}, \bibinfo {author} {\bibfnamefont {V.~N.}\ \bibnamefont {Smelyanskiy}}, \bibinfo {author} {\bibfnamefont {R.}~\bibnamefont {Babbush}}, \bibinfo {author} {\bibfnamefont {N.}~\bibnamefont {Ding}}, \bibinfo {author} {\bibfnamefont {Z.}~\bibnamefont {Jiang}}, \bibinfo {author} {\bibfnamefont {M.~J.}\ \bibnamefont {Bremner}}, \bibinfo {author} {\bibfnamefont {J.~M.}\ \bibnamefont {Martinis}},\ and\ \bibinfo {author} {\bibfnamefont {H.}~\bibnamefont {Neven}},\ }\bibfield  {title} {\bibinfo {title} {Characterizing quantum supremacy in near-term devices},\ }\href {https://doi.org/10.1038/s41567-018-0124-x} {\bibfield  {journal} {\bibinfo  {journal} {Nature Phys.}\ }\textbf {\bibinfo {volume} {14}},\ \bibinfo {pages} {595} (\bibinfo {year} {2018})}\BibitemShut {NoStop}%
\bibitem [{\citenamefont {Arute}\ \emph {et~al.}(2019)\citenamefont {Arute}, \citenamefont {Arya}, \citenamefont {Babbush}, \citenamefont {Bacon}, \citenamefont {Bardin}, \citenamefont {Barends}, \citenamefont {Biswas}, \citenamefont {Boixo}, \citenamefont {Brandao}, \citenamefont {Buell}, \citenamefont {Burkett}, \citenamefont {Chen}, \citenamefont {Chen}, \citenamefont {Chiaro}, \citenamefont {Collins}, \citenamefont {Courtney}, \citenamefont {Dunsworth}, \citenamefont {Farhi}, \citenamefont {Foxen}, \citenamefont {Fowler}, \citenamefont {Gidney}, \citenamefont {Giustina}, \citenamefont {Graff}, \citenamefont {Guerin}, \citenamefont {Habegger}, \citenamefont {Harrigan}, \citenamefont {Hartmann}, \citenamefont {Ho}, \citenamefont {Hoffmann}, \citenamefont {Huang}, \citenamefont {Humble}, \citenamefont {Isakov}, \citenamefont {Jeffrey}, \citenamefont {Jiang}, \citenamefont {Kafri}, \citenamefont {Kechedzhi}, \citenamefont {Kelly}, \citenamefont {Klimov}, \citenamefont {Knysh}, \citenamefont {Korotkov},
  \citenamefont {Kostritsa}, \citenamefont {Landhuis}, \citenamefont {Lindmark}, \citenamefont {Lucero}, \citenamefont {Lyakh}, \citenamefont {Mandr^^c3^^a0}, \citenamefont {McClean}, \citenamefont {McEwen}, \citenamefont {Megrant}, \citenamefont {Mi}, \citenamefont {Michielsen}, \citenamefont {Mohseni}, \citenamefont {Mutus}, \citenamefont {Naaman}, \citenamefont {Neeley}, \citenamefont {Neill}, \citenamefont {Niu}, \citenamefont {Ostby}, \citenamefont {Petukhov}, \citenamefont {Platt}, \citenamefont {Quintana}, \citenamefont {Rieffel}, \citenamefont {Roushan}, \citenamefont {Rubin}, \citenamefont {Sank}, \citenamefont {Satzinger}, \citenamefont {Smelyanskiy}, \citenamefont {Sung}, \citenamefont {Trevithick}, \citenamefont {Vainsencher}, \citenamefont {Villalonga}, \citenamefont {White}, \citenamefont {Yao}, \citenamefont {Yeh}, \citenamefont {Zalcman}, \citenamefont {Neven},\ and\ \citenamefont {Martinis}}]{arute2019quantuma}%
  \BibitemOpen
  \bibfield  {author} {\bibinfo {author} {\bibfnamefont {F.}~\bibnamefont {Arute}}, \bibinfo {author} {\bibfnamefont {K.}~\bibnamefont {Arya}}, \bibinfo {author} {\bibfnamefont {R.}~\bibnamefont {Babbush}}, \bibinfo {author} {\bibfnamefont {D.}~\bibnamefont {Bacon}}, \bibinfo {author} {\bibfnamefont {J.~C.}\ \bibnamefont {Bardin}}, \bibinfo {author} {\bibfnamefont {R.}~\bibnamefont {Barends}}, \bibinfo {author} {\bibfnamefont {R.}~\bibnamefont {Biswas}}, \bibinfo {author} {\bibfnamefont {S.}~\bibnamefont {Boixo}}, \bibinfo {author} {\bibfnamefont {F.~G. S.~L.}\ \bibnamefont {Brandao}}, \bibinfo {author} {\bibfnamefont {D.~A.}\ \bibnamefont {Buell}}, \bibinfo {author} {\bibfnamefont {B.}~\bibnamefont {Burkett}}, \bibinfo {author} {\bibfnamefont {Y.}~\bibnamefont {Chen}}, \bibinfo {author} {\bibfnamefont {Z.}~\bibnamefont {Chen}}, \bibinfo {author} {\bibfnamefont {B.}~\bibnamefont {Chiaro}}, \bibinfo {author} {\bibfnamefont {R.}~\bibnamefont {Collins}}, \bibinfo {author} {\bibfnamefont {W.}~\bibnamefont
  {Courtney}}, \bibinfo {author} {\bibfnamefont {A.}~\bibnamefont {Dunsworth}}, \bibinfo {author} {\bibfnamefont {E.}~\bibnamefont {Farhi}}, \bibinfo {author} {\bibfnamefont {B.}~\bibnamefont {Foxen}}, \bibinfo {author} {\bibfnamefont {A.}~\bibnamefont {Fowler}}, \bibinfo {author} {\bibfnamefont {C.}~\bibnamefont {Gidney}}, \bibinfo {author} {\bibfnamefont {M.}~\bibnamefont {Giustina}}, \bibinfo {author} {\bibfnamefont {R.}~\bibnamefont {Graff}}, \bibinfo {author} {\bibfnamefont {K.}~\bibnamefont {Guerin}}, \bibinfo {author} {\bibfnamefont {S.}~\bibnamefont {Habegger}}, \bibinfo {author} {\bibfnamefont {M.~P.}\ \bibnamefont {Harrigan}}, \bibinfo {author} {\bibfnamefont {M.~J.}\ \bibnamefont {Hartmann}}, \bibinfo {author} {\bibfnamefont {A.}~\bibnamefont {Ho}}, \bibinfo {author} {\bibfnamefont {M.}~\bibnamefont {Hoffmann}}, \bibinfo {author} {\bibfnamefont {T.}~\bibnamefont {Huang}}, \bibinfo {author} {\bibfnamefont {T.~S.}\ \bibnamefont {Humble}}, \bibinfo {author} {\bibfnamefont {S.~V.}\ \bibnamefont
  {Isakov}}, \bibinfo {author} {\bibfnamefont {E.}~\bibnamefont {Jeffrey}}, \bibinfo {author} {\bibfnamefont {Z.}~\bibnamefont {Jiang}}, \bibinfo {author} {\bibfnamefont {D.}~\bibnamefont {Kafri}}, \bibinfo {author} {\bibfnamefont {K.}~\bibnamefont {Kechedzhi}}, \bibinfo {author} {\bibfnamefont {J.}~\bibnamefont {Kelly}}, \bibinfo {author} {\bibfnamefont {P.~V.}\ \bibnamefont {Klimov}}, \bibinfo {author} {\bibfnamefont {S.}~\bibnamefont {Knysh}}, \bibinfo {author} {\bibfnamefont {A.}~\bibnamefont {Korotkov}}, \bibinfo {author} {\bibfnamefont {F.}~\bibnamefont {Kostritsa}}, \bibinfo {author} {\bibfnamefont {D.}~\bibnamefont {Landhuis}}, \bibinfo {author} {\bibfnamefont {M.}~\bibnamefont {Lindmark}}, \bibinfo {author} {\bibfnamefont {E.}~\bibnamefont {Lucero}}, \bibinfo {author} {\bibfnamefont {D.}~\bibnamefont {Lyakh}}, \bibinfo {author} {\bibfnamefont {S.}~\bibnamefont {Mandr^^c3^^a0}}, \bibinfo {author} {\bibfnamefont {J.~R.}\ \bibnamefont {McClean}}, \bibinfo {author} {\bibfnamefont {M.}~\bibnamefont
  {McEwen}}, \bibinfo {author} {\bibfnamefont {A.}~\bibnamefont {Megrant}}, \bibinfo {author} {\bibfnamefont {X.}~\bibnamefont {Mi}}, \bibinfo {author} {\bibfnamefont {K.}~\bibnamefont {Michielsen}}, \bibinfo {author} {\bibfnamefont {M.}~\bibnamefont {Mohseni}}, \bibinfo {author} {\bibfnamefont {J.}~\bibnamefont {Mutus}}, \bibinfo {author} {\bibfnamefont {O.}~\bibnamefont {Naaman}}, \bibinfo {author} {\bibfnamefont {M.}~\bibnamefont {Neeley}}, \bibinfo {author} {\bibfnamefont {C.}~\bibnamefont {Neill}}, \bibinfo {author} {\bibfnamefont {M.~Y.}\ \bibnamefont {Niu}}, \bibinfo {author} {\bibfnamefont {E.}~\bibnamefont {Ostby}}, \bibinfo {author} {\bibfnamefont {A.}~\bibnamefont {Petukhov}}, \bibinfo {author} {\bibfnamefont {J.~C.}\ \bibnamefont {Platt}}, \bibinfo {author} {\bibfnamefont {C.}~\bibnamefont {Quintana}}, \bibinfo {author} {\bibfnamefont {E.~G.}\ \bibnamefont {Rieffel}}, \bibinfo {author} {\bibfnamefont {P.}~\bibnamefont {Roushan}}, \bibinfo {author} {\bibfnamefont {N.~C.}\ \bibnamefont {Rubin}},
  \bibinfo {author} {\bibfnamefont {D.}~\bibnamefont {Sank}}, \bibinfo {author} {\bibfnamefont {K.~J.}\ \bibnamefont {Satzinger}}, \bibinfo {author} {\bibfnamefont {V.}~\bibnamefont {Smelyanskiy}}, \bibinfo {author} {\bibfnamefont {K.~J.}\ \bibnamefont {Sung}}, \bibinfo {author} {\bibfnamefont {M.~D.}\ \bibnamefont {Trevithick}}, \bibinfo {author} {\bibfnamefont {A.}~\bibnamefont {Vainsencher}}, \bibinfo {author} {\bibfnamefont {B.}~\bibnamefont {Villalonga}}, \bibinfo {author} {\bibfnamefont {T.}~\bibnamefont {White}}, \bibinfo {author} {\bibfnamefont {Z.~J.}\ \bibnamefont {Yao}}, \bibinfo {author} {\bibfnamefont {P.}~\bibnamefont {Yeh}}, \bibinfo {author} {\bibfnamefont {A.}~\bibnamefont {Zalcman}}, \bibinfo {author} {\bibfnamefont {H.}~\bibnamefont {Neven}},\ and\ \bibinfo {author} {\bibfnamefont {J.~M.}\ \bibnamefont {Martinis}},\ }\bibfield  {title} {\bibinfo {title} {Quantum supremacy using a programmable superconducting processor},\ }\href {https://doi.org/10.1038/s41586-019-1666-5} {\bibfield
  {journal} {\bibinfo  {journal} {Nature}\ }\textbf {\bibinfo {volume} {574}},\ \bibinfo {pages} {505} (\bibinfo {year} {2019})}\BibitemShut {NoStop}%
\bibitem [{\citenamefont {Huang}\ \emph {et~al.}(2020)\citenamefont {Huang}, \citenamefont {Kueng},\ and\ \citenamefont {Preskill}}]{huang2020predicting}%
  \BibitemOpen
  \bibfield  {author} {\bibinfo {author} {\bibfnamefont {H.-Y.}\ \bibnamefont {Huang}}, \bibinfo {author} {\bibfnamefont {R.}~\bibnamefont {Kueng}},\ and\ \bibinfo {author} {\bibfnamefont {J.}~\bibnamefont {Preskill}},\ }\bibfield  {title} {\bibinfo {title} {Predicting many properties of a quantum system from very few measurements},\ }\href {https://doi.org/10.1038/s41567-020-0932-7} {\bibfield  {journal} {\bibinfo  {journal} {Nature Phys.}\ }\textbf {\bibinfo {volume} {16}},\ \bibinfo {pages} {1050} (\bibinfo {year} {2020})}\BibitemShut {NoStop}%
\bibitem [{\citenamefont {Emerson}\ \emph {et~al.}(2005)\citenamefont {Emerson}, \citenamefont {Alicki},\ and\ \citenamefont {{\.Z}yczkowski}}]{emerson2005scalable}%
  \BibitemOpen
  \bibfield  {author} {\bibinfo {author} {\bibfnamefont {J.}~\bibnamefont {Emerson}}, \bibinfo {author} {\bibfnamefont {R.}~\bibnamefont {Alicki}},\ and\ \bibinfo {author} {\bibfnamefont {K.}~\bibnamefont {{\.Z}yczkowski}},\ }\bibfield  {title} {\bibinfo {title} {Scalable noise estimation with random unitary operators},\ }\href {https://doi.org/10.1088/1464-4266/7/10/021} {\bibfield  {journal} {\bibinfo  {journal} {J. Opt. B: Quantum Semiclass. Opt.}\ }\textbf {\bibinfo {volume} {7}},\ \bibinfo {pages} {S347} (\bibinfo {year} {2005})}\BibitemShut {NoStop}%
\bibitem [{\citenamefont {Dupuis}\ \emph {et~al.}(2014)\citenamefont {Dupuis}, \citenamefont {Berta}, \citenamefont {Wullschleger},\ and\ \citenamefont {Renner}}]{dupuis2014one}%
  \BibitemOpen
  \bibfield  {author} {\bibinfo {author} {\bibfnamefont {F.}~\bibnamefont {Dupuis}}, \bibinfo {author} {\bibfnamefont {M.}~\bibnamefont {Berta}}, \bibinfo {author} {\bibfnamefont {J.}~\bibnamefont {Wullschleger}},\ and\ \bibinfo {author} {\bibfnamefont {R.}~\bibnamefont {Renner}},\ }\bibfield  {title} {\bibinfo {title} {One-shot decoupling},\ }\href {https://doi.org/10.1007/s00220-014-1990-4} {\bibfield  {journal} {\bibinfo  {journal} {Commun. Math. Phys.}\ }\textbf {\bibinfo {volume} {328}},\ \bibinfo {pages} {251} (\bibinfo {year} {2014})}\BibitemShut {NoStop}%
\bibitem [{\citenamefont {Roberts}\ and\ \citenamefont {Yoshida}(2017)}]{roberts2017chaos}%
  \BibitemOpen
  \bibfield  {author} {\bibinfo {author} {\bibfnamefont {D.~A.}\ \bibnamefont {Roberts}}\ and\ \bibinfo {author} {\bibfnamefont {B.}~\bibnamefont {Yoshida}},\ }\bibfield  {title} {\bibinfo {title} {Chaos and complexity by design},\ }\href {https://doi.org/10.1007/JHEP04(2017)121} {\bibfield  {journal} {\bibinfo  {journal} {J. High Energ. Phys.}\ }\textbf {\bibinfo {volume} {2017}},\ \bibinfo {pages} {121}}\BibitemShut {NoStop}%
\bibitem [{\citenamefont {Brand\~{a}o}\ \emph {et~al.}(2016)\citenamefont {Brand\~{a}o}, \citenamefont {Harrow},\ and\ \citenamefont {Horodecki}}]{brandao2016local}%
  \BibitemOpen
  \bibfield  {author} {\bibinfo {author} {\bibfnamefont {F.~G. S.~L.}\ \bibnamefont {Brand\~{a}o}}, \bibinfo {author} {\bibfnamefont {A.~W.}\ \bibnamefont {Harrow}},\ and\ \bibinfo {author} {\bibfnamefont {M.}~\bibnamefont {Horodecki}},\ }\bibfield  {title} {\bibinfo {title} {Local random quantum circuits are approximate polynomial-designs},\ }\href {https://doi.org/10.1007/s00220-016-2706-8} {\bibfield  {journal} {\bibinfo  {journal} {Commun. Math. Phys.}\ }\textbf {\bibinfo {volume} {346}},\ \bibinfo {pages} {397} (\bibinfo {year} {2016})}\BibitemShut {NoStop}%
\bibitem [{\citenamefont {Haferkamp}(2022)}]{haferkamp2022random}%
  \BibitemOpen
  \bibfield  {author} {\bibinfo {author} {\bibfnamefont {J.}~\bibnamefont {Haferkamp}},\ }\bibfield  {title} {\bibinfo {title} {Random quantum circuits are approximate unitary {$t$}-designs in depth {$O\left(nt^{5+o(1)}\right)$}},\ }\href {https://doi.org/10.22331/q-2022-09-08-795} {\bibfield  {journal} {\bibinfo  {journal} {{Quantum}}\ }\textbf {\bibinfo {volume} {6}},\ \bibinfo {pages} {795} (\bibinfo {year} {2022})}\BibitemShut {NoStop}%
\bibitem [{\citenamefont {Mittal}\ and\ \citenamefont {Hunter-Jones}(2023)}]{mittal2023local}%
  \BibitemOpen
  \bibfield  {author} {\bibinfo {author} {\bibfnamefont {S.}~\bibnamefont {Mittal}}\ and\ \bibinfo {author} {\bibfnamefont {N.}~\bibnamefont {Hunter-Jones}},\ }\bibfield  {title} {\bibinfo {title} {Local random quantum circuits form approximate designs on arbitrary architectures},\ }\href {https://doi.org/10.48550/arXiv.2310.19355} {\bibfield  {journal} {\bibinfo  {journal} {arXiv preprint arXiv:2310.19355}\ } (\bibinfo {year} {2023})}\BibitemShut {NoStop}%
\bibitem [{\citenamefont {Belkin}\ \emph {et~al.}(2024)\citenamefont {Belkin}, \citenamefont {Allen}, \citenamefont {Ghosh}, \citenamefont {Kang}, \citenamefont {Lin}, \citenamefont {Sud}, \citenamefont {Chong}, \citenamefont {Fefferman},\ and\ \citenamefont {Clark}}]{belkin2024approximate}%
  \BibitemOpen
  \bibfield  {author} {\bibinfo {author} {\bibfnamefont {D.}~\bibnamefont {Belkin}}, \bibinfo {author} {\bibfnamefont {J.}~\bibnamefont {Allen}}, \bibinfo {author} {\bibfnamefont {S.}~\bibnamefont {Ghosh}}, \bibinfo {author} {\bibfnamefont {C.}~\bibnamefont {Kang}}, \bibinfo {author} {\bibfnamefont {S.}~\bibnamefont {Lin}}, \bibinfo {author} {\bibfnamefont {J.}~\bibnamefont {Sud}}, \bibinfo {author} {\bibfnamefont {F.}~\bibnamefont {Chong}}, \bibinfo {author} {\bibfnamefont {B.}~\bibnamefont {Fefferman}},\ and\ \bibinfo {author} {\bibfnamefont {B.~K.}\ \bibnamefont {Clark}},\ }\bibfield  {title} {\bibinfo {title} {Approximate t-designs in generic circuit architectures},\ }\href {https://doi.org/10.48550/arXiv.2310.19783} {\bibfield  {journal} {\bibinfo  {journal} {arXiv preprint arXiv:2310.19783}\ } (\bibinfo {year} {2024})}\BibitemShut {NoStop}%
\bibitem [{\citenamefont {Schuster}\ \emph {et~al.}(2024)\citenamefont {Schuster}, \citenamefont {Haferkamp},\ and\ \citenamefont {Huang}}]{schuster2024random}%
  \BibitemOpen
  \bibfield  {author} {\bibinfo {author} {\bibfnamefont {T.}~\bibnamefont {Schuster}}, \bibinfo {author} {\bibfnamefont {J.}~\bibnamefont {Haferkamp}},\ and\ \bibinfo {author} {\bibfnamefont {H.-Y.}\ \bibnamefont {Huang}},\ }\bibfield  {title} {\bibinfo {title} {Random unitaries in extremely low depth},\ }\href {http://arxiv.org/abs/2407.07754} {\bibfield  {journal} {\bibinfo  {journal} {arXiv preprint arXiv:2407.07754}\ } (\bibinfo {year} {2024})}\BibitemShut {NoStop}%
\bibitem [{\citenamefont {LaRacuente}\ and\ \citenamefont {Leditzky}(2024)}]{laracuente2024approximate}%
  \BibitemOpen
  \bibfield  {author} {\bibinfo {author} {\bibfnamefont {N.}~\bibnamefont {LaRacuente}}\ and\ \bibinfo {author} {\bibfnamefont {F.}~\bibnamefont {Leditzky}},\ }\bibfield  {title} {\bibinfo {title} {Approximate $t$-designs in generic circuit architectures},\ }\href {https://doi.org/10.48550/arXiv.2407.07876} {\bibfield  {journal} {\bibinfo  {journal} {arXiv preprint arXiv.2407.07876}\ } (\bibinfo {year} {2024})}\BibitemShut {NoStop}%
\bibitem [{\citenamefont {Li}\ \emph {et~al.}(2024)\citenamefont {Li}, \citenamefont {Zheng}, \citenamefont {Liu}, \citenamefont {Jiang},\ and\ \citenamefont {Liu}}]{li2023designs}%
  \BibitemOpen
  \bibfield  {author} {\bibinfo {author} {\bibfnamefont {Z.}~\bibnamefont {Li}}, \bibinfo {author} {\bibfnamefont {H.}~\bibnamefont {Zheng}}, \bibinfo {author} {\bibfnamefont {J.}~\bibnamefont {Liu}}, \bibinfo {author} {\bibfnamefont {L.}~\bibnamefont {Jiang}},\ and\ \bibinfo {author} {\bibfnamefont {Z.-W.}\ \bibnamefont {Liu}},\ }\bibfield  {title} {\bibinfo {title} {Designs from local random quantum circuits with $\mathrm{SU}(d)$ symmetry},\ }\href {https://doi.org/10.1103/PRXQuantum.5.040349} {\bibfield  {journal} {\bibinfo  {journal} {PRX Quantum}\ }\textbf {\bibinfo {volume} {5}},\ \bibinfo {pages} {040349} (\bibinfo {year} {2024})}\BibitemShut {NoStop}%
\bibitem [{\citenamefont {Hearth}\ \emph {et~al.}(2025)\citenamefont {Hearth}, \citenamefont {Flynn}, \citenamefont {Chandran},\ and\ \citenamefont {Laumann}}]{hearth2025unitary}%
  \BibitemOpen
  \bibfield  {author} {\bibinfo {author} {\bibfnamefont {S.~N.}\ \bibnamefont {Hearth}}, \bibinfo {author} {\bibfnamefont {M.~O.}\ \bibnamefont {Flynn}}, \bibinfo {author} {\bibfnamefont {A.}~\bibnamefont {Chandran}},\ and\ \bibinfo {author} {\bibfnamefont {C.~R.}\ \bibnamefont {Laumann}},\ }\bibfield  {title} {\bibinfo {title} {Unitary $k$-designs from random number-conserving quantum circuits},\ }\href {https://doi.org/10.1103/PhysRevX.15.021022} {\bibfield  {journal} {\bibinfo  {journal} {Phys. Rev. X}\ }\textbf {\bibinfo {volume} {15}},\ \bibinfo {pages} {021022} (\bibinfo {year} {2025})}\BibitemShut {NoStop}%
\bibitem [{\citenamefont {Mitsuhashi}\ \emph {et~al.}(2025)\citenamefont {Mitsuhashi}, \citenamefont {Suzuki}, \citenamefont {Soejima},\ and\ \citenamefont {Yoshioka}}]{mitsuhashi2025unitary}%
  \BibitemOpen
  \bibfield  {author} {\bibinfo {author} {\bibfnamefont {Y.}~\bibnamefont {Mitsuhashi}}, \bibinfo {author} {\bibfnamefont {R.}~\bibnamefont {Suzuki}}, \bibinfo {author} {\bibfnamefont {T.}~\bibnamefont {Soejima}},\ and\ \bibinfo {author} {\bibfnamefont {N.}~\bibnamefont {Yoshioka}},\ }\bibfield  {title} {\bibinfo {title} {Unitary designs of symmetric local random circuits},\ }\href {https://doi.org/10.1103/PhysRevLett.134.180404} {\bibfield  {journal} {\bibinfo  {journal} {Phys. Rev. Lett.}\ }\textbf {\bibinfo {volume} {134}},\ \bibinfo {pages} {180404} (\bibinfo {year} {2025})}\BibitemShut {NoStop}%
\bibitem [{\citenamefont {Zeier}\ and\ \citenamefont {Zimbor\'{a}s}(2015)}]{zeier2015squares}%
  \BibitemOpen
  \bibfield  {author} {\bibinfo {author} {\bibfnamefont {R.}~\bibnamefont {Zeier}}\ and\ \bibinfo {author} {\bibfnamefont {Z.}~\bibnamefont {Zimbor\'{a}s}},\ }\bibfield  {title} {\bibinfo {title} {On squares of representations of compact lie algebras},\ }\href {https://doi.org/10.1063/1.4928410} {\bibfield  {journal} {\bibinfo  {journal} {J. Math. Phys.}\ }\textbf {\bibinfo {volume} {56}},\ \bibinfo {pages} {081702} (\bibinfo {year} {2015})}\BibitemShut {NoStop}%
\bibitem [{\citenamefont {Chen}\ \emph {et~al.}(2024)\citenamefont {Chen}, \citenamefont {Haah}, \citenamefont {Haferkamp}, \citenamefont {Liu}, \citenamefont {Metger},\ and\ \citenamefont {Tan}}]{chen2024incompressibilityspectralgapsrandom}%
  \BibitemOpen
  \bibfield  {author} {\bibinfo {author} {\bibfnamefont {C.-F.}\ \bibnamefont {Chen}}, \bibinfo {author} {\bibfnamefont {J.}~\bibnamefont {Haah}}, \bibinfo {author} {\bibfnamefont {J.}~\bibnamefont {Haferkamp}}, \bibinfo {author} {\bibfnamefont {Y.}~\bibnamefont {Liu}}, \bibinfo {author} {\bibfnamefont {T.}~\bibnamefont {Metger}},\ and\ \bibinfo {author} {\bibfnamefont {X.}~\bibnamefont {Tan}},\ }\bibfield  {title} {\bibinfo {title} {Incompressibility and spectral gaps of random circuits},\ }\href {https://doi.org/10.48550/arXiv.2406.07478} {\bibfield  {journal} {\bibinfo  {journal} {arXiv preprint arXiv:2406.07478}\ } (\bibinfo {year} {2024})}\BibitemShut {NoStop}%
\bibitem [{\citenamefont {Hulse}\ \emph {et~al.}(2024{\natexlab{b}})\citenamefont {Hulse}, \citenamefont {Liu},\ and\ \citenamefont {Marvian}}]{hulse2024unitary}%
  \BibitemOpen
  \bibfield  {author} {\bibinfo {author} {\bibfnamefont {A.}~\bibnamefont {Hulse}}, \bibinfo {author} {\bibfnamefont {H.}~\bibnamefont {Liu}},\ and\ \bibinfo {author} {\bibfnamefont {I.}~\bibnamefont {Marvian}},\ }\bibfield  {title} {\bibinfo {title} {Unitary designs from random symmetric quantum circuits},\ }\href {https://doi.org/10.48550/arXiv.2408.14463} {\bibfield  {journal} {\bibinfo  {journal} {arXiv preprint arXiv:2408.14463}\ } (\bibinfo {year} {2024}{\natexlab{b}})}\BibitemShut {NoStop}%
\bibitem [{knapp2002lie()}]{knapp2002lie}%
  \BibitemOpen
  \bibinfo {note} {A. W. Knapp, {\it Lie Groups Beyond an Introduction}. 2nd ed. (Birkh{\"a}user, Boston, 2002)}\BibitemShut {NoStop}%
\end{thebibliography}%

\appendix

\section{Proofs of the theorems for the concrete symmetries (Theorems~\ref{SMthm:Z2}, \ref{SMthm:U1_general_locality}, \ref{SMthm:U1_small_locality}, \ref{SMthm:SU2_general_locality}, and \ref{SMthm:SU2_small_locality})} \label{SMsec:theorem_specific}

Before going into the concrete cases of symmetries, we prepare a simple useful lemma for general symmetries, which we use many times.

\begin{lemma} \label{SMlem:solution_bound}
    Let $n, t\in\mathbb{N}$, $R$ be a unitary representation of a group $G$, $\Lambda$ be the set of labels of irreducible representations appearing in the decomposition of $R$, $(x_\lambda)_{\lambda\in\Lambda}$ satisfy Eqs.~\eqref{SMeq:SMthm:general_cond1} and \eqref{SMeq:SMthm:general_cond2}, and $\Lambda'$ be an arbitrary subset of $\Lambda$. 
    Then, $|\sum_{\lambda\in\Lambda'} m_\lambda x_\lambda|\leq t$. 
\end{lemma}

\begin{proof}
    By the triangle inequality, we have 
    \begin{align}
        \sum_{\lambda\in\Lambda} m_\lambda |x_\lambda| 
        =\sum_{\lambda\in\Lambda'} m_\lambda |x_\lambda|+\sum_{\lambda\in\Lambda\backslash\Lambda'} m_\lambda |x_\lambda| 
        \geq \left|\sum_{\lambda\in\Lambda'} m_\lambda x_\lambda\right|+\left|\sum_{\lambda\in\Lambda\backslash\Lambda'} m_\lambda x_\lambda\right|. \label{SMeq:SMlem:solution_bound1}
    \end{align}
    By Eq.~\eqref{SMeq:SMthm:general_cond2}, we have 
    \begin{align}
        \sum_{\lambda\in\Lambda\backslash\Lambda'} m_\lambda x_\lambda 
        =-\sum_{\lambda\in\Lambda'} m_\lambda x_\lambda. \label{SMeq:SMlem:solution_bound2}
    \end{align}
    By plugging Eq.~\eqref{SMeq:SMlem:solution_bound2} into Eq.~\eqref{SMeq:SMlem:solution_bound1}, we get 
    \begin{align}
        \sum_{\lambda\in\Lambda} m_\lambda |x_\lambda| 
        \geq 2\left|\sum_{\lambda\in\Lambda'} m_\lambda x_\lambda\right|. \label{SMeq:SMlem:solution_bound3}
    \end{align}
    By Eqs.~\eqref{SMeq:SMthm:general_cond1} and \eqref{SMeq:SMlem:solution_bound3}, we get $\left|\sum_{\lambda\in\Lambda'} m_\lambda x_\lambda\right|\leq t$. 
\end{proof}

\subsection{$\mathbb{Z}_2$ symmetry} \label{SMsubsec:Z2}

In this section, we consider the representation $R$ of $\mathbb{Z}_2$ defined by Eq.~\eqref{SMeq:Z2_representation}. 
Since $\bigcup_{\gamma\in\Gamma} \mathcal{U}_{n, G, R}^\gamma$ is semi-universal for $\mathcal{U}_{n, G, R}$~\cite{marvian2023theory}, we can use Theorem~\ref{SMthm:general}. 
We note that the representation $R$ can be decomposed into two inequivalent irreducible representations $R_\lambda(g):=(-1)^{\lambda g}$ with $\lambda\in\Lambda:=\{0, 1\}$ with multiplicity $m_\lambda=2^{n-1}$, which corresponds to the spectral decomposition of $\mathrm{Z}^{\otimes n}$. 
By using these results, we can easily derive the maximal order of asymptotic unitary designs under the $\mathbb{Z}_2$ symmetry presented as Theorem~\ref{SMthm:Z2}. \\

\noindent
\textit{Proof of Theorem~\ref{SMthm:Z2}.} 
    Since $R$ can be written as $T^{\otimes n}$ with a single-qubit representation $T$, by using Theorem~\ref{SMthm:general}, we consider the condition on $t$ such that Eqs.~\eqref{SMeq:SMthm:general_cond1} and \eqref{SMeq:SMthm:general_cond4} do not have a nontrivial integer solution. 
    Since $m_\lambda$ is given by $2^{n-1}$, Eq.~\eqref{SMeq:SMthm:general_cond1} is equivalent to 
    \begin{align}
        2^{n-1}(|x_0|+|x_1|)\leq 2t. \label{SMeq:SMthm:Z2_1}
    \end{align}
    Since every $A'\in\mathfrak{u}_{n, G, R}^\gamma$ can be decomposed into the direct sum $A_0\oplus A_1$ with $A_\lambda$ acting on the eigenspace of $\mathrm{Z}^{\otimes k}$ with eigenvalue $(-1)^\lambda$, we have $f_\lambda(A)=2^{n-k}(\mathrm{tr}(A_0)+\mathrm{tr}(A_1))$, which does not depend in $\lambda$. 
    Thus Eq.~\eqref{SMeq:SMthm:general_cond4} is equivalent to 
    \begin{align}
        x_0+x_1=0. \label{SMeq:SMthm:Z2_2}
    \end{align}
    By plugging Eq.~\eqref{SMeq:SMthm:Z2_2} into Eq.~\eqref{SMeq:SMthm:Z2_1}, we get $2^n |x_0|\leq 2t$, which yields $2^{n-1}|x_0|\leq t$, which implies that Eqs.~\eqref{SMeq:SMthm:general_cond1} and \eqref{SMeq:SMthm:general_cond4} have no nontrivial integer solution if and only if $t<2^{n-1}$. 
\hfill $\square$

\subsection{U(1) symmetry} \label{SMsubsec:U1}

In this section, we consider the representation $R$ of $\mathrm{U}(1)$ defined by Eq.~\eqref{SMeq:U1_representation}. 
Similarly to the $\mathrm{Z}_2$ symmetry, $\bigcup_{\gamma\in\Gamma} \mathcal{U}_{n, G, R}^\gamma$ is semi-universal for $\mathcal{U}_{n, G, R}$~\cite{marvian2023theory}, and thus we can use Theorem~\ref{SMthm:general}.

First, we explicitly present the conditions of Eqs.~\eqref{SMeq:SMthm:general_cond1} and \eqref{SMeq:SMthm:general_cond4} in Theorem~\ref{SMthm:general}.

\begin{lemma} \label{SMlem:U1_cond}
    Let $n, k, t\in\mathbb{N}$, $2\leq k\leq n-1$, and $R$ be a unitary representation of $G=\mathrm{U}(1)$ defined by Eq.~\eqref{SMeq:U1_representation}. 
    Then, the distribution of the $(G, R)$-symmetric $k$-local random circuit is an asymptotic $(G, R)$-symmetric unitary $t$-design if and only if there exists no nontrivial integer solution $\bm{x}=(x_\lambda)_{\lambda=0, ..., n}\in\mathbb{Z}^{n+1}$ satisfying 
    \begin{align}
        &\sum_{\lambda=0}^n \binom{n}{\lambda}|x_\lambda|\leq 2t, \label{SMeq:SMlem:U1_cond1}\\
        &\sum_{\lambda=j}^{n-k+j} \binom{n-k}{\lambda-j} x_\lambda=0\ \forall j\in\{0, 1, ..., k\}. \label{SMeq:SMlem:U1_cond2}
    \end{align}
\end{lemma}

\begin{proof}
    Since the representation $R$ defined by Eq.~\eqref{SMeq:U1_representation} is the tensor product of representation of a single-qudit representation, Theorem~\ref{SMthm:general} implies that the condition for the distribution of $(G, R)$-symmetric $k$-local random circuits forming an asymptotic unitary $t$-design if and only if Eqs.~\eqref{SMeq:SMthm:general_cond1} and \eqref{SMeq:SMthm:general_cond4} have no nontrivial integer solution. 
    In the following, it is sufficient to show that Eqs.~\eqref{SMeq:SMthm:general_cond1} and \eqref{SMeq:SMthm:general_cond4} are equivalent to Eqs.~\eqref{SMeq:SMlem:U1_cond1} and \eqref{SMeq:SMlem:U1_cond2}, respectively. 
    The representation $R$ can be decomposed into $n+1$ inequivalent irreducible representations, which are given by $R_\lambda(e^{i\theta})=e^{i(n-2\lambda)\theta}$ for $\lambda\in\Lambda=\{0, 1, ..., n\}$. 
    Since each of the representations $R_\lambda$ corresponds to the eigenvalue of $\sum_{j=1}^n \mathrm{I}^{\otimes j-1}\otimes \mathrm{Z}\otimes \mathrm{I}^{\otimes n-j}$, we find that the multiplicity $m_\lambda$ of $R_\lambda$ is given by $\binom{n}{\lambda}$. 
    Thus Eq.~\eqref{SMeq:SMthm:general_cond1} is rewritten as Eq.~\eqref{SMeq:SMlem:U1_cond1}. 
    We fix $\gamma\in\Gamma$ and consider $(G, R)$-symmetric $k$-local operators that act nontrivially on some fixed $k$ qubits. 
    We note that such operators can be written as a linear combination of operators in the form of $A_j\otimes \mathrm{I}^{\otimes n-k}$ with $A_j$ acting only on the eigenspaces of the sum of the Pauli-Z operators on $k$ qubits with eigenvalues $k-2j$ for $j=0, 1, ..., k$ and $I$ acting on the rest $n-k$ qubits. 
    Then, we have 
    \begin{align}
        f_\lambda(A_j\otimes \mathrm{I}^{\otimes n-k})=\binom{n-k}{\lambda-j}\mathrm{tr}(A_j), 
    \end{align}
    and Eq.~\eqref{SMeq:SMthm:general_cond4} is rewritten as Eq.~\eqref{SMeq:SMlem:U1_cond2}. 
\end{proof}

In the following, we present the explicit condition on $t$ such that Eqs.~\eqref{SMeq:SMlem:U1_cond1} and \eqref{SMeq:SMlem:U1_cond2} in Lemma~\ref{SMlem:U1_cond} have no nontrivial integer solution. 
First, we give a sufficient condition in the following lemma.

\begin{lemma} \label{SMlem:U1_nontrivial_existence_general}
    Let $n, k\in\mathbb{N}$, $2\leq k\leq n-1$, and $R$ be a unitary representation of $G=\mathrm{U}(1)$ defined by Eq.~\eqref{SMeq:U1_representation}. 
    Then, Eqs.~\eqref{SMeq:SMlem:U1_cond1} and \eqref{SMeq:SMlem:U1_cond2} have a nontrivial integer solution $\bm{x}=(x_\lambda)_{\lambda=0, ..., n}\in\mathbb{Z}^{n+1}$ if 
    \begin{align}
        t\geq\frac{2^{\lfloor k/2\rfloor}}{\left\lceil \frac{k}{2}\right\rceil!}\prod_{\alpha=1}^{\lceil k/2\rceil} (n-k+2\alpha-1). \label{SMeq:SMlem:U1_nontrivial_existence_general0}
    \end{align}
\end{lemma}

\begin{proof}
    We define a vector $\bm{y}_{n, k}=(y_{n, k, \lambda})_{\lambda\in\Lambda}\in\mathbb{Z}^\Lambda$ by 
    \begin{align}
        y_{n, k, \lambda}:=\frac{(-1)^\lambda}{(n-k-1)!}\prod_{\alpha=1}^{n-k-1} \left(\lambda-\left\lceil \frac{k}{2}\right\rceil-\alpha\right), 
    \end{align}
    and we show that $\bm{y}_{n, k}$ is a nontrivial integer solution of Eqs.~\eqref{SMeq:SMlem:U1_cond1} and \eqref{SMeq:SMlem:U1_cond2}. 
    By the definition of $\bm{y}_{n, k}$, we have 
    \begin{align}
        \sum_{\lambda=0}^n \binom{n-k}{\lambda-j} y_{n, k, \lambda} 
        =&\sum_{\lambda=0}^n \binom{n-k}{\lambda-j} \frac{(-1)^\lambda}{(n-k-1)!}\prod_{\alpha=1}^{n-k-1} \left(\lambda-\left\lceil \frac{k}{2}\right\rceil-\alpha\right) \nonumber\\
        =&\sum_{\lambda=0}^n \binom{n-k}{\lambda-j} \frac{(-1)^\lambda}{(n-k-1)!}\left.\left(\frac{d}{dz}\right)^{n-k-1} z^{\lambda-\lceil k/2\rceil-1}\right|_{z=1} \nonumber\\
        =&\frac{(-1)^j}{(n-k-1)!}\left.\left(\frac{d}{dz}\right)^{n-k-1}z^{j-\lceil k/2\rceil-1} \sum_{\lambda=0}^n \binom{n-k}{\lambda-j}(-z)^{\lambda-j} \right|_{z=1} \nonumber\\
        =&\frac{(-1)^j}{(n-k-1)!}\left.\left(\frac{d}{dz}\right)^{n-k-1}z^{j-\lceil k/2\rceil-1}(1-z)^{n-k} \right|_{z=1} \nonumber\\
        =&0, \label{SMeq:SMlem:U1_nontrivial_existence_general1}
    \end{align}
    which implies that $\bm{y}_{n, k}$ is a nontrivial integer solution of Eq.~\eqref{SMeq:SMlem:U1_cond2}. 
    The definition of $\bm{y}_{n, k}$ also implies that 
    \begin{align}
        &\sum_{\lambda=0}^n \binom{n}{\lambda}|y_{n, k, \lambda}| \nonumber\\
        =&\sum_{\lambda=0}^n \binom{n}{\lambda}\frac{1}{(n-k-1)!}\prod_{\alpha=1}^{n-k-1} \left|\lambda-\left\lceil \frac{k}{2}\right\rceil-\alpha\right| \nonumber\\
        =&\sum_{\lambda=0}^{\lceil k/2\rceil} \binom{n}{\lambda}\frac{1}{(n-k-1)!}\prod_{\alpha=1}^{n-k-1} \left(\alpha+\left\lceil \frac{k}{2}\right\rceil-\lambda\right) 
        +\sum_{\lambda=n-k+\lceil k/2\rceil}^n \binom{n}{\lambda}\frac{1}{(n-k-1)!}\prod_{\alpha=1}^{n-k-1} \left(\lambda-\left\lceil \frac{k}{2}\right\rceil-\alpha\right) \nonumber\\
        =&\sum_{\lambda=0}^{\lceil k/2\rceil} \binom{n}{\lambda}\binom{n-k-1+\left\lceil \frac{k}{2}\right\rceil-\lambda}{n-k-1}+\sum_{\lambda=n-k+\lceil k/2\rceil}^n \binom{n}{\lambda}\binom{\lambda-\left\lceil \frac{k}{2}\right\rceil-1}{n-k-1}. \label{SMeq:SMlem:U1_nontrivial_existence_general2}
    \end{align}
    We note that 
    \begin{align}
        \sum_{\lambda=n-k+\lceil k/2\rceil}^n \binom{n}{\lambda}\binom{\lambda-\left\lceil \frac{k}{2}\right\rceil-1}{n-k-1}
        =\sum_{\lambda=n-\lfloor k/2\rfloor}^n \binom{n}{\lambda}\binom{\lambda-k+\left\lfloor \frac{k}{2}\right\rfloor-1}{n-k-1}
        =\sum_{\lambda=0}^{\lfloor k/2\rfloor} \binom{n}{\lambda}\binom{n-k-1+\left\lfloor \frac{k}{2}\right\rfloor-\lambda}{n-k-1}. \label{SMeq:SMlem:U1_nontrivial_existence_general3}
    \end{align}
    By plugging Eq.~\eqref{SMeq:SMlem:U1_nontrivial_existence_general3} into Eq.~\eqref{SMeq:SMlem:U1_nontrivial_existence_general2}, we get 
    \begin{align}
        \sum_{\lambda=0}^n \binom{n}{\lambda}|y_{n, k, \lambda}|
        =\sum_{\lambda=0}^{\lceil k/2\rceil} \binom{n}{\lambda}\binom{n-k-1+\left\lceil \frac{k}{2}\right\rceil-\lambda}{n-k-1}+\sum_{\lambda=0}^{\lfloor k/2\rfloor} \binom{n}{\lambda}\binom{n-k-1+\left\lfloor \frac{k}{2}\right\rfloor-\lambda}{n-k-1}. \label{SMeq:SMlem:U1_nontrivial_existence_general4}
    \end{align}
    We note that for any $j\geq 0$, 
    \begin{align}
        \sum_{\lambda=0}^j \binom{n}{\lambda}\binom{n-k-1+j-\lambda}{n-k-1}
        =\sum_{\lambda=0}^j \binom{n}{\lambda}\binom{n-k-1+j-\lambda}{j-\lambda} 
        =a_{n, k, j}, \label{SMeq:SMlem:U1_nontrivial_existence_general5}
    \end{align}
    where we define 
    \begin{align}
        a_{n, k, j}:=\sum_{\lambda=0}^j \binom{n}{j-\lambda}\binom{n-k+\lambda-1}{\lambda} \label{SMeq:a_nkj_def}
    \end{align}
    for $n, k, j\in\mathbb{Z}$ satisfying $0\leq k\leq n-1$ and $j\geq 0$.
    By applying Eq.~\eqref{SMeq:SMlem:U1_nontrivial_existence_general5} to Eq.~\eqref{SMeq:SMlem:U1_nontrivial_existence_general4}, we get 
    \begin{align}
         \sum_{\lambda=0}^n \binom{n}{\lambda}|y_{n, k, \lambda}| 
         =a_{n, k, \lceil k/2\rceil}+a_{n, k, \lfloor k/2\rfloor}
         =2\cdot\frac{2^{\lfloor k/2\rfloor}}{\left\lceil \frac{k}{2}\right\rceil!}\prod_{\alpha=1}^{\lceil k/2\rceil} (n-k+2\alpha-1), \label{SMeq:SMlem:U1_nontrivial_existence_general6}
    \end{align}
    where we used Lemma~\ref{SMlem:sequence_property2}. 
    Thus $\bm{y}_{n, k}$ is a nontrivial integer solution of Eqs.~\eqref{SMeq:SMlem:U1_cond1} and \eqref{SMeq:SMlem:U1_cond2} when $t$ satisfies Eq.~\eqref{SMeq:SMlem:U1_nontrivial_existence_general0}. 
\end{proof}

Next, we show that for sufficiently large $n$, the condition on $t$ in Lemma~\ref{SMlem:U1_nontrivial_existence_general} is a necessary condition for the equations in Lemma~\ref{SMlem:U1_cond} having no nontrivial integer solution.

\begin{lemma} \label{SMlem:U1_nontrivial_nonexistence_general}
    Let $n, k, t\in\mathbb{N}$, $2\leq k\leq n-1$, $R$ be a unitary representation of $G=\mathrm{U}(1)$ defined by Eq.~\eqref{SMeq:U1_representation}, and $b_{n, j}$ defined for $j\in\mathbb{N}$ by 
    \begin{align}
        b_{n, j}:=\frac{2^{\lfloor j/2\rfloor}}{\left\lceil \frac{j}{2}\right\rceil!}\prod_{\alpha=1}^{\lceil j/2\rceil} (n-j+2\alpha-1) \label{SMeq:U1_bound_def}
    \end{align}
    satisfy $b_{n, k}\leq\binom{n}{j}$ for all $j\in\{\lceil k/2\rceil+1, \lceil k/2\rceil+2, ..., n-\lceil k/2\rceil-1\}$, and when $k$ is odd, $b_{n, j}$ also satisfy $b_{n, k}\leq b_{n, k+1}$. 
    Then, Eqs.~\eqref{SMeq:SMlem:U1_cond1} and \eqref{SMeq:SMlem:U1_cond2} have no nontrivial integer solution $\bm{x}=(x_\lambda)_{\lambda=0, ..., n}\in\mathbb{Z}^{n+1}$ if and only if $t<b_{n, k}$. 
\end{lemma}

We note that the assumptions are satisfied when $n\geq 2^k$, as we show in Lemma~\ref{SMlem:U1_sufficient_cond}.

\begin{proof}
    We take arbitrary integer solution $\bm{x}$ of Eqs.~\eqref{SMeq:SMlem:U1_cond1} and \eqref{SMeq:SMlem:U1_cond2}, and show that $\bm{x}=\bm{0}$. 
    When $t<b_{n, k}$, Eq.~\eqref{SMeq:SMlem:U1_cond2} implies that 
    \begin{align}
        \sum_{\lambda=0}^n \binom{n}{\lambda}|x_\lambda|<2b_{n, k}. \label{SMeq:SMlem:U1_nontrivial_nonexistence_general1}
    \end{align}
    Before going into the detailed proof process, we note that the assumption of $b_{n, k}\leq b_{n, k+1}$ for odd $k$ implies that $n\geq k+2$, because when $k$ is odd and $n=k+1$, we have 
    \begin{align}
        b_{n, k}
        =\frac{2^{(k-1)/2}}{(\frac{k+1}{2})!}\prod_{\alpha=1}^{(k+1)/2} 2\alpha 
        =\frac{1}{2}\cdot\frac{2^{(k+1)/2}}{(\frac{k+1}{2})!}\left(\prod_{\alpha=2}^{(k+1)/2} 2\alpha\right)\cdot 2
        >\frac{2^{(k+1)/2}}{(\frac{k+1}{2})!}\left(\prod_{\alpha=2}^{(k+1)/2} (2\alpha-1)\right) 
        =b_{n, k+1}, 
    \end{align}
    which contradicts with $b_{n, k}\leq b_{n, k+1}$. 
    By using Lemma~\ref{SMlem:solution_bound} with $\Lambda'=\{ \lambda \}$ for $\lambda\in\{\lceil k/2\rceil+1, \lceil k/2\rceil+2, ..., n-\lceil k/2\rceil-1\}$, and the assumption of $b_{n, k}\leq\binom{n}{\lceil k/2\rceil+1}$, we have 
    \begin{align}
        x_\lambda=0\ \forall \lambda\in\{\lceil k/2\rceil+1, \lceil k/2\rceil+2, ..., n-\lceil k/2\rceil-1\}, \label{SMeq:SMlem:U1_nontrivial_nonexistence_general2}
    \end{align}
    which means that $\bm{x}$ is in the space orthogonal to the space spanned by $\{\bm{w}_j\}_{j=0}^k$ and $\{\bm{v}_j\}_{j=\lceil k/2\rceil+1}^{n-\lceil k/2\rceil-1}$, where $\bm{w}_j:=(\binom{n-k}{\lambda-j})_{\lambda\in\Lambda}$ for $j\in\{0, 1, ..., k\}$ and $\bm{v}_j:=(\delta_{\lambda, j})_{\lambda\in\Lambda}$ for $j\in\{\lceil k/2\rceil+1, \lceil k/2\rceil+2, ..., n-\lceil k/2\rceil-1\}$. 
    We note that these vectors are linearly independent, because we can show that $\sum_{j=0}^k p_j\bm{w}_j+\sum_{j=\lceil k/2\rceil+1}^{n-\lceil k/2\rceil-1} q_j\bm{v}_j=0$ implies that $p_j=0$ and $q_j=0$ for all $j$ by looking at the elements for $\lambda=0, 1, ..., \lceil k/2\rceil, n, n-1, ..., n-\lfloor k/2\rfloor+1, \lceil k/2\rceil+1, \lceil k/2\rceil+2, ..., n-\lceil k/2\rceil-1$ in order. 
    Since $n\geq k+1$ when $k$ is even and from $n\geq k+2$ when $k$ is odd, we have $\lceil k/2\rceil\leq n-\lceil k/2\rceil-1$, which implies that the size of the set $\{\lceil k/2\rceil+1, \lceil k/2\rceil+2, ..., n-\lceil k/2\rceil-1\}$ is given by $(n-\lceil k/2\rceil-1)-\lceil k/2\rceil=n-2\lceil k/2\rceil-1$. 
    Thus the dimension of the linear space of $\bm{x}$ satisfying Eqs.~\eqref{SMeq:SMlem:U1_cond1} and \eqref{SMeq:SMlem:U1_nontrivial_nonexistence_general2} is given by $(n+1)-[(k+1)+(n-2\lceil k/2\rceil-1)]=2\lceil k/2\rceil-k+1$, which is $1$ when $k$ is even and $2$ when $k$ is odd.

    When $k$ is even, by noting that $\bm{y}_{n, k}$ satisfies Eqs.~\eqref{SMeq:SMlem:U1_cond1} and \eqref{SMeq:SMlem:U1_nontrivial_nonexistence_general2}, $\bm{x}$ can be written as $\bm{x}=r\bm{y}_{n, k}$ with some $r\in\mathbb{R}$. 
    Since $\bm{x}$ is an integer vector and $y_{n, k, \lceil k/2\rceil}=(-1)^{n-k+\lceil k/2\rceil-1}$, we have $r=x_{\lceil k/2\rceil}/y_{n, k, \lceil k/2\rceil}\in\mathbb{Z}$. 
    By Eq.~\eqref{SMeq:SMlem:U1_nontrivial_existence_general6}, we have 
    \begin{align}
        \sum_{\lambda=0}^n \binom{n}{\lambda} |x_\lambda| 
        =\sum_{\lambda=0}^n \binom{n}{\lambda} |ry_{n, k, \lambda}| 
        =2|r|b_{n, k}. \label{SMeq:SMlem:U1_nontrivial_nonexistence_general3}
    \end{align}
    By plugging this into Eq.~\eqref{SMeq:SMlem:U1_nontrivial_nonexistence_general1}, we get $|r|<1$, which implies that $r=0$. 
    Thus Eqs.~\eqref{SMeq:SMlem:U1_cond1} and \eqref{SMeq:SMlem:U1_cond2} have no nontrivial integer solution.

    When $k$ is odd, to prepare a basis of the linear space of $\bm{x}$ satisfying Eqs.~\eqref{SMeq:SMlem:U1_cond1} and \eqref{SMeq:SMlem:U1_nontrivial_nonexistence_general2}, we define $\bm{y}'_{n, k}=(y'_{n, k, \lambda})_{\lambda\in\Lambda}$ by $y'_{n, k, \lambda}:=y_{n, k, n-\lambda}$. 
    We can see that $\bm{y}'_{n, k}$ is also a nontrivial integer solution of Eq.~\eqref{SMeq:SMlem:U1_cond2} by considering the transformation of $\lambda\mapsto n-\lambda$ in Eq.~\eqref{SMeq:SMlem:U1_nontrivial_existence_general1}. 
    We note that $\bm{y}_{n, k}$ and $\bm{y}'_{n, k}$ are linearly independent, which can be confirmed by $y_{n, k, \lceil k/2\rceil}=(-1)^{\lceil k/2\rceil+n-k-1}$, $y_{n, k, \lceil k/2\rceil+n-k-1}=0$, $y'_{n, k, \lceil k/2\rceil}=0$, and $y'_{n, k, \lceil k/2\rceil+n-k-1}=(-1)^{\lceil k/2\rceil}$. 
    Thus, $\bm{x}$ can be written as $\bm{x}=r\bm{y}_{n, k}+r'\bm{y}'_{n, k}$ with some $r, r'\in\mathbb{R}$. 
    Since $\bm{x}$ is an integer vector, we have $r=x_{\lceil k/2\rceil}/y_{n, k, \lceil k/2\rceil}\in\mathbb{Z}$ and $r'=x_{\lceil k/2\rceil+n-k-1}/y_{n, k, \lceil k/2\rceil+n-k-1}\in\mathbb{Z}$. 
    By Eq.~\eqref{SMeq:SMlem:U1_nontrivial_nonexistence_general2}, we have 
    \begin{align}
        \sum_{\lambda=0}^n \binom{n}{\lambda}|x_\lambda|
        =\sum_{\lambda=0}^{\lceil k/2\rceil} \binom{n}{\lambda}(|x_\lambda|+|x_{n-\lambda}|) 
        =\sum_{\lambda=0}^{\lceil k/2\rceil} \binom{n}{\lambda}(|r y_{n, k, \lambda}+r' y'_{n, k, \lambda}|+|r y'_{n, k, \lambda}+r' y_{n, k, \lambda}|), 
    \end{align}
    where we used $x_{n-\lambda}=r y_{n, k, n-\lambda}+r' y'_{n, k, n-\lambda}=r y'_{n, k, \lambda}+r' y_{n, k, \lambda}$. 
    By the triangle inequality, we have 
    \begin{align}
        \sum_{\lambda=0}^n \binom{n}{\lambda}|x_\lambda| 
        \geq&\sum_{\lambda=0}^{\lceil k/2\rceil} \binom{n}{\lambda}[(|r| |y_{n, k, \lambda}|-|r'| |y'_{n, k, \lambda}|)+(|r'| |y_{n, k, \lambda}|-|r| |y'_{n, k, \lambda}|)] \nonumber\\
        =&(|r|+|r'|)\sum_{\lambda=0}^{\lceil k/2\rceil} \binom{n}{\lambda}(|y_{n, k, \lambda}|-|y'_{n, k, \lambda}|). \label{SMeq:SMlem:U1_nontrivial_nonexistence_general4}
    \end{align}
    For $\lambda\in\{0, 1, ..., \lceil k/2\rceil\}$, we have 
    \begin{align}
        &|y_{n, k, \lambda}|
        =\frac{1}{(n-k-1)!}\prod_{\alpha=1}^{n-k-1} \left(\alpha+\left\lceil \frac{k}{2}\right\rceil-\lambda\right), \nonumber\\
        &|y'_{n, k, \lambda}|
        =\frac{1}{(n-k-1)!}\prod_{\alpha=1}^{n-k-1} \left(\alpha-1+\left\lceil \frac{k}{2}\right\rceil-\lambda\right) 
        =\frac{1}{(n-k-1)!}\prod_{\alpha=0}^{n-k-2} \left(\alpha+\left\lceil \frac{k}{2}\right\rceil-\lambda\right), 
    \end{align}
    which imply that 
    \begin{align}
        |y_{n, k, \lambda}|-|y'_{n, k, \lambda}|
        =&\frac{1}{(n-k-2)!}\prod_{\alpha=1}^{n-k-2} \left(\alpha+\left\lceil \frac{k}{2}\right\rceil-\lambda\right) \nonumber\\
        =&\frac{1}{[n-(k+1)-1]!}\prod_{\alpha=1}^{n-(k+1)-1} \left(\alpha+\left\lceil \frac{k+1}{2}\right\rceil-\lambda\right) \nonumber\\
        =&|y_{n, k+1, \lambda}|. \label{SMeq:SMlem:U1_nontrivial_nonexistence_general5}
    \end{align}
    By plugging Eq.~\eqref{SMeq:SMlem:U1_nontrivial_nonexistence_general5} into Eq.~\eqref{SMeq:SMlem:U1_nontrivial_nonexistence_general4}, we get 
    \begin{align}
        \sum_{\lambda=0}^n \binom{n}{\lambda}|x_\lambda| 
        \geq (|r|+|r'|)\sum_{\lambda=0}^{\lceil k/2\rceil} \binom{n}{\lambda} |y_{n, k+1, \lambda}| 
        =(|r|+|r'|)\cdot\frac{1}{2}\sum_{\lambda=0}^n \binom{n}{\lambda} |y_{n, k+1, \lambda}|, \label{SMeq:SMlem:U1_nontrivial_nonexistence_general6}
    \end{align}
    where we used $y_{n, k+1, \lambda}=0$ for all $\lambda\in\{\lceil k/2\rceil+1, \lceil k/2\rceil+2, ..., n-\lceil k/2\rceil-1\}$ and $|y_{n, k+1, \lambda}|=|y_{n, k+1, n-\lambda}|$ for all $\lambda\in\{0, 1, ..., \lceil k/2\rceil\}$ in the equality. 
    By plugging Eq.~\eqref{SMeq:SMlem:U1_nontrivial_existence_general6} into Eq.~\eqref{SMeq:SMlem:U1_nontrivial_nonexistence_general6}, we get 
    \begin{align}
        \sum_{\lambda=0}^n \binom{n}{\lambda}|x_\lambda| 
        =(|r|+|r'|) b_{n, k+1} 
        \geq (|r|+|r'|) b_{n, k}, \label{SMeq:SMlem:U1_nontrivial_nonexistence_general7}
    \end{align}
    where we used the assumption of $b_{n, k}\leq b_{n, k+1}$ in the second inequality. 
    By Eqs.~\eqref{SMeq:SMlem:U1_nontrivial_nonexistence_general1} and \eqref{SMeq:SMlem:U1_nontrivial_nonexistence_general7}, we have $|r|+|r'|<2$, which means $r=0$ or $r'=0$. 
    In both cases, we have 
    \begin{align}
        \sum_{\lambda=0}^n \binom{n}{\lambda}|x_\lambda|
        =\sum_{\lambda=0}^n \binom{n}{\lambda}|ry_{n, k, \lambda}+r'y_{n, k, n-\lambda}| 
        =\sum_{\lambda=0}^n \binom{n}{\lambda}(|r||y_{n, k, \lambda}|+|r'||y_{n, k, n-\lambda}|) 
        =(|r|+|r'|)\cdot 2b_{n, k}, \label{SMeq:SMlem:U1_nontrivial_nonexistence_general8}
    \end{align}
    where we used Eq.~\eqref{SMeq:SMlem:U1_nontrivial_existence_general6} in the last equality. 
    By Eqs.~\eqref{SMeq:SMlem:U1_nontrivial_nonexistence_general1} and \eqref{SMeq:SMlem:U1_nontrivial_nonexistence_general8}, we get $|r|+|r'|<1$, which implies that $r=r'=0$. 
    We can therefore conclude that Eqs.~\eqref{SMeq:SMlem:U1_cond1} and \eqref{SMeq:SMlem:U1_cond2} have no nontrivial integer solution when $t<b_{n, k}$. 
\end{proof}

By using the lemmas above, we can prove Theorem~\ref{SMthm:U1_general_locality} as follows: \\

\noindent
\textit{Proof of Theorem~\ref{SMthm:U1_general_locality}.}
    In Lemma~\ref{SMlem:U1_cond}, we have explicitly rewritten the equations in Theorem~\ref{SMthm:general} in the $\mathrm{U}(1)$ case. 
    When $t$ does not satisfy Eq.~\eqref{SMeq:SMthm:U1_general_locality1}, by Lemma~\ref{SMlem:U1_nontrivial_existence_general}, there exists a nontrivial integer solution for all $n\geq k+1$. 
    When $t$ satisfies Eq.~\eqref{SMeq:SMthm:U1_general_locality1}, by Lemma~\ref{SMlem:U1_nontrivial_nonexistence_general}, there exists no nontrivial integer solution under a certain assumption about $n$ and $k$, which are guaranteed when $n\geq 2^k$ by Lemma~\ref{SMlem:U1_sufficient_cond}. 
\hfill $\square$\\

By directly considering the condition on $t$ such that Eqs.~\eqref{SMeq:SMlem:U1_cond1} and \eqref{SMeq:SMlem:U1_cond2} in Lemma~\ref{SMlem:U1_cond} have no nontrivial integer solution for the region of $n$ that does not satisfy the assumption in Lemma~\ref{SMlem:U1_nontrivial_nonexistence_general}, we can prove Theorem~\ref{SMthm:U1_small_locality}. \\

\noindent
\textit{Proof of Theorem~\ref{SMthm:U1_small_locality}.}
    By Lemma~\ref{SMlem:U1_cond}, the distribution of the $(G, R)$-symmetric $k$-local random circuit forming an asymptotic unitary $t$-design if and only if Eqs.~\eqref{SMeq:SMlem:U1_cond1} and \eqref{SMeq:SMlem:U1_cond2} have no nontrivial integer solution. 
    We note that $b_{n, k}$ defined by Eq.~\eqref{SMeq:U1_bound_def} satisfies $b_{n, 2}=2(n-1)$ $b_{n, 3}=n(n-2)$, and $b_{n, 4}=2(n-1)(n-3)$. 
    Since we have proven the existence of a nontrivial integer solution when $t$ satisfies Eq.~\eqref{SMeq:U1_small_locality_order} by Lemma~\ref{SMlem:U1_nontrivial_existence_general}, it is sufficient to show the nonexistence of a nontrivial integer solution of Eqs.~\eqref{SMeq:SMlem:U1_cond1} and \eqref{SMeq:SMlem:U1_cond2} when $t$ does not satisfy Eq.~\eqref{SMeq:U1_small_locality_order}.

    First, we consider the case when $k=2$. 
    We can confirm that the assumption in Lemma~\ref{SMlem:U1_nontrivial_nonexistence_general} holds for all $n\geq 3$. 
    When $n\geq 4$, for any $j\in\{2, 3, ..., n-2\}$, we have 
    \begin{align}
        \binom{n}{j} 
        \geq\binom{n}{2} 
        =\frac{n}{4}b_{n, 2} 
        \geq b_{n, 2},  
    \end{align}
    and this trivially holds when $n=3$, because the set $\{2, 3, ..., n-2\}$ is empty. 
    By Lemma~\ref{SMlem:U1_nontrivial_nonexistence_general}, Eqs.~\eqref{SMeq:SMlem:U1_cond1} and \eqref{SMeq:SMlem:U1_cond2} have no nontrivial integer solution for all $n\geq 3$.

    Next, we consider the case when $k=3$. 
    We can confirm that the assumptions in Lemma~\ref{SMlem:U1_nontrivial_nonexistence_general} hold for $n=5$ and $n\geq 7$. 
    As for the first assumption, when $n\geq 7$, we have for any $j\in\{3, 4, ..., n-3\}$, 
    \begin{align}
        \binom{n}{j} 
        \geq\binom{n}{3} 
        =\frac{n-1}{6}b_{n, 3} 
        \geq b_{n, 3}, 
    \end{align}
    and this trivially holds when $n=5$, because the set $\{3, 4, ..., n-3\}$ is empty. 
    As for the second assumption, we have 
    \begin{align}
        b_{n, 4} 
        =b_{n, 3}+(n-1)(n-5)+1 
        \geq b_{n, 3}. 
    \end{align}
    By Lemma~\ref{SMlem:U1_nontrivial_nonexistence_general}, Eqs.~\eqref{SMeq:SMlem:U1_cond1} and \eqref{SMeq:SMlem:U1_cond2} have no nontrivial integer solution when $n=5$ or $n\geq 7$. 
    Thus, we have only to check the cases of $n=4$ and $6$ in the following.

    \begin{itemize}
    \item
    When $n=4$, Eqs.~\eqref{SMeq:SMlem:U1_cond1} and \eqref{SMeq:SMlem:U1_cond2} are explicitly written as 
    \begin{align}
        &|x_0|+4|x_1|+6|x_2|+4|x_3|+|x_4|\leq 2t, \\
        &x_j+x_{j+1}=0\ \forall j\in\{0, 1, 2, 3\}. 
    \end{align}
    We show that these equations have no nontrivial integer solution for $t<8$. 
    We take an arbitrary integer solution $\bm{x}$. 
    Equation~\eqref{SMeq:SMlem:U1_cond2} implies $x_j=(-1)^j x_0$ for all $j\in\{1, 2, 3, 4\}$. 
    By plugging this into Eq.~\eqref{SMeq:SMlem:U1_cond1}, we get $16|x_0|\leq 2t<16$, which implies $x_0=0$. 
    We thus have $\bm{x}=\bm{0}$, which implies that Eqs.~\eqref{SMeq:SMlem:U1_cond1} and \eqref{SMeq:SMlem:U1_cond2} have no nontrivial integer solution.

    \item
    When $n=6$, Eqs.~\eqref{SMeq:SMlem:U1_cond1} and \eqref{SMeq:SMlem:U1_cond2} are explicitly written as 
    \begin{align}
        &|x_0|+6|x_1|+15|x_2|+20|x_3|+15|x_4|+6|x_5|+|x_6|\leq 2t, \\
        &x_j+3x_{j+1}+3x_{j+2}+x_{j+3}=0\ \forall j\in\{0, 1, 2, 3\}. 
    \end{align}
    We show that these equations have no nontrivial integer solution for $t<24$. 
    We take an arbitrary integer solution $\bm{x}$. 
    We define $y_j:=x_j+x_{6-j}$ for $j=0, 1$, and $2$. 
    By Eq.~\eqref{SMeq:SMlem:U1_cond2}, we have $y_0=9y_2+16x_3$ and $y_1=-4y_2-6x_3$. 
    By Lemma~\ref{SMlem:solution_bound}, we have $|y_2|\leq 1$ and $|x_3|\leq 1$. 
    We thus have $(y_2, x_3)=\pm(1, 1)$, $\pm(0, 1)$, $\pm(-1, 1)$, $\pm(1, 0)$, or $(0, 0)$, which implies $\sum_{j=0}^6 \binom{6}{j}|x_j|\geq |y_0|+6|y_1|+15|y_2|+20|x_3|=120$, $72$, $54$, $48$, or $0$, respectively, where we used the triangle inequality. 
    By combining this with Eq.~\eqref{SMeq:SMlem:U1_cond1}, we get $\bm{x}=\bm{0}$. 
    We can therefore conclude that Eqs.~\eqref{SMeq:SMlem:U1_cond1} and \eqref{SMeq:SMlem:U1_cond2} have no nontrivial integer solution.     
    \end{itemize}

    Finally, we consider the case when $k=4$. 
    We can confirm that the assumption in Lemma~\ref{SMlem:U1_nontrivial_nonexistence_general} holds for $n=5$ and $n\geq 11$. 
    When $n\geq 11$, for any $j\in\{3, 4, ..., n-3\}$, we have 
    \begin{align}
        \binom{n}{j} 
        \geq\binom{n}{3} 
        =b_{n, 3}+\frac{(n-1)[(n-3)(n-11)+3]}{6} 
        \geq b_{n, 3}, 
    \end{align}
    and this trivially holds when $n=5$, because the set $\{3, 4, ..., n-3\}$ is empty. 
    By Lemma~\ref{SMlem:U1_nontrivial_nonexistence_general}, Eqs.~\eqref{SMeq:SMlem:U1_cond1} and \eqref{SMeq:SMlem:U1_cond2} have no nontrivial integer solution when $n=5$ or $n\geq 11$. 
    Thus, we have only to check the cases of $n=6, 7, 8, 9$, and $10$ in the following.

    \begin{itemize}
    \item
    When $n=6$, Eqs.~\eqref{SMeq:SMlem:U1_cond1} and \eqref{SMeq:SMlem:U1_cond2} are explicitly written as   
    \begin{align}
        &|x_0|+6|x_1|+15|x_2|+20|x_3|+15|x_4|+6|x_5|+|x_6|\leq 2t, \\
        &x_j+2x_{j+1}+x_{j+2}=0\ \forall j\in\{0, 1, 2, 3, 4\}. 
    \end{align}
    We show that these equations have no nontrivial integer solution for $t<30$. 
    We take an arbitrary integer solution $\bm{x}$. 
    We define $y_j:=x_j+x_{n-j}$ for $j=0$, $1$, and $2$. 
    By Eq.~\eqref{SMeq:SMlem:U1_cond2}, we have $y_0=-2x_3$, $y_1=2x_3$, and $y_2=-2x_3$, which imply that $64|x_3|=|y_0|+6|y_1|+15|y_2|+20|x_3|\leq\sum_{j=0}^6 \binom{6}{j}|x_j|\leq 2t<60$. 
    We thus get $y_0=y_1=y_2=x_3=0$. 
    By Eq.~\eqref{SMeq:SMlem:U1_cond2}, we have $x_0=-x_6=3x_2$ and $x_1=-x_5=-2x_2$. 
    By noting that $60|x_2|=\sum_{j=0}^6 \binom{6}{j}|x_j|\leq 2t<60$, we get $\bm{x}=\bm{0}$. 
    We can therefore conclude that Eqs.~\eqref{SMeq:SMlem:U1_cond1} and \eqref{SMeq:SMlem:U1_cond2} have no nontrivial integer solution.

    \item
    When $n=7$, Eqs.~\eqref{SMeq:SMlem:U1_cond1} and \eqref{SMeq:SMlem:U1_cond2} are explicitly written as 
    \begin{align}
        &|x_0|+7|x_1|+21|x_2|+35|x_3|+35|x_4|+21|x_5|+7|x_6|+|x_7|\leq 2t, \\
        &x_{j+0}+3x_{j+1}+3x_{j+2}+x_{j+3}=0\ \forall j\in\{0, 1, 2, 3, 4\}. 
    \end{align}
    We show that these equations have no nontrivial integer solution for $t<48$. 
    We take an arbitrary integer solution $\bm{x}$. 
    We define $y_j:=x_j+x_{n-j}$ for $j=0, 1, 2$, and $3$. 
    By Eq.~\eqref{SMeq:SMlem:U1_cond2}, we have $y_0=-7y_3$, $y_1=5y_3$, and $y_2=-3y_3$, which imply that $140|y_3|=|y_0|+7|y_1|+21|y_2|+35|y_3|\leq\sum_{j=0}^7 \binom{7}{j}|x_j|\leq 2t<96$, where we used the triangle inequality. 
    We thus get $y_0=y_1=y_2=y_3=0$. 
    By noting that $21|x_2|+35|x_3|=(21|x_2|+35|x_3|+35|x_4|+21|x_5|)/2<48$, we have $(x_2, x_3)=\pm(0, 1)$, $\pm(1, 0)$, $\pm(2, 0)$, or $(0, 0)$, which implies  $\sum_{j=0}^7 \binom{7}{j}|x_j|=108$, $96$, $192$, or $0$, respectively. 
    By combining this with Eq.~\eqref{SMeq:SMlem:U1_cond1}, we get $\bm{x}=\bm{0}$. 
    We can therefore conclude that Eqs.~\eqref{SMeq:SMlem:U1_cond1} and \eqref{SMeq:SMlem:U1_cond2} have no nontrivial integer solution.

    \item
    When $n=8$, Eqs.~\eqref{SMeq:SMlem:U1_cond1} and \eqref{SMeq:SMlem:U1_cond2} are explicitly written as 
    \begin{align}
        &|x_0|+8|x_1|+28|x_2|+56|x_3|+70|x_4|+56|x_5|+28|x_6|+8|x_7|+|x_8|\leq 2t, \\
        &x_{j+0}+4x_{j+1}+6x_{j+2}+4x_{j+3}+x_{j+4}=0\ \forall j\in\{0, 1, 2, 3, 4\}.  
    \end{align}
    We show that these equations have no nontrivial integer solution for $t<70$. 
    We take an arbitrary integer solution $\bm{x}$. 
    By Lemma~\ref{SMlem:solution_bound}, we have $x_4=0$. 
    We define $y_j:=x_j+x_{n-j}$ for $j=0, 1, 2$, and $3$. 
    By Eq.~\eqref{SMeq:SMlem:U1_cond2}, we have $y_0=-16y_3$, $y_1=9y_3$, and $y_2=-4y_3$, which imply that $256|y_3|=\sum_{j=0}^3 \binom{8}{j} |y_j|=\sum_{j=0}^8 \binom{8}{j}|x_j|\leq 2t<140$. 
    We thus get $y_0=y_1=y_2=y_3=0$. 
    By noting that $28|x_2|+56|x_3|=(28|x_2|+56|x_3|+56|x_5|+28|x_6|)/2\leq t<70$, we have $(x_2, x_3)=\pm(0, 1)$, $\pm(1, 0)$, $\pm(2, 0)$, $(0, 0)$, which implies $\sum_{j=0}^8 \binom{8}{j}|x_j|=224$, $140$, $280$, or $0$, respectively. 
    By combining this with Eq.~\eqref{SMeq:SMlem:U1_cond1}, we get $\bm{x}=\bm{0}$. 
    We can therefore conclude that Eqs.~\eqref{SMeq:SMlem:U1_cond1} and \eqref{SMeq:SMlem:U1_cond2} have no nontrivial integer solution.

    \item
    When $n=9$, Eqs.~\eqref{SMeq:SMlem:U1_cond1} and \eqref{SMeq:SMlem:U1_cond2} are explicitly written as 
    \begin{align}
        &|x_0|+9|x_1|+36|x_2|+84|x_3|+126|x_4|
        +126|x_5|+84|x_6|+36|x_7|+9|x_8|+|x_9|\leq 2t, \\
        &x_{j+0}+5x_{j+1}+10x_{j+2}+10x_{j+3}+5x_{j+4}+x_{j+5}=0\ \forall j\in\{0, 1, 2, 3, 4\}. 
    \end{align}
    We show that these equations have no nontrivial integer solution for $t<96$. 
    We take an arbitrary integer solution $\bm{x}$. 
    By Lemma~\ref{SMlem:solution_bound}, we have $x_4=x_5=0$. 
    We define $y_j:=x_j+x_{n-j}$ for $j=0, 1, 2$, and $3$. 
    By Eq.~\eqref{SMeq:SMlem:U1_cond2}, we have $y_0=-30y_3$, $y_1=14y_3$, and $y_2=-5y_3$, which imply $420|y_3|
    =\sum_{j=0}^3 \binom{9}{j}|y_j|
    \leq \sum_{j=0}^9 \binom{9}{j}|x_j|\leq 2t
    <192$. 
    We thus get $y_0=y_1=y_2=y_3=0$. 
    By noting that $36|x_2|+84|x_3|=(36|x_2|+84|x_3|+84|x_6|+36|x_7|)/2\leq t$, we have $(x_2, x_3)=\pm(0, 1)$, $\pm(1, 0)$, $\pm(2, 0)$, or $(0, 0)$, which implies $\sum_{j=0}^9 \binom{9}{j}|x_j|=400$, $192$, $384$, or $0$, respectively. 
    By combining this with Eq.~\eqref{SMeq:SMlem:U1_cond1}, we get $\bm{x}=\bm{0}$. 
    We can therefore conclude that Eqs.~\eqref{SMeq:SMlem:U1_cond1} and \eqref{SMeq:SMlem:U1_cond2} have no nontrivial integer solution.

    \item
    When $n=10$, Eqs.~\eqref{SMeq:SMlem:U1_cond1} and \eqref{SMeq:SMlem:U1_cond2} are explicitly written as 
    \begin{align}
        &|x_0|+10|x_1|+45|x_2|+120|x_3|+210|x_4|+252|x_5|
        +210|x_6|+120|x_7|+45|x_8|+10|x_9|+|x_{10}|\leq 2t, \\
        &x_{j+0}+6x_{j+1}+15x_{j+2}+20x_{j+3}+15x_{j+4}+6x_{j+5}+x_{j+6}=0\ \forall j\in\{0, 1, 2, 3, 4\}. 
    \end{align}
    We show that these equations have no nontrivial integer solution for $t<126$. 
    We take an arbitrary integer solution $\bm{x}$. 
    By Lemma~\ref{SMlem:solution_bound}, we have $x_4=x_5=x_6=0$. 
    We define $y_j:=x_j+x_{n-j}$ for $j=0, 1, 2$, and $3$. 
    By Eq.~\eqref{SMeq:SMlem:U1_cond2}, we have $y_0=-50y_3$, $y_1=20y_3$, and $y_2=-6y_3$, which imply $640|y_3|=\sum_{j=0}^3 \binom{10}{j} |y_j|\leq\sum_{j=0}^{10} \binom{10}{j} |x_j|\leq 2t<252$. 
    We thus get $y_0=y_1=y_2=y_3=0$. 
    By noting that $45|x_2|+120|x_3|=(45|x_2|+120|x_3|+120|x_7|+45|x_8|)/2\leq t$, we have $(x_2, x_3)=\pm(0, 1)$, $\pm(1, 0)$, $\pm(2, 0)$, $(0, 0)$, which implies $\sum_{j=0}^{10} \binom{10}{j}|x_j|=648$, $252$, $504$, or $0$, respectively. 
    By combining this with Eq.~\eqref{SMeq:SMlem:U1_cond1}, we get $\bm{x}=\bm{0}$. 
    We can therefore conclude that Eqs.~\eqref{SMeq:SMlem:U1_cond1} and \eqref{SMeq:SMlem:U1_cond2} have no nontrivial integer solution. 
    \end{itemize}

\hfill $\square$

\subsection{SU(2) symmetry} \label{SMsubsec:SU2}

In this section, we consider the representation $R$ of $\mathrm{SU}(2)$ defined by Eq.~\eqref{SMeq:SU2_representation}. 
Since $\bigcup_{\gamma\in\Gamma} \mathcal{U}_{n, G, R}^\gamma$ is semi-universal for $\mathcal{U}_{n, G, R}$~\cite{marvian2024rotationally}, we can use the result of Theorem~\ref{SMthm:general}. 
By noting that $R$ can be written as $T^{\otimes n}$ with a representation $T$ on a single qubit, we consider the condition for Eqs.~\eqref{SMeq:SMthm:general_cond1} and \eqref{SMeq:SMthm:general_cond4} having no nontrivial integer solutions. 
Since $R$ can be irreducibly decomposed into spin-$\lambda$ representations with $\lambda\in\{n/2, n/2-1, ..., n/2-\lfloor n/2\rfloor\}$, we use this $\lambda$ as the index for the irreducible representation appearing in $R$, i.e., $\Lambda$ in Theorem~\ref{SMthm:general} is given by 
\begin{align}
    \Lambda=\{n/2, n/2-1, ..., n/2-\lfloor n/2\rfloor\}. \label{SMeq:SU2_irrep_def}
\end{align}

As a preparation, we derive a property about $f_\lambda(Q_\sigma)$, where $Q_\sigma$ is the permutation operator that brings the $j$th qubit to the $\sigma(j)$th qubit.

\begin{lemma} \label{SMlem:permutation_weight}
    Let $n\in\mathbb{N}$, $R$ be a unitary representation of $G=\mathrm{SU}(2)$ on $n$ qubits defined by Eq.~\eqref{SMeq:SU2_representation}, $\Lambda$ be given by Eq.~\eqref{SMeq:SU2_irrep_def}, $\lambda\in\Lambda$, $\sigma\in\mathfrak{S}_n$ be decomposed as $\sigma=\sigma_1\sigma_2\cdots\sigma_L$ with $p_l$-cycles $\sigma_l\in\mathfrak{S}_n$ nontrivially acting on disjoint subsets of $\{1, 2, ..., n\}$, and $\widetilde{f}_\lambda$ be defined by 
    \begin{align}
        \widetilde{f}_\lambda(A):=\sum_{\kappa\in\Lambda, \kappa\geq\lambda} f_\kappa(A)\ \forall A\in\mathcal{L}_{n, G, R}. 
    \end{align} 
    Then, 
    \begin{align}
        &\widetilde{f}_\lambda(I)
        =\binom{n}{\frac{n}{2}-\lambda}, \label{SMeq:SMlem:permutation_weight1}\\
        &\widetilde{f}_\lambda(Q_\sigma)
        =\sum_{q_1, q_2, ..., q_L\in\{0, 1\}} \binom{n-\sum_{l=1}^L p_l}{\frac{n}{2}-\lambda-\sum_{l=1}^L q_l p_l}. \label{SMeq:SMlem:permutation_weight2}
    \end{align}
\end{lemma}

\vspace{5mm}

We note that a $p$-cycle means a permutation $\sigma$ that nontrivially acts only on $p$ elements $j_1, j_2, ..., j_p$ and satisfies $\sigma(j_1)=j_2$, $\sigma(j_2)=j_3$, ..., $\sigma(j_p)=j_1$.

\begin{proof}
    We take an orthonormal basis $\{\ket{\lambda, \mu, \alpha}\}_{\mu\in\{\lambda, \lambda-1, ..., -\lambda\}, \alpha\in\{1, 2, ..., m_\lambda\}}$ of the spin-$\lambda$ representation space such that 
    \begin{align}
        &[(\mathrm{X}^\mathrm{tot})^2+(\mathrm{Y}^\mathrm{tot})^2+(\mathrm{Z}^\mathrm{tot})^2]\ket{\lambda, \mu, \alpha}
        =4\lambda(\lambda+1)\ket{\lambda, \mu, \alpha}, \\
        &\mathrm{Z}^\mathrm{tot}\ket{\lambda, \mu, \alpha}=2\mu\ket{\lambda, \mu, \alpha}, 
    \end{align}
    where $\mathrm{X}^\mathrm{tot}$, $\mathrm{Y}^\mathrm{tot}$, and $\mathrm{Z}^\mathrm{tot}$ are the sum of all the Pauli operators on the $n$ qubits, and $\alpha$ is the index for degeneracy. 
    We can take orthonormal bases $\{\ket{\mu}\}$ and $\{\ket{\alpha}\}$ of the representation space and the multiplicity space such that 
    \begin{align}
        F_\lambda(\ket{\mu}\otimes \ket{\alpha})
        =\ket{\lambda, \mu, \alpha}. 
    \end{align} 
    Thus, for any $\lambda, \kappa\in\Lambda$ satisfying $\lambda\leq\kappa$, we have 
    \begin{align}
        \sum_{\alpha\in\{1, 2, ..., m_\kappa\}} \bra{\kappa, \lambda, \alpha}A\ket{\kappa, \lambda, \alpha} 
        =&\sum_{\alpha\in\{1, 2, ..., m_\kappa\}} \sum_{\kappa'\in\Lambda} (\bra{\lambda}\otimes \bra{\alpha})F_\kappa^\dag F_{\kappa'}(I\otimes A_{\kappa'})F_{\kappa'}^\dag F_\kappa(\ket{\lambda}\otimes \ket{\alpha}) \nonumber\\
        =&\sum_{\alpha\in\{1, 2, ..., m_\kappa\}} (\bra{\lambda}\otimes \bra{\alpha})(I\otimes A_\kappa)(\ket{\lambda}\otimes \ket{\alpha}) \nonumber\\
        =&\mathrm{tr}(A_\kappa) \nonumber\\
        =&f_\kappa(A), 
    \end{align}
    where we used $F_\kappa^\dag F_{\kappa'}$ is the identity when $\kappa=\kappa'$ and otherwise $0$ in the second equality. 
    By the definition of $\widetilde{f}_\lambda$, we get 
    \begin{align}
        \widetilde{f}_\lambda(Q_\sigma)
        =\sum_{\kappa\in\Lambda, \kappa\geq\lambda} \sum_{\alpha\in\{1, 2, ..., m_\kappa\}} \bra{\kappa, \lambda, \alpha}Q_\sigma \ket{\kappa, \lambda, \alpha}. \label{SMeq:SMlem:permutation_weight3}
    \end{align}
    We note that $\{\ket{\kappa, \lambda, \alpha}\}_{\kappa\geq\lambda, \alpha\in\{1, 2, ..., m_\kappa\}}$ is an orthonormal basis of the eigenspace of $\mathrm{Z}^\mathrm{tot}$ with eigenvalue $2\lambda$, and we can also take another orthonormal basis $\{\ket{a_1 a_2 \cdots a_n}\}_{(a_1, a_2, ..., a_n)\in S_\lambda}$, where $S_\lambda:=\{(a_1, a_2, ..., a_n)\in\{0, 1\}^n\ |\ \#\{j\in\{1, 2, ..., n\}\ |\ a_j=1\}=n/2-\lambda\}$, and $\ket{a_1 a_2 \cdots a_n}$ is the tensor product of the eigenvectors $\ket{a}$ of the single-qubit Pauli-Z operators satisfying $\mathrm{Z}\ket{0}=(-1)^a\ket{a}$. 
    By the basis transformation, we can rewrite Eq.~\eqref{SMeq:SMlem:permutation_weight3} as 
    \begin{align}
        \widetilde{f}_\lambda(Q_\sigma) 
        =&\sum_{(a_1, a_2, ..., a_n)\in S_\lambda} \bra{a_1 a_2 \cdots a_n}Q_\sigma \ket{a_1 a_2 \cdots a_n} \nonumber\\
        =&\sum_{(a_1, a_2, ..., a_n)\in S_\lambda} \braket{a_{\sigma(1)} a_{\sigma(2)} \cdots a_{\sigma(n)} | a_1 a_2 \cdots a_n} \nonumber\\
        =&\#\{(a_1, a_2, ..., a_n)\in S_\lambda\ |\ a_{\sigma(j)}=a_j\ \forall j\in\{1, 2, ..., n\}\}. \label{SMeq:SMlem:permutation_weight4}
    \end{align}
    When $\sigma$ is the identity, Eq.~\eqref{SMeq:SMlem:permutation_weight4} implies that 
    \begin{align}
        \widetilde{f}_\lambda(I) 
        =\# S_\lambda 
        =\binom{n}{\frac{n}{2}-\lambda}. 
    \end{align}
    When $\sigma$ is decomposed into disjoint cycles $\{\sigma_l\}_{l=1}^L$, we take disjoint subsets $D_1, D_2, ..., D_L$ of $\{1, 2, ..., L\}$ such that $\sigma_l$ nontrivially acts on $D_l$ for all $l\in\{1, 2, ..., L\}$. 
    Since $\sigma_l$ is a $p_l$-cycle, $\# D_l=p_l$. 
    The condition $a_{\sigma(j)}=a_j\ \forall j\in\{1, 2, ..., n\}\}$ means that $a_j$ must be identical for every element $j$ in $D_l$. 
    When $a_j=q_l$ for all $j\in D_l$ with some $q_l\in\{0, 1\}$, the number of strings $(a_1, a_2, ..., a_n)\in S_\lambda$ is $\binom{n-\sum_{l=1}^L p_l}{n/2-\lambda-\sum_{l=1}^L q_l p_l}$. 
    By summing them up over all $q_1, q_2, ..., q_L\in\{0, 1\}$, we get Eq.~\eqref{SMeq:SMlem:permutation_weight2}. 
\end{proof}

By using Lemma~\ref{SMlem:permutation_weight}, we give the explicit expression of Eqs.~\eqref{SMeq:SMthm:general_cond1} and \eqref{SMeq:SMthm:general_cond4} in Theorem~\ref{SMthm:general}. 
The following lemma is the counterpart of Lemma~\ref{SMlem:U1_cond} in the $\mathrm{SU}(2)$ case.

\begin{lemma} \label{SMlem:SU2_cond}
    Let $n, k, t\in\mathbb{N}$, $R$ be a unitary representation of $G=\mathrm{SU}(2)$ defined by Eq.~\eqref{SMeq:SU2_representation}, and $\Lambda$ be given by Eq.~\eqref{SMeq:SU2_irrep_def}. 
    Then, the distribution of the $(G, R)$-symmetric $k$-local random circuit is an asymptotic $(G, R)$-symmetric unitary $t$-design if and only if there exists no nontrivial integer solution $\bm{x}=(x_\lambda)_{\lambda\in\Lambda}\in\mathbb{Z}^\Lambda$ satisfying 
    \begin{align}
        &\sum_{\lambda\in\Lambda} \left(\binom{n}{\frac{n}{2}-\lambda}-\binom{n}{\frac{n}{2}-\lambda-1}\right)|x_\lambda|\leq 2t, \label{SMeq:SMlem:SU2_cond1}\\
        &\sum_{\lambda\in\Lambda} \left(\binom{n-2j}{\frac{n}{2}-j-\lambda}-\binom{n-2j}{\frac{n}{2}-j-\lambda-1}\right)x_\lambda=0\ \forall j\in\left\{0, 1, 2, ..., \left\lfloor \frac{k}{2}\right\rfloor\right\}. \label{SMeq:SMlem:SU2_cond2}
    \end{align}
\end{lemma}

\begin{proof}
    Since $R$ is a tensor product of representation on a single qubit, Theorem~\ref{SMthm:general} implies that the distribution of $(G, R)$-symmetric $k$-local random circuit is an asymptotic unitary $t$-design if and only if Eqs.~\eqref{SMeq:SMthm:general_cond1} and \eqref{SMeq:SMthm:general_cond4} have no nontrivial integer solution. 
    Thus it is sufficient to show that Eqs.~\eqref{SMeq:SMthm:general_cond1} and \eqref{SMeq:SMthm:general_cond4} are equivalent to Eqs.~\eqref{SMeq:SMlem:SU2_cond1} and \eqref{SMeq:SMlem:SU2_cond2}, respectively. 
    By Lemma~\ref{SMlem:permutation_weight}, we have 
    \begin{align}
        m_\lambda
        =f_\lambda(I) 
        =\widetilde{f}_\lambda(I)-\widetilde{f}_{\lambda+1}(I) 
        =\binom{n}{\frac{n}{2}-\lambda}-\binom{n}{\frac{n}{2}-\lambda-1}. 
    \end{align}
    Thus Eq.~\eqref{SMeq:SMthm:general_cond1} is equivalent to Eq.~\eqref{SMeq:SMlem:SU2_cond1}.

    In the following, we show the equivalence between Eq.~\eqref{SMeq:SMthm:general_cond4} and Eq.~\eqref{SMeq:SMlem:SU2_cond2}. 
    By the Schur-Weyl duality, every $(G, R)$-symmetric operator $A\in\mathcal{L}_{k, G, R}$ can be written as a linear combination of the permutation operators $Q_\sigma$'s with permutations $\sigma\in\mathfrak{S}_k$, where $Q_\sigma$ is the operator that brings the $j$th qubit to $\sigma(j)$th qubit. 
    Thus Eq.~\eqref{SMeq:SMthm:general_cond4} is equivalent to 
    \begin{align}
        \sum_{\lambda\in\Lambda} f_\lambda(Q_\sigma\otimes \mathrm{I}^{\otimes n-k})x_\lambda=0\ \forall\sigma\in\mathfrak{S}_k. \label{SMeq:SMlem:SU2_cond4}
    \end{align}

    First, we show that Eq.~\eqref{SMeq:SMlem:SU2_cond2} implies Eq.~\eqref{SMeq:SMlem:SU2_cond4}. 
    We note that the permutation $\sigma$ nontrivially acting on at most $k$ elements can be written as $\sigma=\sigma_1\sigma_2\cdots\sigma_L$ with some disjoint $p_l$-cycles satisfying $\sum_{l=1}^L p_l=k$. 
    By Lemma~\ref{SMlem:permutation_weight}, we have 
    \begin{align}
        \widetilde{f}_\lambda(Q_\sigma\otimes \mathrm{I}^{\otimes n-k}) 
        =&\sum_{q_1, q_2, ..., q_L\in\{0, 1\}} \binom{n-k}{\frac{n}{2}-\lambda-\sum_{l=1}^L q_l p_l} \nonumber\\
        =&\frac{1}{2}\sum_{q_1, q_2, ..., q_L\in\{0, 1\}} \left(\binom{n-k}{\frac{n}{2}-\lambda-\sum_{l=1}^L q_l p_l}+\binom{n-k}{\frac{n}{2}-\lambda-\sum_{l=1}^L (1-q_l) p_l}\right) \nonumber\\
        =&\frac{1}{2}\sum_{q_1, q_2, ..., q_L\in\{0, 1\}} \left(\binom{n-k}{\frac{n}{2}-\lambda-\sum_{l=1}^L q_l p_l}+\binom{n-k}{\frac{n}{2}-\lambda-k+\sum_{l=1}^L q_l p_l}\right) \nonumber\\
        =&\sum_{j=0}^k u_j\left(\binom{n-k}{\frac{n}{2}-\lambda-j}+\binom{n-k}{\frac{n}{2}-\lambda-k+j}\right) \nonumber\\
        =&\sum_{j=0}^{\lfloor k/2\rfloor} \widetilde{u}_j\left(\binom{n-k}{\frac{n}{2}-\lambda-j}+\binom{n-k}{\frac{n}{2}-\lambda-k+j}\right), \label{SMeq:SMlem:SU2_cond7}
    \end{align}
    where $u_j$ is defined by 
    \begin{align}
        u_j:=\frac{1}{2}\#\left\{(q_1, q_2, ..., q_L)\in\{0, 1\}^L \ \middle|\ \sum_{l=1}^L q_l p_l=j\right\} 
    \end{align}
    for $j\in\{0, 1, ..., k\}$, and $\widetilde{u}_j$ is defined by $\widetilde{u}_j:=u_j+u_{k-j}$ for $j\in\{0, 1, ..., (k-1)/2\}$ when $k$ is odd, and $\widetilde{u}_j:=u_j+u_{k-j}$ for $j\in\{0, 1, ..., k/2-1\}$ and $\widetilde{u}_{k/2}:=u_{k/2}$ when $k$ is even. 
    By Lemma~\ref{SMlem:SU2_coefficient1}, for any $\lambda\in\Lambda$, we can take $(v_{j, l})_{j, l\in\{0, 1, ..., \lfloor k/2\rfloor\}}\in\mathbb{R}^{(\lfloor k/2\rfloor+1)^2}$ such that 
    \begin{align}
        \binom{n-k}{\frac{n}{2}-\lambda-j}+\binom{n-k}{\frac{n}{2}-\lambda-k+j}
        =\sum_{l=0}^{\lfloor k/2\rfloor} v_{j, l}\binom{n-2l}{\frac{n}{2}-\lambda-l}. \label{SMeq:SMlem:SU2_cond8}
    \end{align}
    By plugging Eq.~\eqref{SMeq:SMlem:SU2_cond8} into Eq.~\eqref{SMeq:SMlem:SU2_cond7}, we get 
    \begin{align}
        \widetilde{f}_\lambda(Q_\sigma\otimes \mathrm{I}^{\otimes n-k})
        =\sum_{l=0}^{\lfloor k/2\rfloor} \sum_{j=0}^{\lfloor k/2\rfloor} \widetilde{u}_j v_{j, l}\binom{n-2l}{\frac{n}{2}-\lambda-l}, 
    \end{align}
    which implies 
    \begin{align}
        f_\lambda(Q_\sigma\otimes \mathrm{I}^{\otimes n-k})
        =\widetilde{f}_\lambda(Q_\sigma\otimes \mathrm{I}^{\otimes n-k})-\widetilde{f}_{\lambda+1}(Q_\sigma\otimes \mathrm{I}^{\otimes n-k}) 
        =\sum_{l=0}^{\lfloor k/2\rfloor} \sum_{j=0}^{\lfloor k/2\rfloor} \widetilde{u}_j v_{j, l}\left(\binom{n-2l}{\frac{n}{2}-l-\lambda}-\binom{n-2l}{\frac{n}{2}-l-\lambda-1}\right). 
    \end{align}
    Thus Eq.~\eqref{SMeq:SMlem:SU2_cond2} implies Eq.~\eqref{SMeq:SMlem:SU2_cond4}.

    Next, we show that Eq.~\eqref{SMeq:SMlem:SU2_cond4} implies Eq.~\eqref{SMeq:SMlem:SU2_cond2}. 
    By Lemma~\ref{SMlem:SU2_coefficient2}, for any $\lambda\in\Lambda$, we can take $(w_{j, l})_{j, l\in\{0, 1, ..., \lfloor k/2\rfloor\}}\in\mathbb{R}^{(\lfloor k/2\rfloor+1)^2}$ such that for any $j\in\{0, 1, ..., \lfloor k/2\rfloor\}$, 
    \begin{align}
        \binom{n-2j}{\frac{n}{2}-\lambda-j} 
        =\sum_{l=0}^{\lfloor k/2\rfloor} w_{j, l}\left(\binom{n-l}{\frac{n}{2}-\lambda}+\binom{n-l}{\frac{n}{2}-\lambda-l}\right). \label{SMeq:SMlem:SU2_cond9}
    \end{align}
    For each $l\in\{1, 2, ..., \lfloor k/2\rfloor\}$, we take some $l$-cycle $\zeta_l$. 
    By Lemma~\ref{SMlem:permutation_weight}, we have 
    \begin{align}
        \widetilde{f}_\lambda(I)=\binom{n}{\frac{n}{2}-\lambda},\ 
        \widetilde{f}_\lambda(Q_{\zeta_l}\otimes \mathrm{I}^{\otimes n-k}) 
        =\binom{n-l}{\frac{n}{2}-\lambda}+\binom{n-j}{\frac{n}{2}-\lambda-l}. \label{SMeq:SMlem:SU2_cond10}
    \end{align}
    By using Eq.~\eqref{SMeq:SMlem:SU2_cond10}, we can rewrite Eq.~\eqref{SMeq:SMlem:SU2_cond9} as 
    \begin{align}
        \binom{n-2j}{\frac{n}{2}-\lambda-j} 
        =2w_{j, 0}\widetilde{f}_\lambda(I)+\sum_{l=1}^{\lfloor k/2\rfloor} w_{j, l}\widetilde{f}_\lambda(Q_{\zeta_l}\otimes \mathrm{I}^{\otimes n-k}), 
    \end{align}
    which implies that 
    \begin{align}
        \binom{n-2j}{\frac{n}{2}-\lambda-j}-\binom{n-2j}{\frac{n}{2}-\lambda-j-1}
        =&2w_{j, 0}(\widetilde{f}_\lambda(I)-\widetilde{f}_{\lambda+1}(I))+\sum_{l=1}^{\lfloor k/2\rfloor} w_{j, l}(\widetilde{f}_\lambda(Q_{\zeta_l}\otimes \mathrm{I}^{\otimes n-k})-\widetilde{f}_{\lambda+1}(Q_{\zeta_l}\otimes \mathrm{I}^{\otimes n-k})) 
        \nonumber\\
        =&2w_{j, 0}f_\lambda(I)+\sum_{l=1}^{\lfloor k/2\rfloor} w_{j, l}f_\lambda(Q_{\zeta_l}\otimes \mathrm{I}^{\otimes n-k}). 
    \end{align}
    Thus Eq.~\eqref{SMeq:SMlem:SU2_cond4} implies Eq.~\eqref{SMeq:SMlem:SU2_cond2}. 
\end{proof}

In the following, we consider the condition on $t$ for Eqs.~\eqref{SMeq:SMlem:SU2_cond1} and \eqref{SMeq:SMlem:SU2_cond2} in Lemma~\ref{SMlem:SU2_cond} having no nontrivial integer solution. 
First, we present a sufficient condition for the existence of a nontrivial integer solution. 
The following lemma is the counterpart of Lemma~\ref{SMlem:U1_nontrivial_existence_general} in the $\mathrm{SU}(2)$ case.

\begin{lemma} \label{SMlem:SU2_nontrivial_existence_general}
	Let $n, k\in\mathbb{N}$, $k\geq 2$, $n\geq 2(\lfloor k/2\rfloor+1)$, and $R$ be a unitary representation of $G=\mathrm{SU}(2)$ defined by Eq.~\eqref{SMeq:SU2_representation}. 
    Then, Eqs.~\eqref{SMeq:SMlem:SU2_cond1} and \eqref{SMeq:SMlem:SU2_cond2} have a nontrivial integer solution if 
    \begin{align}
        t\geq\frac{2^{\lfloor k/2\rfloor}}{(\left\lfloor \frac{k}{2}\right\rfloor+1)!}\prod_{\alpha=1}^{\lfloor k/2\rfloor+1} (n-2\alpha+1). \label{SMeq:SMlem:SU2_nontrivial_existence_general0}
    \end{align} 
\end{lemma}

We exclude the case $n=k+1$ with even $k$ in this lemma. 
In that case, the equations have no nontrivial integer solution for all $t\in\mathbb{N}$.

\begin{proof}
    We define $\bm{y}=(y_\lambda)_{\lambda\in\Lambda}$ by 
    \begin{align}
        y_\lambda:=(-1)^{n/2-\lambda}\binom{\lambda+\frac{n}{2}-\left\lfloor \frac{k}{2}\right\rfloor-1}{\lambda-\frac{n}{2}+\left\lfloor \frac{k}{2}\right\rfloor+1}, \label{SMeq:SMlem:SU2_nontrivial_existence_general1}
    \end{align}
    and show that $\bm{y}$ is a nontrivial integer solution of Eqs.~\eqref{SMeq:SMlem:SU2_cond1} and \eqref{SMeq:SMlem:SU2_cond2}. 
    By the definition of $\bm{y}$, we can show that $\bm{y}$ satisfies Eqs.~\eqref{SMeq:SMlem:SU2_cond1} as follows: 
    \begin{align}
        &\sum_{\lambda\in\{n/2, n/2-1, ..., n/2-\lfloor k/2\rfloor-1\}} \left(\binom{n-2j}{\frac{n}{2}-j-\lambda}-\binom{n-2j}{\frac{n}{2}-j-\lambda-1}\right)y_\lambda \nonumber\\
        =&(-1)^{\lfloor k/2\rfloor+1}\sum_{\kappa=0}^{\lfloor k/2\rfloor+1} \left(\binom{n-2j}{\left\lfloor \frac{k}{2}\right\rfloor-j+1-\kappa}-\binom{n-2j}{\left\lfloor \frac{k}{2}\right\rfloor-j-\kappa}\right)\cdot (-1)^\kappa\binom{n-2\left\lfloor \frac{k}{2}\right\rfloor-2+\kappa}{\kappa} \nonumber\\
        =&(-1)^{\lfloor k/2\rfloor+1}\sum_{\kappa=0}^{\lfloor k/2\rfloor+1} \left(\binom{n-2j}{\left\lfloor \frac{k}{2}\right\rfloor-j+1-\kappa}-\binom{n-2j}{\left\lfloor \frac{k}{2}\right\rfloor-j-\kappa}\right)\cdot (-1)^\kappa\binom{n-2\left\lfloor \frac{k}{2}\right\rfloor-2+\kappa}{\kappa} \nonumber\\
        =&(-1)^{\lfloor k/2\rfloor+1}\left(\binom{2\left\lfloor \frac{k}{2}\right\rfloor-2j+1}{\left\lfloor \frac{k}{2}\right\rfloor-j+1}-\binom{2\left\lfloor \frac{k}{2}\right\rfloor-2j+1}{\left\lfloor \frac{k}{2}\right\rfloor-j}\right) \nonumber\\
        =&0, \label{SMeq:SMlem:SU2_nontrivial_existence_general2}
    \end{align}
    where the second equality can be confirmed by comparing the coefficients of $z^{\lfloor k/2\rfloor-j+1}$ in both sides of $(1-z)(1+z)^{n-2j}\cdot(1+z)^{-(n-2\lfloor k/2\rfloor-1)}=(1-z)(1+z)^{2\lfloor k/2\rfloor-2j+1}$. 
    The definition of $\bm{y}$ (Eq.~\eqref{SMeq:SMlem:SU2_nontrivial_existence_general1}) also implies that 
    \begin{align}
        &\sum_{\lambda\in\{n/2, n/2-1, ..., n/2-\lfloor k/2\rfloor-1\}} \left(\binom{n}{\frac{n}{2}-\lambda}-\binom{n}{\frac{n}{2}-\lambda-1}\right)|y_\lambda| \nonumber\\
        =&\sum_{\kappa=0}^{\lfloor k/2\rfloor+1} \left(\binom{n}{\left\lfloor \frac{k}{2}\right\rfloor+1-\kappa}-\binom{n}{\left\lfloor \frac{k}{2}\right\rfloor-\kappa}\right)\binom{n-2\left\lfloor \frac{k}{2}\right\rfloor-2+\kappa}{\kappa} \nonumber\\
        =&\sum_{\kappa=0}^{\lfloor k/2\rfloor+1} \binom{n}{\left\lfloor \frac{k}{2}\right\rfloor+1-\kappa}\binom{n-2\left\lfloor \frac{k}{2}\right\rfloor-1+\kappa}{\kappa} \nonumber\\
        =&a_{n, 2(\lfloor k/2\rfloor+1), \lfloor k/2\rfloor+1} \nonumber\\
        =&\frac{2^{\lfloor k/2\rfloor+1}}{(\left\lfloor \frac{k}{2}\right\rfloor+1)!}\prod_{\alpha=1}^{\lfloor k/2\rfloor+1} (n-2\alpha+1) \label{SMeq:SMlem:SU2_nontrivial_existence_general3}
    \end{align}
    where the second equality can be confirmed by comparing the coefficients of $z^{\lfloor k/2\rfloor+1}$ in both sides of $(1-z)(1+z)^n\cdot (1-z)^{-(n-2\lfloor k/2\rfloor-1)}=(1+z)^n/(1-z)^{n-2\lfloor k/2\rfloor-2}$, and the third and fourth equalities follow from the definition of $a_{n, k, j}$ (Eq.~\eqref{SMeq:a_nkj_def}) and Lemma~\ref{SMlem:sequence_property1}, respectively. 
    Thus $\bm{y}$ is a nontrivial integer solution of Eqs.~\eqref{SMeq:SMlem:SU2_cond1} and \eqref{SMeq:SMlem:SU2_cond2} if $t$ satisfies Eq.~\eqref{SMeq:SMlem:SU2_nontrivial_existence_general0}. 
\end{proof}

Next, we show that for sufficiently large $n$, the condition on $t$ presented in Lemma~\ref{SMlem:SU2_nontrivial_existence_general} is necessary for Eqs.~\eqref{SMeq:SMlem:SU2_cond1} and \eqref{SMeq:SMlem:SU2_cond2} in Lemma~\ref{SMlem:SU2_cond} having no nontrivial integer solution. 
The following lemma is the counterpart of Lemma~\ref{SMlem:SU2_nontrivial_nonexistence_general} in the $\mathrm{SU}(2)$ case.

\begin{lemma} \label{SMlem:SU2_nontrivial_nonexistence_general}
	Let $n, k\in\mathbb{N}$, $k\geq 2$, $n\geq 2(\lfloor k/2\rfloor+1)$, $R$ be a unitary representation of $G=\mathrm{SU}(2)$ defined by Eq.~\eqref{SMeq:SU2_representation}, and $c_{n, k}$ defined by 
    \begin{align}
        c_{n, k}:=\frac{2^{\lfloor k/2\rfloor}}{\left(\left\lfloor \frac{k}{2}\right\rfloor+1\right)!}\prod_{\alpha=1}^{\lfloor k/2\rfloor+1} (n-2\alpha+1) \label{SMeq:SU2_bound_def}
    \end{align}
    satisfy $c_{n, k}\leq \binom{n}{j+1}-\binom{n}{j}$ for all $j\in\{\lfloor k/2\rfloor+1, \lfloor k/2\rfloor+2, ..., \lfloor n/2\rfloor-1\}$. 
    Then, Eqs.~\eqref{SMeq:SMlem:SU2_cond1} and \eqref{SMeq:SMlem:SU2_cond2} have no nontrivial integer solution if and only if $t<c_{n, k}$. 
\end{lemma}

We note that the condition on $n$ and $k$ is satisfied when $n\geq 2^{2(\lfloor k/2\rfloor+1)}$, as we show in Lemma~\ref{SMlem:SU2_sufficient_cond}.

\begin{proof}
    We take arbitrary integer solution $\bm{x}$ of Eqs.~\eqref{SMeq:SMlem:SU2_cond1} and \eqref{SMeq:SMlem:SU2_cond2}, and show that $\bm{x}=\bm{0}$ when $t<c_{n, k}$. 
    By Lemma~\ref{SMlem:solution_bound} and the assumption that $c_{n, k}\leq\binom{n}{\lceil k/2\rceil+1}$ for $j=\lfloor k/2\rfloor+1$ and $\lfloor n/2\rfloor-1$, we have 
    \begin{align}
        x_\lambda=0\ \forall \lambda\in\left\{\frac{n}{2}-\left\lfloor \frac{k}{2}\right\rfloor-2, \frac{n}{2}-\left\lfloor \frac{k}{2}\right\rfloor-3, ..., \frac{n}{2}-\left\lfloor \frac{n}{2}\right\rfloor\right\}, \label{SMeq:SMlem:SU2_nontrivial_nonexistence_general1}
    \end{align}
    By noting that $\binom{n-2j}{n/2-j-\lambda}-\binom{n-2j}{n/2-j-\lambda-1}=0$ when $\lambda>n/2-j$, we can see that the linear space of $\bm{x}$ satisfying Eqs.~\eqref{SMeq:SMlem:SU2_cond2} and \eqref{SMeq:SMlem:SU2_nontrivial_nonexistence_general1} is $1$-dimensional. 
    By Eq.~\eqref{SMeq:SMlem:SU2_nontrivial_existence_general2}, $\bm{y}$ defined by Eq.~\eqref{SMeq:SMlem:SU2_nontrivial_existence_general1} is a nontrivial solution of Eq.~\eqref{SMeq:SMlem:SU2_cond2} and we can directly confirm that $\bm{y}$ also satisfies Eq.~\eqref{SMeq:SMlem:SU2_nontrivial_nonexistence_general1}. 
    Thus $\bm{x}$ can be written as $\bm{x}=r\bm{y}$ with some $r\in\mathbb{R}$. 
    Since $\bm{x}$ is an integer vector and $y_{n/2-\lfloor k/2\rfloor-1}=(-1)^{\lfloor k/2\rfloor+1}$, we have $r=x_{n/2-\lfloor k/2\rfloor-1}/y_{n/2-\lfloor k/2\rfloor-1}\in\mathbb{Z}$. 
    By Eq.~\eqref{SMeq:SMlem:SU2_nontrivial_existence_general3}, we have 
    \begin{align}
        \sum_{\lambda\in\{n/2, n/2-1, ..., n/2-\lfloor k/2\rfloor-1\}} \left(\binom{n}{\frac{n}{2}-\lambda}-\binom{n}{\frac{n}{2}-\lambda-1}\right)|x_\lambda| 
        =2|r|c_{n, k}. \label{SMeq:SMlem:SU2_nontrivial_nonexistence_general2}
    \end{align}
    When $t<c_{n, k}$, Eq.~\eqref{SMeq:SMlem:SU2_cond2} implies that 
    \begin{align}
        \sum_{\lambda\in\{n/2, n/2-1, ..., n/2-\lfloor k/2\rfloor-1\}} \left(\binom{n}{\frac{n}{2}-\lambda}-\binom{n}{\frac{n}{2}-\lambda-1}\right)|x_\lambda|<2c_{n, k}. \label{SMeq:SMlem:SU2_nontrivial_nonexistence_general3}
    \end{align}
    By plugging Eq.~\eqref{SMeq:SMlem:SU2_nontrivial_nonexistence_general2} into Eq.~\eqref{SMeq:SMlem:SU2_nontrivial_nonexistence_general3}, we get $r=0$, which implies that Eqs.~\eqref{SMeq:SMlem:SU2_cond1} and \eqref{SMeq:SMlem:SU2_cond2} have no nontrivial integer solution when $t<c_{n, k}$. 
\end{proof}

By using the lemmas above, we can prove Theorem~\ref{SMthm:SU2_general_locality} as follows: \\

\noindent
\textit{Proof of Theorem~\ref{SMthm:SU2_general_locality}.}
    In Lemma~\ref{SMlem:SU2_cond}, we have explicitly rewritten the equations in Theorem~\ref{SMthm:general} in the $\mathrm{SU}(2)$ case. 
    When $t$ does not satisfy Eq.~\eqref{SMeq:SMthm:SU2_general_locality1}, by Lemma~\ref{SMlem:SU2_nontrivial_existence_general}, there exists a nontrivial integer solution for all $k\geq 2(\lfloor k/2\rfloor+1)$. 
    When $t$ satisfies Eq.~\eqref{SMeq:SMthm:SU2_general_locality1}, by Lemma~\ref{SMlem:SU2_nontrivial_nonexistence_general}, there exists no nontrivial integer solution under a certain assumption about $n$ and $k$, which is guaranteed when $n\geq 2^{2(\lfloor k/2\rfloor+1)}$ by Lemma~\ref{SMlem:SU2_sufficient_cond}. 
\hfill $\square$\\

For the proof of Theorem~\ref{SMthm:SU2_small_locality}, we directly consider the condition on $t$ such that Eqs.~\eqref{SMeq:SMlem:SU2_cond1} and \eqref{SMeq:SMlem:SU2_cond2} have no nontrivial integer solution for the region of $n$ where we cannot use Lemma~\ref{SMlem:SU2_nontrivial_nonexistence_general}. \\

\noindent
\textit{Proof of Theorem~\ref{SMthm:SU2_small_locality}.}
    By Lemma~\ref{SMlem:SU2_cond}, the distribution of the $(G, R)$-symmetric $k$-local random circuit forming an asymptotic unitary $t$-design if and only if Eqs.~\eqref{SMeq:SMlem:SU2_cond1} and \eqref{SMeq:SMlem:SU2_cond2} have no nontrivial integer solution.

    First, we consider the case when $k=2$ or $k=3$. 
    Since Eqs.~\eqref{SMeq:SMlem:SU2_cond1} and \eqref{SMeq:SMlem:SU2_cond2} are the same when $k=2$ and when $k=3$, it is sufficient to consider the case when $k=2$. 
    We note that the common assumption $n\leq 2(\lfloor k/2\rfloor+1)$ in Lemmas~\ref{SMlem:SU2_nontrivial_existence_general} and \ref{SMlem:SU2_nontrivial_nonexistence_general} is satisfied when $n\geq 4$. 
    We also note that when $n=4$, $5$, or $n\geq 11$, the other assumption in Lemma~\ref{SMlem:SU2_nontrivial_existence_general} holds, i.e., 
    \begin{align}
        (n-1)(n-3)
        \leq\binom{n}{j+1}-\binom{n}{j}\ \forall j\in\left\{2, 3, ..., \left\lfloor \frac{n}{2}\right\rfloor-1\right\}. \label{SMeq:SMthm:SU2_small_locality1}
    \end{align} 
    By combining Lemmas~\ref{SMlem:SU2_nontrivial_existence_general} and \ref{SMlem:SU2_nontrivial_nonexistence_general}, we get the conclusion. 
    The proof of Eq.~\eqref{SMeq:SMthm:SU2_small_locality1} is as follows: 
    When $n=4$ or $5$, since the set $\{2, 3, ..., \lfloor n/2\rfloor-1\}$ is empty, Eq.~\eqref{SMeq:SMthm:SU2_small_locality1} trivially holds. 
    When $n\geq 11$, it is sufficient to show that $\binom{n}{j+1}-\binom{n}{j}\geq (n-1)(n-3)$ only for $j=2$ and $j=\lfloor n/2\rfloor-1$ by Lemma~\ref{SMlem:binomial_SU2}. 
    For the proof of the case of $j=2$, we have 
    \begin{align}
        \binom{n}{3}-\binom{n}{2}
        =n(n-1)\cdot\frac{n-5}{6} 
        \geq (n-1)(n-3)\cdot 1 
        =(n-1)(n-3). 
    \end{align}
    For the proof of the case of $j=\lfloor n/2\rfloor-1$, when $n=11$, we can directly confirm that 
    \begin{align}
        \binom{n}{\left\lfloor \frac{n}{2}\right\rfloor}-\binom{n}{\left\lfloor \frac{n}{2}\right\rfloor-1}
        =\binom{11}{5}-\binom{11}{4}=132
        \geq 80
        =(11-1)(11-3)
        =(n-1)(n-3). 
    \end{align}
    When $n\geq 12$, we have 
    \begin{align}
        \binom{n}{\left\lfloor \frac{n}{2}\right\rfloor}-\binom{n}{\left\lfloor \frac{n}{2}\right\rfloor-1}
        =&\frac{n(n-1)(n-2)(n-3)}{5!} \left(\prod_{\alpha=6}^{\lfloor n/2\rfloor} \frac{n-2\left\lfloor \frac{n}{2}\right\rfloor+2+\alpha}{\alpha}\right)\left(n-2\left\lfloor \frac{n}{2}\right\rfloor+1\right) \nonumber\\
        \geq&\frac{12\cdot (n-1)\cdot 10\cdot (n-3)}{120} \nonumber\\
        =&(n-1)(n-3). 
    \end{align}

    For $n=3$, $6$, $7$, $8$, $9$, and $10$, we get the conclusion by explicitly writing down the equations in Lemma~\ref{SMlem:SU2_cond}.   
    \begin{itemize}
    \item 
    When $n=3$, Eq.~\eqref{SMeq:SMlem:SU2_cond2} is explicitly rewritten as 
    \begin{align}
        &x_{3/2}+2x_{1/2}=0, \\
        &x_{3/2}=0,  
    \end{align}
    which implies that $x_{3/2}=x_{1/2}=0$. 
    Therefore, Eqs.~\eqref{SMeq:SMlem:SU2_cond1} and \eqref{SMeq:SMlem:SU2_cond2} do not have a nontrivial solution for all $t\in\mathbb{N}$.

    \item
	When $n=6$. 
    Eqs.~\eqref{SMeq:SMlem:SU2_cond1} and \eqref{SMeq:SMlem:SU2_cond2} are explicitly written as 
	\begin{align}
        &|x_3|+5|x_2|+9|x_1|+5|x_0|\leq 2t, \label{SMeq:SMlem:SU2_sym_n6_nontrivial1}\\
		&x_3+5x_2+9x_1+5x_0=0, \label{SMeq:SMlem:SU2_sym_n6_nontrivial2}\\
        &x_2+3x_1+2x_0=0. \label{SMeq:SMlem:SU2_sym_n6_nontrivial3}
	\end{align}
    If $t\geq 10$, Eqs.~\eqref{SMeq:SMlem:SU2_cond1} and \eqref{SMeq:SMlem:SU2_cond2} have a nontrivial integer solution $(x_3, x_2, x_1, x_0)=(1, -1, 1, -1)$. 
    If $t<10$, any integer solution $\bm{x}$ satisfies $|x_1|\leq 1$ and $|x_0|\leq 1$ by Lemma~\ref{SMlem:solution_bound}. 
    We thus have $(x_1, x_0)=\pm (1, 1), \pm (1, 0), \pm (0, 1), \pm (1, -1)$, or $(0, 0)$, which implies $|x_3|+5|x_2|+9|x_1|+5|x_0|=50, 30, 20, 20$, or $0$, respectively. 
    By combining this with Eq.~\eqref{SMeq:SMlem:SU2_cond1}, we get $\bm{x}=\bm{0}$. 
    We can therefore conclude that Eqs.~\eqref{SMeq:SMlem:SU2_cond1} and \eqref{SMeq:SMlem:SU2_cond2} have no nontrivial integer solution.

    \item
	When $n=7$. 
    Eqs.~\eqref{SMeq:SMlem:SU2_cond1} and \eqref{SMeq:SMlem:SU2_cond2} are explicitly written as 
	\begin{align}
        &|x_{7/2}|+6|x_{5/2}|+14|x_{3/2}|+14|x_{1/2}|\leq 2t, \label{SMeq:SMlem:SU2_sym_n7_nontrivial1}\\
		&x_{7/2}+6x_{5/2}+14x_{3/2}+14x_{1/2}=0, \label{SMeq:SMlem:SU2_sym_n7_nontrivial2}\\
		&x_{5/2}+4x_{3/2}+5x_{1/2}=0. \label{SMeq:SMlem:SU2_sym_n7_nontrivial3}
	\end{align}
    If $t\geq 20$, Eqs.~\eqref{SMeq:SMlem:SU2_cond1} and \eqref{SMeq:SMlem:SU2_cond2} have a nontrivial integer solution $(x_{7/2}, x_{5/2}, x_{3/2}, x_{1/2})=(6, -1, -1, 1)$. 
    If $t<20$, any integer solution $\bm{x}$ satisfies $|x_{3/2}|\leq 1$ and $|x_{1/2}|\leq 1$ by Lemma~\ref{SMlem:solution_bound}. 
    We thus have $(x_{3/2}, x_{1/2})=\pm (1, 1), \pm (1, 0), \pm (0, 1), \pm (1, -1)$, or $(0, 0)$, which implies $|x_{7/2}|+6|x_{5/2}|+14|x_{3/2}|+14|x_{1/2}|=108, 48, 60, 40$, or $0$, respectively. 
    By combining this with Eq.~\eqref{SMeq:SMlem:SU2_cond1}, we get $\bm{x}=\bm{0}$.

    \item
	When $n=8$. 
    Eqs.~\eqref{SMeq:SMlem:SU2_cond1} and \eqref{SMeq:SMlem:SU2_cond2} are explicitly written as 
	\begin{align}
        &|x_4|+7|x_3|+20|x_2|+28|x_1|+14|x_0|\leq 2t, \label{SMeq:SMlem:SU2_sym_n8_nontrivial1}\\
		&x_4+7x_3+20x_2+28x_1+14x_0=0, \label{SMeq:SMlem:SU2_sym_n8_nontrivial2}\\
		&x_3+5x_2+9x_1+5x_0=0. \label{SMeq:SMlem:SU2_sym_n8_nontrivial3}
	\end{align}
    If $t\geq 20$, Eqs.~\eqref{SMeq:SMlem:SU2_cond1} and \eqref{SMeq:SMlem:SU2_cond2} have a nontrivial integer solution $(x_4, x_3, x_2, x_1, x_0)=(6, 0, -1, 0, 1)$. 
    If $t<20$, any integer solution $\bm{x}$ satisfies $x_2=x_1=0$ and $|x_0|\leq 1$ by Lemma~\ref{SMlem:solution_bound}. 
    We thus have $x_4=21x_0$ and $x_3=-5x_0$, which implies $|x_4|+7|x_3|+14|x_0|=70|x_0|$. 
    By combining this with Eq.~\eqref{SMeq:SMlem:SU2_cond1}, we get $\bm{x}=\bm{0}$.

    \item
	When $n=9$. 
    Eqs.~\eqref{SMeq:SMlem:SU2_cond1} and \eqref{SMeq:SMlem:SU2_cond2} are explicitly written as 
    \begin{align}
        &|x_{9/2}|+8|x_{7/2}|+27|x_{5/2}|+48|x_{3/2}|+42|x_{1/2}|\leq 2t, \label{SMeq:SMlem:SU2_sym_n9_nontrivial1}\\
		&x_{9/2}+8x_{7/2}+27x_{5/2}+48x_{3/2}+42x_{1/2}=0, \label{SMeq:SMlem:SU2_sym_n9_nontrivial2}\\
		&x_{7/2}+6x_{5/2}+14x_{3/2}+14x_{1/2}=0. \label{SMeq:SMlem:SU2_sym_n9_nontrivial3}
	\end{align}    
    If $t\geq 48$, the existence of a nontrivial integer solution of Eqs.~\eqref{SMeq:SMlem:SU2_cond1} and \eqref{SMeq:SMlem:SU2_cond2} has been proven in Lemma~\ref{SMlem:SU2_nontrivial_existence_general}. 
    If $t<48$, any integer solution $\bm{x}$ satisfies $x_{3/2}=0$, $|x_{5/2}|\leq 1$, and $|x_{1/2}|\leq 1$ by Lemma~\ref{SMlem:solution_bound}. 
    We thus have $(x_{5/2}, x_{1/2})=\pm (1, 1), \pm (1, 0), \pm (0, 1), \pm (1, -1)$, or $(0, 0)$, which implies $|x_{9/2}|+8|x_{7/2}|+27|x_{5/2}|+48|x_{3/2}|+42|x_{1/2}|=320, 96, 224, 182$, or $0$, respectively. 
    By combining this with Eq.~\eqref{SMeq:SMlem:SU2_cond1}, we get $\bm{x}=\bm{0}$.

    \item
	When $n=10$. 
    Eqs.~\eqref{SMeq:SMlem:SU2_cond1} and \eqref{SMeq:SMlem:SU2_cond2} are explicitly written as 
	\begin{align}
        &|x_5|+9|x_4|+35|x_3|+75|x_2|+90|x_1|+42|x_0|\leq 2t, \label{SMeq:SMlem:SU2_sym_n10_nontrivial1}\\
		&x_5+9x_4+35x_3+75x_2+90x_1+42x_0=0, \label{SMeq:SMlem:SU2_sym_n10_nontrivial2}\\
		&x_4+7x_3+20x_2+28x_1+14x_0=0. \label{SMeq:SMlem:SU2_sym_n10_nontrivial3}
	\end{align}
    If $t\geq 63$, the existence of a nontrivial integer solution of Eqs.~\eqref{SMeq:SMlem:SU2_cond1} and \eqref{SMeq:SMlem:SU2_cond2} has been proven in Lemma~\ref{SMlem:SU2_nontrivial_existence_general}. 
    If $t<63$, any integer solution $\bm{x}$ satisfies $x_2=x_1=0$, $|x_3|\leq 1$, and $|x_3|\leq 1$ by Lemma~\ref{SMlem:solution_bound}. 
    We thus have $(x_3, x_1)=\pm (1, 1), \pm (1, 0), \pm (0, 1), \pm (1, -1)$, or $(0, 0)$, which implies $|x_5|+9|x_4|+35|x_3|+75|x_2|+90|x_1|+42|x_0|=378, 126, 252, 196$, or $0$, respectively. 
    By combining this with Eq.~\eqref{SMeq:SMlem:SU2_cond1}, we get $\bm{x}=\bm{0}$. 
    \end{itemize}

    Next, we consider the case when $k=4$. 
    We note that when $n=6$, $7$, or $n\geq 18$, the assumption in Lemma~\ref{SMlem:SU2_nontrivial_nonexistence_general} holds, i.e., 
    \begin{align}
        \frac{2}{3}(n-1)(n-3)(n-5) 
        \leq\binom{n}{j+1}-\binom{n}{j}\ \forall j\in\left\{3, 4, ..., \left\lfloor \frac{n}{2}\right\rfloor-1\right\}. \label{SMeq:SMthm:SU2_small_locality2}
    \end{align} 
    Thus we can use Lemma~\ref{SMlem:SU2_nontrivial_nonexistence_general}, and by combining it with lemma~\ref{SMlem:SU2_nontrivial_existence_general}, we get the conclusion. 
    The proof of Eq.~\eqref{SMeq:SMthm:SU2_small_locality2} is as follows: 
    When $n=6$ or $7$, since the set $\{3, 4, ..., \lfloor n/2\rfloor-1\}$ is empty, Eq.~\eqref{SMeq:SMthm:SU2_small_locality2} trivially holds. 
    When $n\geq 18$, by Lemma~\ref{SMlem:binomial_SU2}, it is sufficient to show that $\binom{n}{j+1}-\binom{n}{j}\geq 2(n-1)(n-3)(n-5)/3$ only for $j=3$ and $j=\lfloor n/2\rfloor-1$. 
    For the proof of the case of $j=3$, we have 
    \begin{align}
        \binom{n}{4}-\binom{n}{3}
        =&\frac{n(n-1)(n-2)(n-7)}{24} \nonumber\\
        =&\frac{2}{3}(n-1)\cdot\frac{n-2}{16}\cdot n(n-7) \nonumber\\
        \geq&\frac{2}{3}(n-1)\cdot 1\cdot (n-3)(n-5) \nonumber\\
        =&\frac{2}{3}(n-1)(n-3)(n-5). 
    \end{align}
    For the proof of $j=\lfloor n/2\rfloor-1$, we have 
    \begin{align}
        \binom{n}{\left\lfloor \frac{n}{2}\right\rfloor}-\binom{n}{\left\lfloor \frac{n}{2}\right\rfloor-1}
        =&\frac{n(n-1)(n-2)(n-3)(n-4)(n-5)}{7!} \left(\prod_{\alpha=8}^{\lfloor n/2\rfloor} \frac{n-2\left\lfloor \frac{n}{2}\right\rfloor+2+\alpha}{\alpha}\right)\left(n-2\left\lfloor \frac{n}{2}\right\rfloor+1\right) \nonumber\\
        \geq&\frac{18\cdot (n-1)\cdot 16\cdot (n-3)\cdot 14\cdot (n-5)}{7!} \nonumber\\
        =&\frac{4}{5}(n-1)(n-3)(n-5) \nonumber\\
        \geq&\frac{2}{3}(n-1)(n-3)(n-5). 
    \end{align}

    In the following, we confirm the results for $n=5$, $8$, $9$, ..., $17$ by explicitly writing down the equations in Lemma~\ref{SMlem:SU2_cond}. 
    \begin{itemize}
    \item 
    When $n=5$, Eq.~\eqref{SMeq:SMlem:SU2_cond2} is explicitly written as 
    \begin{align}
        &x_{5/2}+4x_{3/2}+5x_{1/2}=0, \\
        &x_{3/2}+2x_{1/2}=0, \\
        &x_{1/2}=0, 
    \end{align}
    which implies that $x_{5/2}=x_{3/2}=x_{1/2}=0$. 
    Therefore, Eqs.~\eqref{SMeq:SMlem:SU2_cond1} and \eqref{SMeq:SMlem:SU2_cond2} do not have a nontrivial solution for all $t\in\mathbb{N}$.

    \item
    When $n=8$, Eqs.~\eqref{SMeq:SMlem:SU2_cond1} and \eqref{SMeq:SMlem:SU2_cond2} are explicitly written as 
    \begin{align}
        &|x_4|+7|x_3|+20|x_2|+28|x_1|+14|x_0|\leq 2t, \\
        &x_4+7x_3+20x_2+28x_1+14x_0=0, \\
        &x_3+5x_2+9x_1+5x_0=0, \\
        &x_2+3x_1+2x_0=0. 
    \end{align}
    If $t\geq 35$, Eqs.~\eqref{SMeq:SMlem:SU2_cond1} and \eqref{SMeq:SMlem:SU2_cond2} have a nontrivial integer solution $(x_4, x_3, x_2, x_1, x_0)=(1, -1, 1, -1, 1)$. 
    If $t<35$, any integer solution $\bm{x}$ satisfies $|x_2|\leq 1$, $|x_1|\leq 1$, and $|x_0|\leq 2$ by Lemma~\ref{SMlem:solution_bound}. 
    We thus have $(x_2, x_1, x_0)=\pm(1, 1, -2)$, $\pm(1, -1, 1)$, or $(0, 0, 0)$, which implies $|x_4|+7|x_3|+20|x_2|+28|x_1|+14|x_0|=112$, $70$, or $0$, respectively. 
    By combining this with Eq.~\eqref{SMeq:SMlem:SU2_cond1}, we get $\bm{x}=\bm{0}$. 
    We can therefore conclude that Eqs.~\eqref{SMeq:SMlem:SU2_cond1} and \eqref{SMeq:SMlem:SU2_cond2} have no nontrivial integer solution.

    \item
    When $n=9$, Eqs.~\eqref{SMeq:SMlem:SU2_cond1} and \eqref{SMeq:SMlem:SU2_cond2} are explicitly written as 
    \begin{align}
        &|x_{9/2}|+8|x_{7/2}|+27|x_{5/2}|+48|x_{3/2}|+42|x_{1/2}|\leq 2t, \\
        &x_{9/2}+8x_{7/2}+27x_{5/2}+48x_{3/2}+42x_{1/2}=0, \\
        &x_{7/2}+6x_{5/2}+14x_{3/2}+14x_{1/2}=0, \\
        &x_{5/2}+4x_{3/2}+5x_{1/2}=0. 
    \end{align}    
    If $t\geq 90$, Eqs.~\eqref{SMeq:SMlem:SU2_cond1} and \eqref{SMeq:SMlem:SU2_cond2} have a nontrivial integer solution $(x_{9/2}, x_{7/2}, x_{5/2}, x_{3/2}, x_{1/2})=(-15, 6, -1, -1, 1)$. 
    If $t<90$, any integer solution $\bm{x}$ satisfies $|x_{5/2}|\leq 3$, $|x_{3/2}|\leq 1$, and $|x_{1/2}|\leq 2$ by Lemma~\ref{SMlem:solution_bound}. 
    We thus have $(x_{5/2}, x_{3/2}, x_{1/2})=\pm(-1, 1, -1)$ or $(0, 0, 0)$, which implies $|x_{9/2}|+8|x_{7/2}|+27|x_{5/2}|+48|x_{3/2}|+42|x_{1/2}|=180$ or $0$, respectively. 
    By combining this with Eq.~\eqref{SMeq:SMlem:SU2_cond1}, we get $\bm{x}=\bm{0}$.

    \item
    When $n=10$, Eqs.~\eqref{SMeq:SMlem:SU2_cond1} and \eqref{SMeq:SMlem:SU2_cond2} are explicitly written as 
    \begin{align}
        &|x_5|+9|x_4|+35|x_3|+75|x_2|+90|x_1|+42|x_0|\leq 2t, \\
        &x_5+9x_4+35x_3+75x_2+90x_1+42x_0=0, \\
        &x_4+7x_3+20x_2+28x_1+14x_0=0, \\
        &x_3+5x_2+9x_1+5x_0=0. 
    \end{align}
    If $t\geq 96$, Eqs.~\eqref{SMeq:SMlem:SU2_cond1} and \eqref{SMeq:SMlem:SU2_cond2} have a nontrivial integer solution $(x_5, x_4, x_3, x_2, x_1, x_0)=(-21, 6, 0, -1, 0, 1)$. 
    If $t<96$, any integer solution $\bm{x}$ satisfies $|x_3|\leq 2, |x_2|\leq 1, |x_1|\leq 1, |x_0|\leq 2$ by Lemma~\ref{SMlem:solution_bound}. 
    We thus have $(x_3, x_2, x_1, x_0)=\pm(1, -1, 1, -1), \pm(1, 0, 1, -2), \pm(0, 1, 0, -1), \textrm{ or } (0, 0, 0, 0)$, which implies $|x_5|+9|x_4|+35|x_3|+75|x_2|+90|x_1|+42|x_0|=256, 294, 192, \textrm{ or } 0$, respectively. 
    By combining this with Eq.~\eqref{SMeq:SMlem:SU2_cond1}, we get $\bm{x}=\bm{0}$.

    \item
    When $n=11$, Eqs.~\eqref{SMeq:SMlem:SU2_cond1} and \eqref{SMeq:SMlem:SU2_cond2} are explicitly written as 
    \begin{align}
        &|x_{11/2}|+10|x_{9/2}|+44|x_{7/2}|+110|x_{5/2}|
        +165|x_{3/2}|+132|x_{1/2}|\leq 2t, \\
        &x_{11/2}+10x_{9/2}+44x_{7/2}+110x_{5/2}+165x_{3/2}+132x_{1/2}=0, \\
        &x_{9/2}+8x_{7/2}+27x_{5/2}+48x_{3/2}+42x_{1/2}=0, \\
        &x_{7/2}+6x_{5/2}+14x_{3/2}+14x_{1/2}=0. 
    \end{align}
    If $t\geq 192$, Eqs.~\eqref{SMeq:SMlem:SU2_cond1} and \eqref{SMeq:SMlem:SU2_cond2} have a nontrivial integer solution $(x_{11/2}, x_{9/2}, x_{7/2}, x_{5/2}, x_{3/2}, x_{1/2})=(27, -6, 0, 0, 1, -1)$. 
    If $t<192$, any integer solution $\bm{x}$ satisfies $|x_{7/2}|\leq 4$, $|x_{5/2}|\leq 1$, $|x_{3/2}|\leq 1$, and $|x_{1/2}|\leq 1$ by Lemma~\ref{SMlem:solution_bound}. 
    We thus have $(x_{7/2}, x_{5/2}, x_{3/2}, x_{1/2})=\pm(0, 0, 1, -1) \textrm{ or } (0, 0, 0, 0)$, which implies $|x_{11/2}|+10|x_{9/2}|+44|x_{7/2}|+110|x_{5/2}|+165|x_{3/2}|+132|x_{1/2}|=384$ or $0$, respectively. 
    By combining this with Eq.~\eqref{SMeq:SMlem:SU2_cond1}, we get $\bm{x}=\bm{0}$.

    \item
    When $n=12$, Eqs.~\eqref{SMeq:SMlem:SU2_cond1} and \eqref{SMeq:SMlem:SU2_cond2} are explicitly written as 
    \begin{align}
        &|x_6|+11|x_5|+54|x_4|+154|x_3|+275|x_2|+297|x_1|+132|x_0|\leq 2t, \\ 
        &x_6+11x_5+54x_4+154x_3+275x_2+297x_1+132x_0=0, \\
        &x_5+9x_4+35x_3+75x_2+90x_1+42x_0=0, \\
        &x_4+7x_3+20x_2+28x_1+14x_0=0. 
    \end{align}
    If $t\geq 330$, Eqs.~\eqref{SMeq:SMlem:SU2_cond1} and \eqref{SMeq:SMlem:SU2_cond2} have a nontrivial integer solution $(x_6, x_5, x_4, x_3, x_2, x_1, x_0)=(33, -6, 0, 0, 0, 1, -2)$. 
    If $t<330$, any integer solution $\bm{x}$ satisfies $m_\lambda |x_\lambda|<330$ for all $\lambda\in\{0, 1, 2, 3, 4\}$ and $|m_\lambda x_\lambda+m_\kappa x_\kappa|<330$ for all $\lambda, \kappa\in\{0, 1, 2, 3, 4\}$ satisfying $\lambda\neq\kappa$ by Lemma~\ref{SMlem:solution_bound}. 
    We thus have $(x_6, x_5, x_4, x_3, x_2, x_1, x_0)=\pm(132, -28, 0, 2, 0, 0, -1)$, $\pm(33, -6, 0, 0, 0, 1, -2)$, $\pm(34, -7, 1, -1, 1, 0, -1)$, or $(0, 0, 0, 0, 0, 0, 0)$, which implies $|x_6|+11|x_5|+54|x_4|+154|x_3|+275|x_2|+297|x_1|+132|x_0|=880$, $660$, $726$, or $0$, respectively. 
    By combining this with Eq.~\eqref{SMeq:SMlem:SU2_cond1}, we get $\bm{x}=\bm{0}$.

    \item
    When $n=13$, Eqs.~\eqref{SMeq:SMlem:SU2_cond1} and \eqref{SMeq:SMlem:SU2_cond2} are explicitly written as 
    \begin{align}
        &|x_{13/2}|+12|x_{11/2}|+65|x_{9/2}|+208|x_{7/2}|
        +429|x_{5/2}|+572|x_{3/2}|+429|x_{1/2}|\leq 2t, \\
        &x_{13/2}+12x_{11/2}+65x_{9/2}+208x_{7/2}
        +429x_{5/2}+572x_{3/2}+429x_{1/2}=0, \\
        &x_{11/2}+10x_{9/2}+44x_{7/2}+110x_{5/2}+165x_{3/2}+132x_{1/2}=0, \\
        &x_{9/2}+8x_{7/2}+27x_{5/2}+48x_{3/2}+42x_{1/2}=0. 
    \end{align}
    If $t\geq 640$, Lemma~\ref{SMlem:SU2_nontrivial_existence_general} implies that Eqs.~\eqref{SMeq:SMlem:SU2_cond1} and \eqref{SMeq:SMlem:SU2_cond2} have a nontrivial integer solution. 
    If $t<640$, any integer solution $\bm{x}$ satisfies $m_\lambda |x_\lambda|<640$ for all $\lambda\in\{1/2, 3/2, 5/2, 7/2, 9/2\}$ and $|m_\lambda x_\lambda+m_\kappa x_\kappa|<640$ for all $\lambda, \kappa\in\{1/2, 3/2, 5/2, 7/2, 9/2\}$ satisfying $\lambda\neq\kappa$ by Lemma~\ref{SMlem:solution_bound}. 
    By these conditions and Eq.~\eqref{SMeq:SMlem:SU2_cond2} in the case of $j=2$, we get $x_{1/2}=x_{3/2}=x_{5/2}=x_{7/2}=x_{9/2}=0$. 
    By plugging this into Eq.~\eqref{SMeq:SMlem:SU2_cond2} in the cases of $j=0$ and $1$, we get $\bm{x}=\bm{0}$.

    \item
    When $n=14$, Eqs.~\eqref{SMeq:SMlem:SU2_cond1} and \eqref{SMeq:SMlem:SU2_cond2} are explicitly written as 
    \begin{align}
        &|x_7|+13|x_6|+77|x_5|+273|x_4|+637|x_3|+1001|x_2|+1001|x_1|+429|x_0|\leq 2t, \\
        &x_7+13x_6+77x_5+273x_4+637x_3+1001x_2+1001x_1+429x_0=0, \\
        &x_6+11x_5+54x_4+154x_3+275x_2+297x_1+132x_0=0, \\
        &x_5+9x_4+35x_3+75x_2+90x_1+42x_0=0. 
    \end{align}
    If $t\geq 858$, Lemma~\ref{SMlem:SU2_nontrivial_existence_general} implies that Eqs.~\eqref{SMeq:SMlem:SU2_cond1} and \eqref{SMeq:SMlem:SU2_cond2} have a nontrivial integer solution. 
    If $t<858$, any integer solution $\bm{x}$ satisfies $m_\lambda |x_\lambda|<640$ for all $\lambda\in\{0, 1, 2, 3, 4, 5\}$ and $|m_\lambda x_\lambda+m_\kappa x_\kappa|<640$ for all $\lambda, \kappa\in\{0, 1, 2, 3, 4, 5\}$ satisfying $\lambda\neq\kappa$ by Lemma~\ref{SMlem:solution_bound}. 
    By these conditions and Eq.~\eqref{SMeq:SMlem:SU2_cond2} in the case of $j=2$, we get $x_0=x_1=x_2=x_3=x_4=x_5=0$. 
    By plugging this into Eq.~\eqref{SMeq:SMlem:SU2_cond2} in the cases of $j=0$ and $1$, we get $\bm{x}=\bm{0}$.

    \item
    When $n=15$, Eqs.~\eqref{SMeq:SMlem:SU2_cond1} and \eqref{SMeq:SMlem:SU2_cond2} are explicitly written as 
    \begin{align}
        &|x_{15/2}|+14|x_{13/2}|+90|x_{11/2}|+350|x_{9/2}|+910|x_{7/2}|
        +1638|x_{5/2}|+2002|x_{3/2}|+1430|x_{1/2}|\leq 2t, \\
        &x_{15/2}+14x_{13/2}+90x_{11/2}+350x_{9/2}+910x_{7/2}
        +1638x_{5/2}+2002x_{3/2}+1430x_{1/2}=0, \\
        &x_{13/2}+12x_{11/2}+65x_{9/2}+208x_{7/2}
        +429x_{5/2}+572x_{3/2}+429x_{1/2}=0, \\
        &x_{11/2}+10x_{9/2}+44x_{7/2}+110x_{5/2}+165x_{3/2}+132x_{1/2}=0. 
    \end{align}
    If $t\geq 1120$, Lemma~\ref{SMlem:SU2_nontrivial_existence_general} implies that Eqs.~\eqref{SMeq:SMlem:SU2_cond1} and \eqref{SMeq:SMlem:SU2_cond2} have a nontrivial integer solution. 
    If $t<1120$, any integer solution $\bm{x}$ satisfies $x_{5/2}=x_{3/2}=x_{1/2}=0$, $|x_{11/2}|\leq 12$, $|x_{9/2}|\leq 3$, and $|x_{7/2}|\leq 1$ by Lemma~\ref{SMlem:solution_bound}. 
    By these conditions and Eq.~\eqref{SMeq:SMlem:SU2_cond2}, we get $(x_{11/2}, x_{9/2}, x_{7/2})=\pm(-10, 1, 0)$ or $(0, 0, 0)$, which implies $|x_{15/2}|+14|x_{13/2}|+90|x_{11/2}|+350|x_{9/2}|+910|x_{7/2}|+1638|x_{5/2}|+2002|x_{3/2}|+1430|x_{1/2}|=2240$ or $0$, respectively. 
    By combining this with Eq.~\eqref{SMeq:SMlem:SU2_cond1}, we get $\bm{x}=\bm{0}$.

    \item
    When $n=16$, Eqs.~\eqref{SMeq:SMlem:SU2_cond1} and \eqref{SMeq:SMlem:SU2_cond2} are explicitly written as 
    \begin{align}
        &|x_8|+15|x_7|+104|x_6|+440|x_5|+1260|x_4|+2548|x_3|
        +3640|x_2|+3432|x_1|+1430|x_0|\leq 2t, \\
        &x_8+15x_7+104x_6+440x_5+1260x_4+2548x_3+3640x_2+3432x_1+1430x_0=0, \\
        &x_7+13x_6+77x_5+273x_4+637x_3+1001x_2+1001x_1+429x_0=0, \\
        &x_6+11x_5+54x_4+154x_3+275x_2+297x_1+132x_0=0. 
    \end{align}
    If $t\geq 1430$, Lemma~\ref{SMlem:SU2_nontrivial_existence_general} implies that Eqs.~\eqref{SMeq:SMlem:SU2_cond1} and \eqref{SMeq:SMlem:SU2_cond2} have a nontrivial integer solution. 
    If $t<1430$, any integer solution $\bm{x}$ satisfies $x_3=x_2=x_1=x_0=0$, $|x_6|\leq 13$, $|x_5|\leq 3$, and $|x_4|\leq 1$ by Lemma~\ref{SMlem:solution_bound}. 
    By these conditions and Eq.~\eqref{SMeq:SMlem:SU2_cond2}, we get $(x_6, x_5, x_4)=\pm(-11, 1, 0)$ or $(0, 0, 0)$, which implies $|x_8|+15|x_7|+104|x_6|+440|x_5|+1260|x_4|+2548|x_3|+3640|x_2|+3432|x_1|+1430|x_0|=2860$ or $0$, respectively. 
    By combining this with Eq.~\eqref{SMeq:SMlem:SU2_cond1}, we get $\bm{x}=\bm{0}$.

    \item
    When $n=17$, Eqs.~\eqref{SMeq:SMlem:SU2_cond1} and \eqref{SMeq:SMlem:SU2_cond2} are explicitly written as 
    \begin{align}
        &|x_{17/2}|+16|x_{15/2}|+119|x_{13/2}|+544|x_{11/2}|+1700|x_{9/2}|
        +3808|x_{7/2}|+6188|x_{5/2}|+7072|x_{3/2}|+4862|x_{1/2}|\leq 2t, \\
        &x_{17/2}+16x_{15/2}+119x_{13/2}+544x_{11/2}+1700x_{9/2}
        +3808x_{7/2}+6188x_{5/2}+7072x_{3/2}+4862x_{1/2}=0, \\
        &x_{15/2}+14x_{13/2}+90x_{11/2}+350x_{9/2}+910x_{7/2}
        +1638x_{5/2}+2002x_{3/2}+1430x_{1/2}=0, \\
        &x_{13/2}+12x_{11/2}+65x_{9/2}+208x_{7/2}
        +429x_{5/2}+572x_{3/2}+429x_{1/2}=0. 
    \end{align}
    If $t\geq 1792$, Lemma~\ref{SMlem:SU2_nontrivial_existence_general} implies that Eqs.~\eqref{SMeq:SMlem:SU2_cond1} and \eqref{SMeq:SMlem:SU2_cond2} have a nontrivial integer solution. 
    If $t<1792$, any integer solution $\bm{x}$ satisfies $x_{7/2}=x_{5/2}=x_{3/2}=x_{1/2}=0$, $|x_{13/2}|\leq 15$, $|x_{11/2}|\leq 3$, and $|x_{9/2}|\leq 1$ by Lemma~\ref{SMlem:solution_bound}. 
    By these conditions and Eq.~\eqref{SMeq:SMlem:SU2_cond2}, we get $(x_{13/2}, x_{11/2}, x_{9/2})=\pm(-12, 1, 0)$ or $(0, 0, 0)$, which implies $|x_{17/2}|+16|x_{15/2}|+119|x_{13/2}|+544|x_{11/2}|+1700|x_{9/2}|+3808|x_{7/2}|+6188|x_{5/2}|+7072|x_{3/2}|+4862|x_{1/2}|=3584$ or $0$, respectively. 
    By combining this with Eq.~\eqref{SMeq:SMlem:SU2_cond1}, we get $\bm{x}=\bm{0}$. 
    \end{itemize}
\hfill $\square$

\section{Technical lemmas} \label{SMsec:technical}

In this appendix, we show several lemmas used in the proof of the main statements.

For the proofs of Lemmas~\ref{SMlem:relative_phase_space} and \ref{SMlem:design_universality}, we prepare a basic property of a compact abelian matrix Lie group.

\begin{lemma} \label{SMlem:Lie_algebra}
    Let $l\in\mathbb{N}$, $\mathcal{Y}$ be a linear subspace of $\mathbb{R}^l$, $H:=\exp(i\cdot\mathrm{diag}(\mathcal{Y}))$ be compact. 
    Then, the Lie algebra $\mathfrak{h}$ of $H$ is $\mathrm{diag}(\mathcal{Y})$. 
\end{lemma}

\begin{proof}
    Since $H$ is a Lie subgroup of $\exp(i\cdot\mathrm{diag}(\mathbb{R}^l))$, and the Lie algebra of $\exp(i\cdot\mathrm{diag}(\mathbb{R}^l))$ is $\mathrm{diag}(\mathbb{R}^l)$, the Lie algebra $\mathfrak{h}$ of $H$ is a subset of $\mathrm{diag}(\mathbb{R}^l)$. 
    Thus, $\mathfrak{h}$ is given by $\mathfrak{h}=\{A\in\mathrm{diag}(\mathbb{R}^l)\ |\ \forall \theta\in\mathbb{R}\ \exp(i\theta A)\in\exp(i\cdot\mathrm{diag}(\mathcal{Y}))\}$, and it is sufficient to show that $\mathfrak{h}=\mathrm{diag}(\mathcal{Y})$. 
    Since $\mathfrak{h}\supset\mathrm{diag}(\mathcal{Y})$ is trivial, we show that $\mathfrak{h}\subset\mathrm{diag}(\mathcal{Y})$ in the following. 
    We take arbitrary $A\in\mathfrak{h}$. 
    Then, for any $\theta\in\mathbb{R}$, $\exp(i\theta A)\in\exp(i\cdot\mathrm{diag}(\mathcal{Y}))$ . 
    When we define $\bm{a}:=\mathrm{diag}^{-1}(A)$, it can be equivalently expressed as $\theta\bm{a}=\bm{y}+2\pi\bm{c}$ with some $\bm{y}\in\mathcal{Y}$ and $\bm{c}\in\mathbb{Z}^l$. 
    We decompose $\bm{a}$ as $\bm{a}=\bm{b}+\bm{b}^\perp$ with some $\bm{b}\in\mathcal{Y}$ and $\bm{b}^\perp\in\mathcal{Y}^\perp$. 
    By taking the inner product of $\theta\bm{a}$ and $\bm{b}^\perp$, we have $\theta\|\bm{b}^\perp\|^2=2\pi\braket{\bm{c}, \bm{b}^\perp}$. 
    Since we can take such $\bm{c}\in\mathbb{Z}^l$ for all $\theta\in\mathbb{R}$, we get $\{\theta\|\bm{b}^\perp\|^2\ |\ \theta\in\mathbb{R}\}\subset\{2\pi\braket{\bm{c}, \bm{b}^\perp}\ |\ \bm{c}\in\mathbb{Z}^l\}$. 
    Since the r.h.s. of this is countable, the l.h.s. is also countable. 
    We therefore get $\|\bm{b}^\perp\|=0$, which implies that $\bm{a}\in\mathcal{Y}$, i.e., $A\in\mathrm{diag}(\mathcal{Y})$. 
\end{proof}

By using the lemma above, we show properties of the group of relative phases for the proof of Lemma~\ref{SMlem:design_universality}.

\begin{lemma} \label{SMlem:relative_phase_space}
    Let $n\in\mathbb{N}$, $R$ be a unitary representation of a group $G$, $\Lambda$ be the set of the labels of the inequivalent irreducible representations appearing in $R$, 
    $\Gamma$ be a finite set, 
    $\mathcal{S}^\gamma$ be a connected compact subgroup of $\mathcal{U}_{n, G, R}$ for all $\gamma\in\Gamma$,  
    $h$ be a function $\mathcal{L}_{n, G, R}$ to $\mathbb{C}^{\Lambda\times\Lambda}$ defined by 
    \begin{align}
        h(A):=\mathrm{diag}((\mathrm{det}(A_\lambda))_{\lambda\in\Lambda})\ \forall A\in\mathcal{L}_{n, G, R} 
    \end{align}
    with $A_\lambda$ determined by Eq.~\eqref{SMeq:sym_op_decomp}, and 
    $\widetilde{\mathcal{V}}$ be defined by Eq.~\eqref{SMeq:relative_phase_space_extended}. 
    Then, $h(\braket{\bigcup_{\gamma\in\Gamma} \mathcal{S}^\gamma})=\exp(i\cdot\mathrm{diag}(\widetilde{\mathcal{V}}))$, and there exist some $J\in\mathbb{N}$ and some orthogonal basis $\{\bm{v}_j\}_{j\in\{1, 2, .., J\}}$ of $\widetilde{\mathcal{V}}$ such that $\bm{v}_j\in\mathbb{Q}^\Lambda$ for all $j\in\{1,2, ..., J\}$. 
\end{lemma}

\begin{proof}
    First, we show that $h(\braket{\bigcup_{\gamma\in\Gamma} \mathcal{S}^\gamma})=\exp(i\cdot\mathrm{diag}(\widetilde{\mathcal{V}}))$. 
    Since $h$ satisfies $h(AB)=h(BA)$ for all $A, B\in\mathcal{L}_{n, G, R}$, we have $h(\braket{\bigcup_{\gamma\in\Gamma} \mathcal{S}^\gamma})=h(\mathcal{T})$, where $\mathcal{T}:=\{e^{i\theta}\prod_{\gamma\in\Gamma} U^\gamma\ |\ \theta\in\mathbb{R}, U^\gamma\in\mathcal{S}^\gamma\ \forall\gamma\in\Gamma\}$. 
    Thus, it is sufficient to show that 
    \begin{align}
        h(\mathcal{T})=\exp(i\cdot\mathrm{diag}(\widetilde{\mathcal{V}})). \label{SMeq:SMlem:relative_phase_space1}
    \end{align}
    We note that for any $\theta\in\mathbb{R}$ and $A^\gamma\in\mathfrak{s}^\gamma$, we have 
    \begin{align}
        h\left(e^{i\theta}\prod_{\gamma\in\Gamma} e^{iA^\gamma}\right)
        =&\mathrm{diag}\left(\left(\mathrm{det}\left(e^{i\theta}\prod_{\gamma\in\Gamma} e^{iA_\lambda^\gamma}\right)\right)_{\lambda\in\Lambda}\right) \nonumber\\
        =&\mathrm{diag}\left(\left(e^{i\theta m_\lambda}\prod_{\gamma\in\Gamma} e^{i\mathrm{tr}(A_\lambda^\gamma)}\right)_{\lambda\in\Lambda}\right) \nonumber\\
        =&\exp\left(i\cdot\mathrm{diag}\left(\left(\theta m_\lambda+\sum_{\gamma\in\Gamma} \mathrm{tr}(A_\lambda^\gamma)\right)_{\lambda\in\Lambda}\right)\right) \nonumber\\
        =&\exp\left(i\cdot\mathrm{diag}\left(\theta\bm{m}+\sum_{\gamma\in\Gamma} \bm{f}(A^\gamma)\right)\right), \label{SMeq:SMlem:relative_phase_space2}
    \end{align}
    where $A_\lambda^\gamma$ is defined by $A^\gamma=\sum_{\lambda\in\Lambda} F_\lambda(\mathrm{id}(\mathbb{C}^{r_\lambda})\otimes A_\lambda^\gamma)F_\lambda^\dag$. 
    For the proof of $h(\mathcal{T})\subset\exp(i\cdot\mathrm{diag}(\widetilde{\mathcal{V}}))$, for any $U\in\mathcal{T}$, we can take some $\theta\in\mathbb{R}$ and $U^\gamma\in\mathcal{S}^\gamma$ such that $U=e^{i\theta}\prod_{\gamma\in\Gamma} U^\gamma$. 
    For each $\gamma\in\Gamma$, since $\mathcal{S}^\gamma$ is connected and compact, we can take $A^\gamma$ such that $e^{iA^\gamma}=U^\gamma$. 
    By Eq.~\eqref{SMeq:SMlem:relative_phase_space2}, we have 
    \begin{align}
        h(U)
        =h\left(e^{i\theta}\prod_{\gamma\in\Gamma} e^{iA^\gamma}\right)
        =\exp\left(i\cdot\mathrm{diag}\left(\theta\bm{m}+\sum_{\gamma\in\Gamma} \bm{f}(A^\gamma)\right)\right) 
        \in&\exp(i\cdot\mathrm{diag}(\widetilde{\mathcal{V}})). \label{SMeq:SMlem:relative_phase_space3}
    \end{align}
    For the proof of $h(\mathcal{T})\supset\exp(i\cdot\mathrm{diag}(\widetilde{\mathcal{V}}))$, we take arbitrary $\bm{v}\in\widetilde{\mathcal{V}}$. 
    By the definition of $\widetilde{\mathcal{V}}$, we can take some $\theta\in\mathbb{R}$ and $A^\gamma\in\mathfrak{s}^\gamma$ such that $\bm{v}=\theta\bm{m}+\sum_{\gamma\in\Gamma} \bm{f}(A^\gamma)$. 
    By Eq.~\eqref{SMeq:SMlem:relative_phase_space2}, we have 
    \begin{align}
        \exp(i\cdot\mathrm{diag}(\bm{v})) 
        =\exp\left(i\cdot\mathrm{diag}\left(\theta\bm{m}+\sum_{\gamma\in\Gamma} \bm{f}(A^\gamma)\right)\right) 
        =h\left(e^{i\theta}\prod_{\gamma\in\Gamma} e^{iA^\gamma}\right) 
        \in h(\mathcal{T}). \label{SMeq:SMlem:relative_phase_space4}
    \end{align}
    By Eqs.~\eqref{SMeq:SMlem:relative_phase_space3} and \eqref{SMeq:SMlem:relative_phase_space4}, we get Eq.~\eqref{SMeq:SMlem:relative_phase_space1}.

    Next, we show that there exist some $J\in\mathbb{N}$ and some orthogonal basis $\{\bm{v}_j\}_{j=1}^J$ of $\widetilde{\mathcal{V}}$ such that $\bm{v}_j\in\mathbb{Q}^\Lambda$. 
    Since $\mathcal{T}$ is a finite product of compact set, $\mathcal{T}$ is compact. 
    The continuity of $h$ implies that $h(\mathcal{T})$ is compact. 
    By the construction of $h$ and $\mathcal{T}$, $h(\mathcal{T})$ is an abelian Lie group. 
    By Corollary~1.103 of Ref.~\cite{knapp2002lie}, $h(\mathcal{T})$ is isomorphic to a torus $\exp(i\cdot\mathrm{diag}(\mathbb{R}^J))$ with some $J\in\mathbb{N}$, where $J$ is the dimension of a torus. 
    We take an isomorphism $\phi$ from $\exp(i\cdot\mathrm{diag}(\mathbb{R}^J))$ to $h(\mathcal{T})$. 
    By Eq.~\eqref{SMeq:SMlem:relative_phase_space1}, $\phi$ gives an isomorphism from $\exp(i\cdot\mathrm{diag}(\mathbb{R}^J))$ to $\exp(i\cdot\mathrm{diag}(\widetilde{\mathcal{V}}))$. 
    By Lemma~\ref{SMlem:Lie_algebra}, the Lie algebras of these two Lie groups are $\mathrm{diag}(\mathbb{R}^J)$ and $\mathrm{diag}(\widetilde{\mathcal{V}})$, respectively. 
    Thus, the derivative $d\phi$ of $\phi$ at the identity gives an isomorphism from $\mathrm{diag}(\mathbb{R}^J)$ to $\mathrm{diag}(\widetilde{\mathcal{V}})$. 
    We denote the standard basis of $\mathbb{R}^J$ by $\{\bm{u}_j\}_{j=1}^J$, and define $\widetilde{\bm{u}}_j\in\mathbb{R}^\Lambda$ by $\mathrm{diag}(\widetilde{\bm{u}}_j):=d\phi(\mathrm{diag}(\bm{u}_j))$ for all $j\in\{1, 2, ..., J\}$. 
    Then, $\{\widetilde{\bm{u}}_j\}_{j=1}^J$ is a basis of $\widetilde{\mathcal{V}}$. 
    By noting that $\exp(i2\pi\cdot\mathrm{diag}(\widetilde{\bm{u}}_j))=\phi(\exp(i2\pi\cdot\mathrm{diag}(\bm{u}_j)))=I$, we have $\widetilde{\bm{u}}_j\in\mathbb{Z}^\Lambda$. 
    We get an orthogonal basis $\{\bm{v}_j\}_{j\in\{1, 2, ..., J\}}$ of $\widetilde{\mathcal{V}}$ by the Gram-Schmidt orthogonalization, i.e., $\bm{v}_j:=\widetilde{\bm{u}}_j-\sum_{j'=1}^{j-1} (\braket{\bm{v}_{j'}, \widetilde{\bm{u}}_j}/\|\bm{v}_{j'}\|^2)\bm{v}_{j'}$, and the basis vectors satisfy $\bm{v}\in\mathbb{Q}^\Lambda$ for all $j\in\{1, 2, ..., J\}$. 
\end{proof}

By using the two lemmas above, we show the equivalent conditions to the universality of the gate set. 
We use the following lemma in the explanation below Theorem~\ref{SMthm:general}.

\begin{lemma} \label{SMlem:design_universality}
    Let $n\in\mathbb{N}$, 
    $R$ be a unitary representation of a group $G$ on $n$ qudits, 
    $\Lambda$ be the set of the labels of the inequivalent irreducible representations appearing in $R$, 
    $\Gamma$ be a finite set, 
    $\mathcal{S}^\gamma$ be a connected compact subgroup of $\mathcal{U}_{n, G, R}$ for all $\gamma\in\Gamma$, 
    $\bigcup_{\gamma\in\Gamma} \mathcal{S}^\gamma$ be semi-universal for $\mathcal{U}_{n, G, R}$, and 
    $\widetilde{\mathcal{V}}$ be defined by Eq.~\eqref{SMeq:relative_phase_space_extended}. 
    Then, the following three statements are equivalent: 
    \begin{align}
        (i)\ & \bigcup_{\gamma\in\Gamma} \mathcal{S}^\gamma \textrm{ is universal for } \mathcal{U}_{n, G, R} \textrm{ up to the global phase}, \\
        (ii)\ & \widetilde{\mathcal{V}}=\mathbb{R}^\Lambda, \\
        (iii)\ & \widetilde{\mathcal{V}}^\perp\cap\mathbb{Z}^\Lambda=\{\bm{0}\}.
    \end{align}
\end{lemma}

\begin{proof}
    First, we prove that $(i)\iff (ii)$. 
    We note that $(i)$ is equivalent to 
    \begin{align}
        \Braket{\left\{e^{i\theta} I\right\}_{\theta\in\mathbb{R}}\cup\bigcup_{\gamma\in\Gamma} \mathcal{S}^\gamma} 
        \supset\left\{\sum_{\lambda\in\Lambda} F_\lambda(I\otimes U_\lambda)F_\lambda^\dag\ \middle|\ U_\lambda\in\mathrm{U}(m_\lambda)\ \forall \lambda\in\Lambda\right\}. \label{SMeq:SMlem:design_universality1}
    \end{align}
    Since $\bigcup_{\gamma\in\mathcal{S}^\gamma}$ is semi-universal for $\mathcal{U}_{n, G, R}$, the group $\{e^{i\theta}I\}_{\theta\in\mathbb{R}}\cup\bigcup_{\gamma\in\Gamma} \mathcal{S}^\gamma$ is also semi-universal for $\mathcal{U}_{n, G, R}$, which can be expressed as 
    \begin{align}
        \Braket{\left\{e^{i\theta}I\right\}_{\theta\in\mathbb{R}}\cup\bigcup_{\gamma\in\Gamma} \mathcal{S}^\gamma} 
        \supset\left\{\sum_{\lambda\in\Lambda} F_\lambda(I\otimes U_\lambda)F_\lambda^\dag\ \middle|\ U_\lambda\in\mathrm{SU}(m_\lambda)\ \forall \lambda\in\Lambda\right\}, \label{SMeq:SMlem:design_universality2}
    \end{align}
    as we have shown in Eq.~\eqref{SMeq:semi_universality_equivalence}.     
    By comparing Eqs.~\eqref{SMeq:SMlem:design_universality1} and \eqref{SMeq:SMlem:design_universality2}, $(i)$ is equivalent to $h(\braket{\{e^{i\theta} I\}_{\theta\in\mathbb{R}}\cup\bigcup_{\gamma\in\Gamma} \mathcal{S}^\gamma})=\exp(i\cdot\mathrm{diag}(\mathbb{R}^\Lambda))$, which can be equivalently written as 
    \begin{align}
        \exp(i\cdot\mathrm{diag}(\widetilde{\mathcal{V}}))=\exp(i\cdot\mathrm{diag}(\mathbb{R}^\Lambda)). \label{SMeq:SMlem:design_universality3}
    \end{align} 
    by Lemma~\ref{SMlem:relative_phase_space}. 
    It is thus sufficient to show that Eq.~\eqref{SMeq:SMlem:design_universality3}$\iff (ii)$. 
    Since the Lie algebras of $\exp(i\cdot\mathrm{diag}(\widetilde{\mathcal{V}}))$ and $\exp(i\cdot\mathrm{diag}(\mathbb{R}^\Lambda))$ are $\mathrm{diag}(\widetilde{\mathcal{V}})$ and $\mathrm{diag}(\mathbb{R}^\Lambda)$, respectively by Lemma~\ref{SMlem:Lie_algebra}, taking the Lie algebras of the both sides of Eq.~\eqref{SMeq:SMlem:design_universality3} gives the proof of Eq.~\eqref{SMeq:SMlem:design_universality3}$\implies (ii)$. 
    The converse $(ii)\implies$Eq.~\eqref{SMeq:SMlem:design_universality3} is trivial.

    Next, we prove that $(ii)\iff (iii)$. 
    The proof of $(ii)\implies (iii)$ is trivial, because $(ii)$ implies $\widetilde{\mathcal{V}}^\perp=\{\bm{0}\}$. 
    In the following, we show that $(iii)\implies (ii)$. 
    We suppose that $(ii)$ does not hold. 
    Then, we can take some $\bm{e}_l\not\in\widetilde{\mathcal{V}}$ from the standard basis $\{\bm{e}_j\}_{j=1}^J$ of $\widetilde{\mathcal{V}}$. 
    By Lemma~\ref{SMlem:relative_phase_space}, we can take an orthogonal basis $\{\bm{v}_j\}_{j=1}^J$ of $\widetilde{\mathcal{V}}$ such that $\bm{v}_j\in\mathbb{Q}^\Lambda$ for all $j\in\{1, 2, ..., J\}$. 
    By using this basis, we define $\widetilde{\bm{d}}:=\bm{e}_l-\sum_{j=1}^J (\braket{\bm{v}_j, \bm{e}_l}/\|\bm{v}_j\|^2)\bm{v}_j\in\mathbb{Q}^\Lambda\backslash\{\bm{0}\}$. 
    We can take some $\alpha\in\mathbb{Q}\backslash\{0\}$ such that $\alpha\widetilde{\bm{d}}\in\mathbb{Z}^\Lambda$. 
    Then, $\bm{d}:=\alpha\widetilde{\bm{d}}$ satisfies $\bm{d}\in (\widetilde{\mathcal{V}}^\perp\cap\mathbb{Z}^\Lambda)\backslash\{\bm{0}\}$. 
\end{proof}

For the proof of Lemma~\ref{SMlem:design_commutant}, we show that the moment operator defined by Eq.~\eqref{SMeq:moment_operator_def} is a projection.

\begin{lemma} \label{SMlem:moment_op_projection}
    Let $n, t\in\mathbb{N}$ and $\mathcal{X}$ be a compact unitary subgroup of $\mathcal{U}_n$. 
    Then, 
    \begin{align}
        M_{t, \mu_{\mathcal{X}}}=\Pi_{E(\mathrm{Comm}(\Omega_t(\mathcal{X})))}. 
    \end{align}
    where $E$ is defined by Eq.~\eqref{SMeq:E_def}, i.e., 
    \begin{align}
        &E(K):=(K\otimes I)\ket{\eta}\ \forall K\in\mathcal{L}(\mathcal{H}^{\otimes n}), \\
        &\ket{\eta}:=\frac{1}{\sqrt{d^{tn}}}\sum_{j=1}^{d^{tn}} \ket{j}\otimes\ket{j}, 
    \end{align}
    and $\{\ket{j}\}_{j=1}^{d^{tn}}$ is an orthonormal basis of $\mathcal{H}^{\otimes t}$. 
\end{lemma}

\begin{proof}
    First, we show that $M_{t, \mu_\mathcal{X}}$ is a projection. 
    By Corollary~8.31 of Ref.~\cite{knapp2002lie}, the compactness of $\mathcal{X}$ implies that $\mathcal{X}$ is unimodular, and thus $\mu_\mathcal{X}$ is also a right-invariant Haar measure.
    Then, we have 
    \begin{align}
        M_{t, \mu_{\mathcal{X}}} 
        =\int_{U\in\mathcal{X}} U^{\otimes t}\otimes U^{*\otimes t} d\mu_{\mathcal{X}}(U) 
        =\int_{U\in\mathcal{X}} (U^\dag)^{\otimes t}\otimes (U^\dag)^{*\otimes t} d\mu_{\mathcal{X}}(U) 
        =M_{t, \mu_{\mathcal{X}}}^\dag. 
    \end{align}
    Since $\mu_{\mathcal{X}}$ is left-invariant, we have 
    \begin{align}
        M_{t, \mu_{\mathcal{X}}}^2 
        =&\left(\int_{V\in\mathcal{X}} V^{\otimes t}\otimes V^{*\otimes t} d\mu_{\mathcal{X}}(V)\right)
        \left(\int_{U\in\mathcal{X}} U^{\otimes t}\otimes U^{*\otimes t} d\mu_{\mathcal{X}}(U)\right) \nonumber\\
        =&\int_{V\in\mathcal{X}} \int_{U\in\mathcal{X}} (VU)^{\otimes t}\otimes (VU)^{*\otimes t} d\mu_{\mathcal{X}}(U) d\mu_{\mathcal{X}}(U) \nonumber\\
        =&\int_{V\in\mathcal{X}} \int_{U\in\mathcal{X}} U^{\otimes t}\otimes U^{*\otimes t} d\mu_{\mathcal{X}}(U) d\mu_{\mathcal{X}}(V) \nonumber\\
        =&\int_{U\in\mathcal{X}} U^{\otimes t}\otimes U^{*\otimes t} d\mu_{\mathcal{X}}(U) \nonumber\\
        =&M_{t, \mu_{\mathcal{X}}}. 
    \end{align}
    These two relations imply that $M_{t, \mu_\mathcal{X}}$ is a projection.

    Next, we show that the projection space is $E(\mathrm{Comm}(\Omega_t(\mathcal{X})))$. 
    We note that 
    \begin{align}
        (U^{\otimes t}\otimes U^{*\otimes t})(L\otimes I)\ket{\eta} 
        =(U^{\otimes t}\otimes U^{*\otimes t})(L\otimes I)(U^{\otimes t}\otimes U^{*\otimes t})^\dag\ket{\eta} 
        =(U^{\otimes t}L U^{\dag\otimes t}\otimes I)\ket{\eta}. 
    \end{align}
    By taking the Haar integral for $U\in\mathcal{X}$, we get 
    \begin{align}
        M_{t, \mu_\mathcal{X}}(L\otimes I)\ket{\eta} 
        =\left[\left(\int_{U\in\mathcal{X}} U^{\otimes t}LU^{\dag\otimes t}d\mu_\mathcal{X}(U)\right)\otimes I\right]\ket{\eta}. \label{SMeq:SMlem:moment_op_projection1}
    \end{align}
    For the proof of $\{\ket{\Psi}\in\mathcal{H}^{\otimes 2t}\ |\ M_{t, \mu_\mathcal{X}}\ket{\Psi}=\ket{\Psi}\}\supset E(\mathrm{Comm}(\Omega_t(\mathcal{X})))$, we take arbitrary $\ket{\Psi}\in E(\mathrm{Comm}(\Omega_t(\mathcal{X})))$. 
    Then, $\ket{\Psi}$ can be written as $\ket{\Psi}=(L\otimes I)\ket{\eta}$ with some $L\in\mathrm{Comm}(\Omega_t(\mathcal{X}))$. 
    Thus, by using Eq.~\eqref{SMeq:SMlem:moment_op_projection1}, we get 
    \begin{align}
        M_{t, \mu_\mathcal{X}}\ket{\Psi} 
        =M_{t, \mu_\mathcal{X}}(L\otimes I)\ket{\eta} 
        =\left[\left(\int_{U\in\mathcal{X}} U^{\otimes t}LU^{\dag\otimes t}d\mu_\mathcal{X}(U)\right)\otimes I\right]\ket{\eta} 
        =(L\otimes I)\ket{\eta} 
        =\ket{\Psi}. 
    \end{align}
    For the proof of $\{\ket{\Psi}\in\mathcal{H}^{\otimes 2t}\ |\ M_{t, \mu_\mathcal{X}}\ket{\Psi}=\ket{\Psi}\}\subset E(\mathrm{Comm}(\Omega_t(\mathcal{X})))$, we take arbitrary $\ket{\Psi}\in\mathcal{H}^{\otimes 2t}$ satisfying $M_{t, \mu_\mathcal{X}}\ket{\Psi}=\ket{\Psi}$. 
    We take $L\in\mathcal{L}(\mathcal{H}^{\otimes t})$ such that $\ket{\Psi}=(L\otimes I)\ket{\eta}$. 
    Then, we have 
    \begin{align}
        M_{t, \mu_\mathcal{X}}(L\otimes I)\ket{\eta} 
        =(L\otimes I)\ket{\eta}. \label{SMeq:SMlem:moment_op_projection2}
    \end{align}
    By Eqs.~\eqref{SMeq:SMlem:moment_op_projection1} and \eqref{SMeq:SMlem:moment_op_projection2}, we get 
    \begin{align}
        \left[\left(\int_{U\in\mathcal{X}} U^{\otimes t}LU^{\dag\otimes t}d\mu_\mathcal{X}(U)\right)\otimes I\right]\ket{\eta}
        =(L\otimes I)\ket{\eta}, 
    \end{align}
    which implies that 
    \begin{align}
        \int_{U\in\mathcal{X}} U^{\otimes t}LU^{\dag\otimes t}d\mu_\mathcal{X}(U)=L. 
    \end{align}
    Then, by the left invariance of $\mu_{\mathcal{X}}$, we have 
    \begin{align}
        V^{\otimes t}LV^{\dag\otimes t} 
        =\int_{U\in\mathcal{X}} (VU)^{\otimes t}L(VU)^{\dag\otimes t}d\mu_\mathcal{X}(U) 
        =\int_{U\in\mathcal{X}} U^{\otimes t}LU^{\dag\otimes t}d\mu_\mathcal{X}(U) 
        =L\ \forall V\in\mathcal{X}, 
    \end{align}
    which means that $L\in\mathrm{Comm}(\Omega_t(\mathcal{X}))$. 
    Thus, we have proven that $\ket{\Psi}\in E(\mathrm{Comm}(\Omega_t(\mathcal{X})))$. 
\end{proof}

For the second step of the proof of Lemma~\ref{SMlem:trivial_to_comm_eq}, we prepare the properties of $\mathscr{S}_t$ defined by Eq.~\eqref{SMeq:symmetrizer_def}.

\begin{lemma} \label{SMlem:perm_symm_proj}
    Let $t\in\mathbb{N}$, $\mathscr{S}_t$ be defined by Eq.~\eqref{SMeq:symmetrizer_def}, $\sigma\in\mathfrak{S}_t$, $V_\sigma$ be defined by Eq.~\eqref{SMeq:inter_system_permutation}, $\Xi$ be a finite set, $L_1, L_2, ..., L_t\in\mathcal{L}(\mathcal{H})$, and $O_\xi\in\mathcal{L}(\mathcal{H})$ for all $\xi\in\Xi$. 
    Then, 
    \begin{align}
        &\mathscr{S}_t(L_1\otimes L_2\otimes\cdots\otimes L_t) 
        =\mathscr{S}_t(L_{\sigma(1)}\otimes L_{\sigma(2)}\otimes\cdots\otimes L_{\sigma(t)}), \label{SMeq:SMlem:perm_symm_proj1}\\
        &\left(\sum_{\xi\in\Xi} O_\xi \right)^{\otimes t}
        =\sum_{\bm{z}\in\mathcal{Z}_t} \frac{t!}{\prod_{\xi\in\Xi} z_\xi!}\mathscr{S}_t\left(\bigotimes_{\xi\in\Xi} O_\xi^{\otimes z_\xi}\right), \label{SMeq:SMlem:perm_symm_proj2}
    \end{align}
    where $\mathcal{Z}_t$ is defined by Eq.~\eqref{SMeq:Z_t_def}, i.e.,
    \begin{align}
        \mathcal{Z}_t:=\left\{\bm{z}'\in(\mathbb{Z}_{\geq 0})^\Xi\ \middle|\ \sum_{(\lambda, \alpha)\in\Xi} z'_{\lambda, \alpha}=t\right\}. 
    \end{align}
\end{lemma}

\begin{proof}
    By the definition of $\mathscr{S}_t$, we directly get Eq.~\eqref{SMeq:SMlem:perm_symm_proj1} as follows: 
    \begin{align}
        \mathscr{S}_t(L_1\otimes L_2\otimes\cdots\otimes L_t) 
        =&\frac{1}{t!}\sum_{\sigma'\in\mathfrak{S}_t} V_{\sigma'}V_\sigma (L_{\sigma(1)}\otimes L_{\sigma(2)}\otimes\cdots\otimes L_{\sigma(t)})V_\sigma^\dag V_{\sigma'}^\dag \nonumber\\
        =&\frac{1}{t!}\sum_{\sigma'\in\mathfrak{S}_t} V_{\sigma'\sigma} (L_{\sigma(1)}\otimes L_{\sigma(2)}\otimes\cdots\otimes L_{\sigma(t)}) V_{\sigma'\sigma}^\dag \nonumber\\
        =&\frac{1}{t!}\sum_{\sigma'\in\mathfrak{S}_t} V_{\sigma'} (L_{\sigma(1)}\otimes L_{\sigma(2)}\otimes\cdots\otimes L_{\sigma(t)}) V_{\sigma'}^\dag \nonumber\\
        =&\mathscr{S}_t(L_{\sigma(1)}\otimes L_{\sigma(2)}\otimes\cdots\otimes L_{\sigma(t)}). 
    \end{align}
    In the following, we show Eq.~\eqref{SMeq:SMlem:perm_symm_proj2}. 
    By the definition of $\mathscr{S}_t$, we have  
    \begin{align}
        \left(\sum_{\xi\in\Xi} O_\xi\right)^{\otimes t}
        =\mathscr{S}_t\left(\left(\sum_{\xi\in\Xi} O_\xi\right)^{\otimes t}\right) 
        =\sum_{\bm{\xi}\in\Xi^t} \mathscr{S}_t\left(\bigotimes_{\zeta\in\Xi} O_{\zeta}^{c_\zeta(\bm{\xi})}\right), \label{SMeq:SMlem:perm_symm_proj3}
    \end{align}
    where $\bm{c}$ is a map from $\Xi^t$ to $\mathcal{Z}_t$, and the $\zeta$ component $c_\zeta$ of $\bm{c}$ is defined by 
    \begin{align}
        c_\zeta(\bm{\xi}):=\#\{u\ |\ \xi_u=\zeta\}. 
    \end{align}
    By noting that the inside of the summation of the r.h.s. of Eq.~\eqref{SMeq:SMlem:perm_symm_proj3} is given by $\bm{c}(\bm{\xi})$, we can change the summation index as follows:  
    \begin{align}
        \sum_{\bm{\xi}\in\Xi^t} \mathscr{S}_t\left(\bigotimes_{\zeta\in\Xi} O_{\zeta}^{c_\zeta(\bm{\xi})}\right)
        =\sum_{\bm{z}\in\mathcal{Z}_t} \#c^{-1}(\bm{z}) \mathscr{S}_t\left(\bigotimes_{\zeta\in\Xi} O_{\zeta}^{z_\zeta}\right). \label{SMeq:SMlem:perm_symm_proj4}
    \end{align}
    By considering the combinatorial interpretation of the multinomial coefficients, we have 
    \begin{align}
        \#c^{-1}(\bm{z})=\frac{t!}{\prod_{\zeta\in\Xi} z_\zeta}. \label{SMeq:SMlem:perm_symm_proj5}
    \end{align}
    By plugging Eq.~\eqref{SMeq:SMlem:perm_symm_proj5} into Eq.~\eqref{SMeq:SMlem:perm_symm_proj4}, we get Eq.~\eqref{SMeq:SMlem:perm_symm_proj2}. 
\end{proof}

We prepare a lemma about the condition for the existence of solutions of Eqs.~\eqref{SMeq:SMthm:general_cond1}, \eqref{SMeq:SMthm:general_cond2}, and \eqref{SMeq:SMthm:general_cond3}. 
We use $(iii)\implies (i)$ in the second step of the proof of Lemma~\ref{SMlem:trivial_to_comm_eq}, and $(i)\implies (ii)$ in the second step of the proof of the Lemma~\ref{SMlem:nontrivial_to_comm_neq}.

\begin{lemma} \label{SMlem:trivial_solution}
    Let $t\in\mathbb{N}$, $\Lambda$ be a finite set, $m_\lambda\in\mathbb{N}$ for all $\lambda\in\Lambda$, $\Xi:=\{(\lambda, \alpha)\ |\ \lambda\in\Lambda, \alpha\in\{1, 2, ..., m_\lambda\}\}$, $\mathcal{V}$ be a linear subspace of $\mathbb{R}^\Lambda$, and $\mathcal{W}:=\bm{\Delta}^{-1}(\mathcal{V})$, where $\bm{\Delta}: \mathbb{R}^\Xi\to\mathbb{R}^\Lambda$ is defined by Eq.~\eqref{SMeq:Delta_vector_def}, i.e., 
    \begin{align}
        \bm{\Delta}(\bm{w})=(\Delta_\lambda(\bm{w}))_{\lambda\in\Lambda}\ \forall\bm{w}\in\mathbb{R}^\Xi 
    \end{align}
    with $\Delta_\lambda: \mathbb{R}^\Xi\to\mathbb{R}$ defined by 
    \begin{align}
        \Delta_\lambda(\bm{w})=\sum_{\alpha=1}^{m_\lambda} w_{\lambda, \alpha}\ \forall\bm{w}\in\mathbb{R}^\Xi. 
    \end{align}
    Then, the following three statements are equivalent: \\
    (i) There exists $\bm{x}=(x_\lambda)_{\lambda\in\Lambda}\in\mathbb{Z}^\Lambda$ such that 
    \begin{align}
        &\bm{x}\neq \bm{0}, \label{SMeq:SMlem:trivial_solution1_1}\\
        &\sum_{\lambda\in\Lambda} m_\lambda |x_\lambda|\leq 2t, \label{SMeq:SMlem:trivial_solution1_2}\\
        &\sum_{\lambda\in\Lambda} m_\lambda x_\lambda=0, \label{SMeq:SMlem:trivial_solution1_3}\\
        &\sum_{\lambda\in\Lambda} v_\lambda x_\lambda=0\ \forall \bm{v}=(v_\lambda)_{\lambda\in\Lambda}\in\mathcal{V}. \label{SMeq:SMlem:trivial_solution1_4}
    \end{align}
    (ii) There exist $\bm{y}=(y_\lambda)_{\lambda\in\Lambda}, \bm{y}'=(y'_\lambda)_{\lambda\in\Lambda}\in(\mathbb{Z}_{\geq 0})^\Lambda$ such that 
    \begin{align}
        &\bm{y}\neq\bm{y}', \label{SMeq:SMlem:trivial_solution2_1}\\
        &\sum_{\lambda\in\Lambda} m_\lambda y_\lambda
        =\sum_{\lambda\in\Lambda} m_\lambda y'_\lambda\leq t, \label{SMeq:SMlem:trivial_solution2_2}\\
        &\sum_{\lambda\in\Lambda} v_\lambda y_\lambda
        =\sum_{\lambda\in\Lambda} v_\lambda y'_\lambda\ \forall \bm{v}=(v_\lambda)_{\lambda\in\Lambda}\in\mathcal{V}. \label{SMeq:SMlem:trivial_solution2_3}
    \end{align}
    (iii) There exist $\bm{z}=(z_{\lambda, \alpha})_{(\lambda, \alpha)\in\Xi}, \bm{z}'=(z'_{\lambda, \alpha})_{(\lambda, \alpha)\in\Xi}\in(\mathbb{Z}_{\geq 0})^\Xi$ such that 
    \begin{align}
        &\bm{z}\neq\bm{z}', \label{SMeq:SMlem:trivial_solution3_1}\\
        &\sum_{(\lambda, \alpha)\in\Xi} z_{\lambda, \alpha}
        =\sum_{(\lambda, \alpha)\in\Xi} z'_{\lambda, \alpha}\leq t, \label{SMeq:SMlem:trivial_solution3_2}\\
        &\sum_{(\lambda, \alpha)\in\Xi} w_{\lambda, \alpha} z_{\lambda, \alpha}
        =\sum_{(\lambda, \alpha)\in\Xi} w_{\lambda, \alpha} z'_{\lambda, \alpha}\ \forall \bm{w}=(w_{\lambda, \alpha})_{(\lambda, \alpha)\in\Xi}\in\mathcal{W}. \label{SMeq:SMlem:trivial_solution3_3}
    \end{align}
\end{lemma}

\begin{proof}
    First, we show $(i)\implies (ii)$. 
    We suppose that we can take $\bm{x}\in\mathbb{Z}^\Lambda$ satisfying Eqs.~\eqref{SMeq:SMlem:trivial_solution1_1}, \eqref{SMeq:SMlem:trivial_solution1_2}, \eqref{SMeq:SMlem:trivial_solution1_3}, and \eqref{SMeq:SMlem:trivial_solution1_4}. 
    We define $\bm{y}, \bm{y}'\in (\mathbb{Z}_{\geq 0})^\Lambda$ by 
    \begin{align}
        &y_\lambda:=\frac{|x_\lambda|+x_\lambda}{2}, \\
        &y'_\lambda:=\frac{|x_\lambda|-x_\lambda}{2}, 
    \end{align}
    which implies that $y_\lambda-y'_\lambda=x_\lambda$. 
    Thus Eq.~\eqref{SMeq:SMlem:trivial_solution2_1}, the equality in Eq.~\eqref{SMeq:SMlem:trivial_solution2_2}, and Eq.~\eqref{SMeq:SMlem:trivial_solution2_3} directly follow from Eqs.~\eqref{SMeq:SMlem:trivial_solution1_1}, \eqref{SMeq:SMlem:trivial_solution1_3}, and \eqref{SMeq:SMlem:trivial_solution1_4}, respectively. 
    The inequality in Eq.~\eqref{SMeq:SMlem:trivial_solution2_2} can be shown as follows: 
    \begin{align}
        \sum_{\lambda\in\Lambda} m_\lambda y_\lambda
        =\frac{1}{2}\left(\sum_{\lambda\in\Lambda} m_\lambda y_\lambda+\sum_{\lambda\in\Lambda} m_\lambda y'_\lambda\right)
        =\frac{1}{2}\sum_{\lambda\in\Lambda} m_\lambda (y_\lambda+y'_\lambda)
        =\frac{1}{2}\sum_{\lambda\in\Lambda} m_\lambda |x_\lambda|
        \leq t.
    \end{align}

    Next, we show $(ii)\implies (iii)$. 
    We suppose that we can take $\bm{y}, \bm{y}'\in (\mathbb{Z}_{\geq 0})^\Lambda$ satisfying Eqs.~\eqref{SMeq:SMlem:trivial_solution2_1}, \eqref{SMeq:SMlem:trivial_solution2_2} and \eqref{SMeq:SMlem:trivial_solution2_3}. 
    We define $\bm{z}, \bm{z}'\in(\mathbb{Z}_{\geq 0})^\Xi$ by $z_{\lambda, \alpha}:=y_\lambda$ and $z'_{\lambda, \alpha}:=y'_\lambda$ for all $(\lambda, \alpha)\in\Xi$. 
    Then, Eq.~\eqref{SMeq:SMlem:trivial_solution3_1} directly follows from Eq.~\eqref{SMeq:SMlem:trivial_solution2_1}. 
    We note that 
    \begin{align}
        &\sum_{(\lambda, \alpha)\in\Xi} z_{\lambda, \alpha}
        =\sum_{\lambda\in\Lambda}\sum_{\alpha=1}^{m_\lambda} y_\lambda
        =\sum_{\lambda\in\Lambda} m_\lambda y_\lambda, \\
        &\sum_{(\lambda, \alpha)\in\Xi} w_{\lambda, \alpha} z_{\lambda, \alpha}
        =\sum_{\lambda\in\Lambda} \left(\sum_{\alpha=1}^{m_\lambda} w_{\lambda, \alpha}\right) y_\lambda 
        =\sum_{\lambda\in\Lambda} \Delta_\lambda(\bm{w}) y_\lambda, 
    \end{align}
    and in the same way, we can show that $\sum_{(\lambda, \alpha)\in\Xi} z'_{\lambda, \alpha}=\sum_{\lambda\in\Lambda} m_\lambda z'_\lambda$ and $\sum_{(\lambda, \alpha)\in\Xi} w_{\lambda, \alpha} z'_{\lambda, \alpha}=\sum_{\lambda\in\Lambda} \Delta_\lambda(\bm{w}) y'_\lambda$. 
    Since $\bm{\Delta}(\bm{w})\in\mathcal{V}$, Eqs.~\eqref{SMeq:SMlem:trivial_solution3_2} and \eqref{SMeq:SMlem:trivial_solution3_3} directly follow from Eqs.~\eqref{SMeq:SMlem:trivial_solution2_2} and \eqref{SMeq:SMlem:trivial_solution2_3}.

    Finally, we show $(iii)\implies (i)$. 
    We suppose that we can take $\bm{z}, \bm{z}'\in (\mathbb{Z}_{\geq 0})^\Xi$ satisfying Eq.~\eqref{SMeq:SMlem:trivial_solution3_1}, \eqref{SMeq:SMlem:trivial_solution3_2}, and \eqref{SMeq:SMlem:trivial_solution3_3}. 
    We take arbitrary $\mu\in\Lambda$ and $\beta, \beta'\in\{1, 2, ..., m_\lambda\}$. 
    By noting that $\bm{w}\in\mathcal{W}$ when $w_{\lambda, \alpha}:=\delta_{\lambda, \mu}(\delta_{\alpha, \beta}-\delta_{\alpha, \beta'})$, Eq.~\eqref{SMeq:SMlem:trivial_solution3_3} implies that $z_{\mu, \beta}-z_{\mu, \beta'}=z'_{\mu, \beta}-z'_{\mu, \beta'}$, which yields $z_{\mu, \beta}-z'_{\mu, \beta}=z_{\mu, \beta'}-z'_{\mu, \beta'}$. 
    Since this holds for all $\beta, \beta'\in\{1, 2, ..., m_\mu\}$, $z_{\lambda, \alpha}-z'_{\lambda, \alpha}$ is independent of $\alpha$. 
    Thus we can define $\bm{x}\in\mathbb{Z}^\Lambda$ such that 
    \begin{align}
        x_\lambda=z_{\lambda, \alpha}-z'_{\lambda, \alpha}\ \forall \alpha\in\{1, 2, ..., m_\lambda\}, \lambda\in\Lambda. \label{SMeq:SMlem:trivial_solution4}
    \end{align}
    By this relation, Eq.~\eqref{SMeq:SMlem:trivial_solution1_1} follows from Eq.~\eqref{SMeq:SMlem:trivial_solution3_1}. 
    By using Eqs.~\eqref{SMeq:SMlem:trivial_solution4} and \eqref{SMeq:SMlem:trivial_solution3_2}, we get Eq.~\eqref{SMeq:SMlem:trivial_solution1_2} as follows: 
    \begin{align}
        \sum_{\lambda\in\Lambda} m_\lambda |x_\lambda|
        =\sum_{(\lambda, \alpha)\in\Xi} |x_\lambda|
        =\sum_{(\lambda, \alpha)\in\Xi} |z_{\lambda, \alpha}-z'_{\lambda, \alpha}| 
        \leq\sum_{(\lambda, \alpha)\in\Xi} (z_{\lambda, \alpha}+z'_{\lambda, \alpha})
        =\sum_{(\lambda, \alpha)\in\Xi} z_{\lambda, \alpha}+\sum_{(\lambda, \alpha)\in\Xi} z'_{\lambda, \alpha} 
        \leq 2t, 
    \end{align}
    where we used the triangle inequality in the first inequality. 
    Similarly, by using Eqs.~\eqref{SMeq:SMlem:trivial_solution4} and \eqref{SMeq:SMlem:trivial_solution3_2}, we get Eq.~\eqref{SMeq:SMlem:trivial_solution1_3} as follows: 
    \begin{align}
        \sum_{\lambda\in\Lambda} m_\lambda x_\lambda
        =\sum_{(\lambda, \alpha)\in\Xi} x_\lambda
        =\sum_{(\lambda, \alpha)\in\Xi} (z_{\lambda, \alpha}-z'_{\lambda, \alpha}) 
        =\sum_{(\lambda, \alpha)\in\Xi} z_{\lambda, \alpha}-\sum_{(\lambda, \alpha)\in\Xi} z'_{\lambda, \alpha}
        =0. 
    \end{align}
    For any $\lambda\in\Lambda$, we arbitrarily take $\beta_\lambda\in\{1, 2, ..., m_\lambda\}$. 
    For any $\bm{v}\in\mathcal{V}$, by noting that $\bm{w}\in\mathcal{W}$ when $w_{\lambda, \alpha}:=v_\lambda\delta_{\alpha, \beta_\lambda}$, Eq.~\eqref{SMeq:SMlem:trivial_solution3_3} implies that $\sum_{\lambda\in\Lambda} z_{\lambda, \beta_\lambda}v_\lambda=\sum_{\lambda\in\Lambda} z'_{\lambda, \beta_\lambda}v_\lambda$, which implies Eq.~\eqref{SMeq:SMlem:trivial_solution1_4}. 
\end{proof}

For the first step of the proof of Lemma~\ref{SMlem:nontrivial_to_comm_neq}, we show the property of the totally antisymmetric state.

\begin{lemma} \label{SMlem:antisymmetric_state}
	Let $m\in\mathbb{N}$, $A$ be a linear operator on $\mathbb{C}^m$ and $\ket{\chi(\mathbb{C}^m)}$ and $\omega_m$ be defined by Eqs.~\eqref{SMeq:antisymmetric_state_def} and \eqref{SMeq:omega_def}, respectively, i.e., 
    \begin{align}
        &\ket{\chi(\mathbb{C}^m)}:=\frac{1}{\sqrt{m!}}\sum_{\sigma\in\mathfrak{S}_m} \mathrm{sgn}(\sigma)\bigotimes_{\alpha=1}^m \ket{\sigma(\alpha)}, \\
		&\omega_t(A):=\sum_{s=1}^t I^{\otimes s-1}\otimes A\otimes I^{\otimes t-s}. 
	\end{align} 
	Then, 
    \begin{align}
        \omega_m(A)\ket{\chi(\mathbb{C}^m)}
        =\mathrm{tr}(A)\ket{\chi(\mathbb{C}^m)}. \label{SMeq:SMlem:antisymmetric_state1}
    \end{align}
\end{lemma}

\begin{proof}
    By using $V_\sigma$ defined by Eq.~\eqref{SMeq:inter_system_permutation}, we have 
    \begin{align}
        \ket{\chi(\mathbb{C}^m)}
        =\frac{1}{\sqrt{m!}}\sum_{\sigma\in\mathfrak{S}_m} \mathrm{sgn}(\sigma)V_{\sigma^{-1}}\left(\bigotimes_{\alpha=1}^m \ket{\alpha}\right). \label{SMeq:SMlem:antisymmetric_state2}
    \end{align}
    This implies that 
    \begin{align}
        \omega_m(A)\ket{\chi(\mathbb{C}^m)}
        =\frac{1}{\sqrt{m!}}\sum_{\sigma\in\mathfrak{S}_m} \mathrm{sgn}(\sigma)\omega_m(A)V_{\sigma^{-1}}\left(\bigotimes_{\alpha=1}^m \ket{\alpha}\right) 
        =\frac{1}{\sqrt{m!}}\sum_{\sigma\in\mathfrak{S}_m} \mathrm{sgn}(\sigma)V_{\sigma^{-1}}\omega_m(A)\left(\bigotimes_{\alpha=1}^m \ket{\alpha}\right). \label{SMeq:SMlem:antisymmetric_state3}
    \end{align}
    We note that 
    \begin{align}
        \omega_m(A)\left(\bigotimes_{\alpha=1}^m \ket{\alpha}\right)
        =&\sum_{\beta=1}^m \left(I^{\otimes \beta-1}\otimes A\otimes I^{\otimes m-\beta}\right) \left[\left(\bigotimes_{\alpha=1}^{\beta-1} \ket{\alpha}\right)\otimes\ket{\beta}\otimes\left(\bigotimes_{\alpha=\beta+1}^m \ket{\alpha}\right)\right] \nonumber\\
        =&\sum_{\beta=1}^m \sum_{\beta'=1}^m \left(\bigotimes_{\alpha=1}^{\beta-1} \ket{\alpha}\right)\otimes a_{\beta', \beta}\ket{\beta'}\otimes\left(\bigotimes_{\alpha=\beta+1}^m \ket{\alpha}\right) \nonumber\\
        =&\sum_{\beta=1}^m \sum_{\beta'=1}^m \frac{I+V_{\tau_{\beta, \beta'}^{-1}}}{2}\left[\left(\bigotimes_{\alpha=1}^{\beta-1} \ket{\alpha}\right)\otimes a_{\beta', \beta}\ket{\beta'}\otimes\left(\bigotimes_{\alpha=\beta+1}^m \ket{\alpha}\right)\right], \label{SMeq:SMlem:antisymmetric_state4}
    \end{align}
    where $a_{\beta', \beta}:=\braket{\beta'|A|\beta}$, and $\tau_{\beta, \beta'}$ is the transposition between $\beta$ and $\beta'$. 
    By plugging Eq.~\eqref{SMeq:SMlem:antisymmetric_state4} into Eq.~\eqref{SMeq:SMlem:antisymmetric_state3}, we get 
    \begin{align}
        \omega_m(A)\ket{\chi(\mathbb{C}^m})=\frac{1}{\sqrt{m!}}\sum_{\beta=1}^m \sum_{\beta'=1}^m \left(\sum_{\sigma\in\mathfrak{S}_m} \mathrm{sgn}(\sigma)V_{\sigma^{-1}}\frac{I+V_{\tau_{\beta, \beta'}^{-1}}}{2}\right)
        \left[\left(\bigotimes_{\alpha=1}^{\beta-1} \ket{\alpha}\right)\otimes a_{\beta', \beta}\ket{\beta'}\otimes\left(\bigotimes_{\alpha=\beta+1}^m \ket{\alpha}\right)\right]. \label{SMeq:SMlem:antisymmetric_state5}
    \end{align}
    We note that 
    \begin{align}
        \sum_{\sigma\in\mathfrak{S}_m} \mathrm{sgn}(\sigma)V_{\sigma^{-1}}\frac{I+V_{\tau_{\beta, \beta'}^{-1}}}{2}
        =&\sum_{\sigma\in\mathfrak{S}_m} \frac{1}{2}\left(\mathrm{sgn}(\sigma)V_{\sigma^{-1}}+\mathrm{sgn}(\tau_{\beta, \beta'})\mathrm{sgn}(\tau_{\beta, \beta'}\sigma)V_{(\tau_{\beta, \beta'}\sigma)^{-1}}\right) \nonumber\\
        =&\sum_{\sigma\in\mathfrak{S}_m} \frac{1+\mathrm{sgn}(\tau_{\beta, \beta'})}{2}\mathrm{sgn}(\sigma)V_{\sigma^{-1}} \nonumber\\
        =&\delta_{\beta, \beta'}\sum_{\sigma\in\mathfrak{S}_m} \mathrm{sgn}(\sigma)V_{\sigma^{-1}}. \label{SMeq:SMlem:antisymmetric_state6}
    \end{align}
    By plugging Eq.~\eqref{SMeq:SMlem:antisymmetric_state6} into Eq.~\eqref{SMeq:SMlem:antisymmetric_state5}, we get 
    \begin{align}
        \omega_m(A)\ket{\chi(\mathbb{C}^m)}
        =&\frac{1}{\sqrt{m!}}\sum_{\beta=1}^m \sum_{\beta'=1}^m \delta_{\beta, \beta'}\sum_{\sigma\in\mathfrak{S}_m} \mathrm{sgn}(\sigma)V_{\sigma^{-1}}\left[\left(\bigotimes_{\alpha=1}^{\beta-1} \ket{\alpha}\right)\otimes a_{\beta', \beta}\ket{\beta'}\otimes\left(\bigotimes_{\alpha=\beta+1}^m \ket{\alpha}\right)\right] \nonumber\\
        =&\frac{1}{\sqrt{m!}}\sum_{\beta=1}^m \sum_{\sigma\in\mathfrak{S}_m} \mathrm{sgn}(\sigma)V_{\sigma^{-1}}\left[\left(\bigotimes_{\alpha=1}^{\beta-1} \ket{\alpha}\right)\otimes a_{\beta, \beta}\ket{\beta}\otimes\left(\bigotimes_{\alpha=\beta+1}^m \ket{\alpha}\right)\right] \nonumber\\
        =&\left(\sum_{\beta=1}^m a_{\beta, \beta}\right)
        \left[\frac{1}{\sqrt{m!}}\sum_{\sigma\in\mathfrak{S}_m} \mathrm{sgn}(\sigma)V_{\sigma^{-1}}\left(\bigotimes_{\alpha=1}^m \ket{\alpha}\right)\right]. \label{SMeq:SMlem:antisymmetric_state7}
    \end{align}
    By Eqs.~\eqref{SMeq:SMlem:antisymmetric_state2} and \eqref{SMeq:SMlem:antisymmetric_state7}, we get Eq.~\eqref{SMeq:SMlem:antisymmetric_state1}. 
\end{proof}

For the third step of the proof of Lemma~\ref{SMlem:nontrivial_to_comm_neq}, we show that the commutant of a Lie group is the same as that of its associated Lie algebra.

\begin{lemma} \label{SMlem:Lie_group_algebra_comm}
    Let $n, t\in\mathbb{N}$, $\mathcal{S}$ be a connected compact Lie subgroup of $\mathcal{U}_n$ and $\mathfrak{s}$ be its associated Lie algebra. 
    Then, 
    \begin{align}
        \mathrm{Comm}(\Omega_t(\mathcal{S})) 
        =\mathrm{Comm}(\omega_t(\mathfrak{s})). 
    \end{align}
\end{lemma}

\begin{proof}
    First, we show that $\mathrm{Comm}(\Omega_t(\mathcal{S}))\subset\mathrm{Comm}(\omega_t(\mathfrak{s}))$. 
    We take arbitrary $L\in\mathrm{Comm}(\Omega_t(\mathcal{S}))$. 
    For any $A\in\mathfrak{s}$ and $\theta\in\mathbb{R}$, we have $e^{i\theta A}\in\mathcal{S}$, which implies that $e^{i\theta \omega_t(A)}=\Omega_t(e^{i\theta A})\in\Omega_t(\mathcal{S})$. 
    Thus we have $[L, e^{i\theta A}]=0$. 
    By taking the derivative at $\theta=0$, we get $[L, A]=0$. 
    Since this holds for all $A\in\mathfrak{s}$, we have $L\in\mathrm{Comm}(\mathfrak{s})$. 

    Next, we show that $\mathrm{Comm}(\Omega_t(\mathcal{S}))\supset\mathrm{Comm}(\omega_t(\mathfrak{s}))$. 
    We take arbitrary $L\in\mathrm{Comm}(\omega_t(\mathfrak{s}))$. 
    Since $\mathcal{S}$ is connected and compact, every $U\in\mathcal{S}$ can be written as $U=e^{iA}$ with some $A\in\mathfrak{s}$. 
    Thus we have $[L, \Omega_t(U)]=[L, e^{i\omega_t(A)}]=0$. 
    Since this holds for all $U\in\mathcal{S}$, we get $L\in\mathrm{Comm}(\Omega_t(\mathcal{S}))$. 
\end{proof}

We show that in Theorem~\ref{SMthm:general},  Eq.~\eqref{SMeq:SMthm:general_cond3} is equivalent to Eq.~\eqref{SMeq:SMthm:general_cond4} in a special case.

\begin{lemma} \label{SMlem:simplified_condition}
    Let $n\in\mathbb{N}$, $T$ be a unitary representation of a group $G$ on a single qudit, $R=T^{\otimes n}$, $\Lambda$ be the set of the labels of the inequivalent irreducible representations appearing in $R$, $\bm{f}$ be defined by Eq.~\eqref{SMeq:trace_vec_def}, $\Gamma$ be a set of subsets of $\{1, 2, ..., n\}$, $k:=\max_{\gamma\in\Gamma} \#\gamma$, $\mathcal{V}:=\mathrm{span}_\mathbb{R}(\{\bm{f}(A)\ |\ \exists \gamma\in\Gamma \textrm{ s.t. } A\in\mathfrak{u}_{n, G, R}^\gamma\})$, $\mathcal{C}:=\{\bm{f}(A\otimes\mathrm{I}^{\otimes n-k})\ |\ A\in\mathcal{L}_{k, G, T^{\otimes k}}\}$, and $\bm{x}\in\mathbb{Z}^\Lambda$. 
    Then, $\sum_{\lambda\in\Lambda} v_\lambda x_\lambda=0$ for all $\bm{v}\in\mathcal{V}$ is equivalent to $\sum_{\lambda\in\Lambda} c_\lambda x_\lambda=0$ for all $\bm{c}\in\mathcal{C}$. 
\end{lemma}

\begin{proof}
    For any $\gamma\in\Gamma$, we can take some qudit permutation operator $P$ such that the map $\mathcal{E}_{\gamma, P}: \mathfrak{u}_{\#\gamma, G, T^{\otimes\#\gamma}}\to\mathfrak{u}_{n, G, T^{\otimes n}}^\gamma$ defined by $\mathcal{E}_{\gamma, P}(A):=P(A\otimes\mathrm{I}^{\otimes n-\#\gamma})P^\dag$ is a bijection. 
    Since $R$ is given by $T^{\otimes n}$, we have $P\in\mathcal{U}_{n, G, R}$, which implies that $P$ can be written as $P=\sum_{\lambda\in\Lambda} F_\lambda(I\otimes P_\lambda)F_\lambda^\dag$ with some $P_\lambda\in\mathcal{U}(\mathcal{M}_\lambda)$. 
    $B:=A\otimes\mathrm{I}^{\otimes n-\#\gamma}$ can also be written as $B=\sum_{\lambda\in\Lambda} F_\lambda ( I\otimes B_\lambda)F_\lambda^\dag$ with some $B\in\mathcal{L}(\mathbb{C}^{m_\lambda})$. 
    Then, we have $\bm{f}(\mathcal{E}_{\gamma, P}(A))=(\mathrm{tr}(P_\lambda B_\lambda P_\lambda^\dag))_{\lambda\in\Lambda}=(\mathrm{tr}(B_\lambda))_{\lambda\in\Lambda}=\bm{f}(A\otimes\mathrm{I}^{\otimes n-\#\gamma})$. 
    By taking the range of this equation over $A\in\mathfrak{u}_{\#\gamma, G, T^{\otimes \#\gamma}}$, we get $\{\bm{f}(A)\ |\ A\in\mathfrak{u}_{n, G, T^{\otimes n}}^\gamma\}=\{\bm{f}(A\otimes\mathrm{I}^{\otimes n-\#\gamma})\ |\ A\in\mathfrak{u}_{\#\gamma, G, T^{\otimes \#\gamma}}\}$, which implies that $\mathcal{V}=\{\bm{f}(A\otimes\mathrm{I}^{\otimes n-k})\ |\ A\in\mathfrak{u}_{k, G, T^{\otimes k}}\}$. 
    Since we have $\mathcal{L}_{k, G, T^{\otimes k}}=\mathfrak{u}_{k, G, T^{\otimes k}}+i\mathfrak{u}_{k, G, T^{\otimes k}}$, it holds that $\mathcal{C}=\mathcal{V}+i\mathcal{V}$. 
    Thus, $\sum_{\lambda\in\Lambda} v_\lambda x_\lambda=0$ for all $\bm{v}\in\mathcal{V}$ is equivalent to $\sum_{\lambda\in\Lambda} c_\lambda x_\lambda=0$ for all $\bm{c}\in\mathcal{C}$. 
\end{proof}

We give a rough sufficient condition for the assumptions in Lemma~\ref{SMlem:U1_nontrivial_existence_general}.

\begin{lemma} \label{SMlem:U1_sufficient_cond}
    Let $n, k\in\mathbb{N}$ satisfy $k\geq 2$ and $n\geq 2^k$, and $b_{n, j}$ be defined by Eq.~\eqref{SMeq:U1_bound_def}, i.e., 
    \begin{align}
        b_{n, j}:=\frac{2^{\lfloor j/2\rfloor}}{\left\lceil \frac{j}{2}\right\rceil!}\prod_{\alpha=1}^{\lceil j/2\rceil} (n-j+2\alpha-1). 
    \end{align}
    Then, $b_{n, k}\leq\binom{n}{\lceil k/2\rceil+1}$. 
    Moreover, when $k$ is odd, $b_{n, k}\leq b_{n, k+1}$. 
\end{lemma}

\begin{proof}
    For the proof of $b_{n, k}\leq\binom{n}{\lceil k/2\rceil+1}$, it is sufficient to show that 
    \begin{align}
        2^{\lfloor k/2\rfloor}\left(\left\lceil \frac{k}{2}\right\rceil+1\right)\prod_{\alpha=1}^{\lceil k/2\rceil} (n-k+2\alpha-1)
        \leq\prod_{\alpha=1}^{\lceil k/2\rceil+1} (n-\alpha+1). \label{SMeq:SMlem:U1_sufficient_cond1}
    \end{align}
    First, we consider the case when $k$ is even. 
    The product part in the definition of $b_{n, k}$ is upper bounded as 
    \begin{align}
        \prod_{\alpha=1}^{\lceil k/2\rceil} (n-k+2\alpha-1) 
        =&\prod_{\alpha=1}^{k/2} (n-k+2\alpha-1) \nonumber\\
        =&\prod_{\alpha=2}^{(k/2)+1} \left[n-k+2\left(\frac{k}{2}+2-\alpha\right)-1\right] \nonumber\\
        =&\prod_{\alpha=2}^{\lceil k/2\rceil+1} [n-\alpha+1-(\alpha-2)] \nonumber\\
        \leq&\prod_{\alpha=2}^{\lceil k/2\rceil+1} (n-\alpha+1). \label{SMeq:SMlem:U1_sufficient_cond2}
    \end{align}
    By the assumption that $n\geq 2^k$, we get 
    \begin{align}
        \left(\left\lceil \frac{k}{2}\right\rceil+1\right)2^{\lfloor k/2\rfloor} 
        \leq 2^{\lceil k/2\rceil}2^{\lfloor k/2\rfloor} 
        =2^k 
        \leq n. \label{SMeq:SMlem:U1_sufficient_cond3}
    \end{align}
    By multiplying Eqs.~\eqref{SMeq:SMlem:U1_sufficient_cond2} and \eqref{SMeq:SMlem:U1_sufficient_cond3}, we get Eq.~\eqref{SMeq:SMlem:U1_sufficient_cond1}.

    Next, we consider the case when $k$ is odd. 
    We note that Eq.~\eqref{SMeq:SMlem:U1_sufficient_cond1} is equivalent to 
    \begin{align}
        2^{\lfloor k/2\rfloor}\left(\left\lceil \frac{k}{2}\right\rceil+1\right)\prod_{\alpha=1}^{\lceil k/2\rceil-1} (n-k+2\alpha-1)
        \leq\prod_{\alpha=2}^{\lceil k/2\rceil+1} (n-\alpha+1), \label{SMeq:SMlem:U1_sufficient_cond4}
    \end{align}
    because the term for $\alpha=\lceil{k/2} \rceil$ in the l.h.s. and the term for $\alpha=1$ in the r.h.s. in Eq.~\eqref{SMeq:SMlem:U1_sufficient_cond1} are both equal to $n$. 
    By noting that $k-1$ is even, we can substitute $n\mapsto n-1$ and $k\mapsto k-1$ in Eq.~\eqref{SMeq:SMlem:U1_sufficient_cond2}. 
    Then, the product part in the l.h.s. of Eq.~\eqref{SMeq:SMlem:U1_sufficient_cond4} is upper bounded as 
    \begin{align}
        \prod_{\alpha=1}^{\lceil k/2\rceil-1} (n-k+2\alpha-1) 
        =&\prod_{\alpha=1}^{\lceil (k-1)/2\rceil} [(n-1)-(k-1)+2\alpha-1] \nonumber\\
        \leq&\prod_{\alpha=2}^{\lceil (k-1)/2\rceil+1} [(n-1)-\alpha+1] \nonumber\\
        =&\prod_{\alpha=2}^{\lceil k/2\rceil} [n-(\alpha+1)+1] \nonumber\\
        =&\prod_{\alpha=3}^{\lceil k/2\rceil+1} (n-\alpha+1). \label{SMeq:SMlem:U1_sufficient_cond5}
    \end{align}
    By the assumption that $n\geq 2^k$ and $\lceil k/2\rceil\geq 2$, we get 
    \begin{align}
        \left(\left\lceil \frac{k}{2}\right\rceil+1\right)2^{\lfloor k/2\rfloor} 
        =\left(\left\lceil \frac{k}{2}\right\rceil+2\right)2^{\lfloor k/2\rfloor}-2^{\lfloor k/2\rfloor} 
        \leq 2\left\lceil \frac{k}{2}\right\rceil 2^{\lfloor k/2\rfloor}-1 
        \leq 2^{\lceil k/2\rceil}2^{\lfloor k/2\rfloor}-1 
        =2^k-1 
        \leq n-1. \label{SMeq:SMlem:U1_sufficient_cond6}
    \end{align}
    By multiplying Eqs.~\eqref{SMeq:SMlem:U1_sufficient_cond5} and \eqref{SMeq:SMlem:U1_sufficient_cond6}, we get Eq.~\eqref{SMeq:SMlem:U1_sufficient_cond4}. 
    For the proof of $b_{n, k}\leq b_{n, k+1}$, it is sufficient to show that 
    \begin{align}
        \prod_{\alpha=1}^{(k+1)/2} (n-2\alpha+2)\leq 2\prod_{\alpha=1}^{(k+1)/2} (n-2\alpha+1). \label{SMeq:SMlem:U1_sufficient_cond7}
    \end{align}
    We note that 
    \begin{align}
        \prod_{\alpha=2}^{(k+1)/2} (n-2\alpha+2)
        =\prod_{\alpha=1}^{(k+1)/2-1} (n-2(\alpha+1)+2) 
        \leq\prod_{\alpha=1}^{(k+1)/2-1} (n-2\alpha+1). \label{SMeq:SMlem:U1_sufficient_cond8}
    \end{align}
    By the assumption of $n\geq 2^k$, we have 
    \begin{align}
        n=2(n-k)-(n-2^k)-(2^k-2k)\leq 2(n-k). \label{SMeq:SMlem:U1_sufficient_cond9}
    \end{align}
    By multiplying Eqs.~\eqref{SMeq:SMlem:U1_sufficient_cond8} and \eqref{SMeq:SMlem:U1_sufficient_cond9}, we get Eq.~\eqref{SMeq:SMlem:U1_sufficient_cond7}. 
\end{proof}

We prepare a lemma making it easier to check whether the assumption in Lemma~\ref{SMlem:SU2_nontrivial_nonexistence_general} holds or not.

\begin{lemma} \label{SMlem:binomial_SU2}
	Let $n, a, b\in\mathbb{Z}$ and $0\leq a\leq b\leq n/2-1$. 
	Then, 
    \begin{align}
        \binom{n}{j+1}-\binom{n}{j}\geq \min\left\{\binom{n}{a+1}-\binom{n}{a}, \binom{n}{b+1}-\binom{n}{b}\right\} 
    \end{align}
    for all $j\in\{a, a+1, ..., b\}$. 
\end{lemma}

\begin{proof}
    We note that 
    \begin{align}
        \left(\binom{n}{j+1}-\binom{n}{j}\right)-\left(\binom{n}{j}-\binom{n}{j-1}\right) 
        =\binom{n}{j-1}\frac{(n-2j)^2-(n+2)}{j(j+1)}. \label{SMeq:SMlem:binomial_SU2}
    \end{align}
    for all $j\in\mathbb{Z}$ satisfying $1\leq j\leq n/2-1$. 
    The r.h.s. is positive at $j=1$ and negative at $j=n/2-1$, and monotonically decreases as $j$ increases. 
    We can thus take $k\in\{a, a+1, ..., b\}$ such that the value of Eq.~\eqref{SMeq:SMlem:binomial_SU2} is positive when $j<k$ and non-positive when $j\geq k$. 
    This means that $\binom{n}{j+1}-\binom{n}{j}$ is increasing while $j\leq k$ and nonincreasing while $k\leq j$. 
    Therefore $\binom{n}{j+1}-\binom{n}{j}$ takes the minimum value at $j=a$ or $j=b$. 
\end{proof}

By using the lemma above, we give a rough sufficient condition for the assumption in Lemma~\ref{SMlem:SU2_nontrivial_nonexistence_general}. 
The following lemma is the counterpart of Lemma~\ref{SMlem:U1_sufficient_cond} in the $\mathrm{SU}(2)$ case.

\begin{lemma} \label{SMlem:SU2_sufficient_cond}
    Let $n, k\in\mathbb{N}$ satisfy $k\geq 2$ and $n\geq 2^{2(\lfloor k/2\rfloor+1)}$, and $c_{n, k}$ be defined by Eq.~\eqref{SMeq:SU2_bound_def}, i.e., 
    \begin{align}
        c_{n, k}:=\frac{2^{\lfloor k/2\rfloor}}{(\lfloor \frac{k}{2}\rfloor+1)!}\prod_{\alpha=1}^{\lfloor k/2\rfloor+1} (n-2\alpha+1). 
    \end{align} 
    Then, $c_{n, k}\leq\binom{n}{j+1}-\binom{n}{j}$ for all $j\in\{\lfloor k/2\rfloor+1, \lfloor k/2\rfloor+2, ..., \lfloor n/2\rfloor-1\}$. 
\end{lemma}

\begin{proof}
    By Lemma~\ref{SMlem:binomial_SU2}, it is sufficient to show that $c_{n, k}\leq\binom{n}{j+1}-\binom{n}{j}$ only for $j=\lfloor k/2\rfloor+1$ and $\lfloor n/2\rfloor-1$. 
    First, we prove that $c_{n, k}\leq\binom{n}{j+1}-\binom{n}{j}$ in the case of $j=\lfloor k/2\rfloor+1$. 
    We note that 
    \begin{align}
        \binom{n}{\lfloor \frac{k}{2}\rfloor+2}-\binom{n}{\lfloor \frac{k}{2}\rfloor+1}
        =&\frac{n-2\lfloor \frac{k}{2}\rfloor-3}{(\lfloor \frac{k}{2}\rfloor+2)!}\left(\prod_{\alpha=1}^{\lfloor k/2\rfloor+1} (n-\alpha+1)\right) \nonumber\\
        \geq&\frac{n-2\lfloor \frac{k}{2}\rfloor-3}{(\lfloor \frac{k}{2}\rfloor+2)!}\prod_{\alpha=1}^{\lfloor k/2\rfloor+1} (n-2\alpha+1) \nonumber\\
        =&\frac{n-2\lfloor \frac{k}{2}\rfloor-3}{(\lfloor \frac{k}{2}\rfloor+2)2^{\lfloor k/2\rfloor}}c_{n, k}. \label{SMeq:SMlem:SU2_sufficient_cond1}
    \end{align}
    Since $n$ satisfies $n\geq 2^{2(\lfloor k/2\rfloor+1)}$, we have 
    \begin{align}
        \left(\left\lfloor \frac{k}{2}\right\rfloor+2\right)2^{\lfloor k/2\rfloor} 
        =&\left[\left(\left\lfloor \frac{k}{2}\right\rfloor+1\right)+1\right]\left(2^{\lfloor k/2\rfloor}+2\right)-2\left\lfloor \frac{k}{2}\right\rfloor-4 \nonumber\\
        \leq& 2^{\lfloor k/2\rfloor+1}\cdot \left(2^{\lfloor k/2\rfloor}+2^{\lfloor k/2\rfloor}\right)-2\left\lfloor \frac{k}{2}\right\rfloor-3 \nonumber\\
        =& 2^{2(\lfloor k/2\rfloor+1)}-2\left\lfloor \frac{k}{2}\right\rfloor-3 \nonumber\\
        \leq& n-2\left\lfloor \frac{k}{2}\right\rfloor-3. \label{SMeq:SMlem:SU2_sufficient_cond2}
    \end{align}
    By Eqs.~\eqref{SMeq:SMlem:SU2_sufficient_cond1} and \eqref{SMeq:SMlem:SU2_sufficient_cond2}, we get $c_{n, k}\leq\binom{n}{\lfloor k/2\rfloor+2}-\binom{n}{\lfloor k/2\rfloor+1}$.

    Next, we prove the inequality in the case of $j=\lfloor n/2\rfloor-1$. 
    When $n$ is even, we note that 
    \begin{align}
        \binom{n}{\left\lfloor \frac{n}{2}\right\rfloor}-\binom{n}{\left\lfloor \frac{n}{2}\right\rfloor-1} 
        =&\frac{1}{(\frac{n}{2})!}\prod_{\alpha=-1}^{n/2-3} (n-\alpha-1) \nonumber\\
        =&\frac{2(n-2)}{(\frac{n}{2}-1)!}(n-1)\prod_{\alpha=2}^{n/2-3} (n-\alpha-1) \nonumber\\
        \geq&\frac{2(n-2)}{(\left\lfloor \frac{k}{2}\right\rfloor+3)!(\frac{n}{2})^{n/2-\lfloor k/2\rfloor-4}}(n-1)\left(\prod_{\alpha=2}^{\lfloor k/2\rfloor+1} (n-\alpha-1)\right)\left(\frac{n}{2}\right)^{n/2-\lfloor k/2\rfloor-4} \nonumber\\
        \geq&\frac{2(n-2)}{(\lfloor \frac{k}{2}\rfloor+3)!}(n-1)\prod_{\alpha=2}^{\lfloor k/2\rfloor+1} (n-2\alpha+1) \nonumber\\
        =&\frac{2(n-2)}{(\lfloor \frac{k}{2}\rfloor+3)!}\prod_{\alpha=1}^{\lfloor k/2\rfloor+1} (n-2\alpha+1) \nonumber\\
        =&\frac{n-2}{(\left\lfloor \frac{k}{2}\right\rfloor+2)(\left\lfloor\frac{k}{2}\right\rfloor+3)2^{\lfloor k/2\rfloor-1}}c_{n, k}. \label{SMeq:SMlem:SU2_sufficient_cond3}
    \end{align}
    Since $n$ satisfies $n\geq 2^{2(\lfloor k/2\rfloor+1)}$, we have 
    \begin{align}
        \left(\left\lfloor \frac{k}{2}\right\rfloor+2\right)\left(\left\lfloor \frac{k}{2}\right\rfloor+3\right)2^{\lfloor k/2\rfloor-1} 
        =&\left(\left\lfloor \frac{k}{2}\right\rfloor+2\right)\left(\left\lfloor \frac{k}{2}\right\rfloor+4\right)2^{\lfloor k/2\rfloor-1}-\left(\left\lfloor \frac{k}{2}\right\rfloor+2\right)2^{\lfloor k/2\rfloor-1} \nonumber\\
        \leq&2^{\lfloor k/2\rfloor+3}\cdot 2^{\lfloor k/2\rfloor-1}-2 \nonumber\\
        =&2^{2(\lfloor k/2\rfloor+1)}-2 \nonumber\\
        \leq&n-2, \label{SMeq:SMlem:SU2_sufficient_cond4}
    \end{align}
    where we note that we can prove that $(j+2)(j+4)\leq 2^{j+3}$ for all $j\in\mathbb{N}$ by the mathematical induction. 
    By Eqs.~\eqref{SMeq:SMlem:SU2_sufficient_cond3} and \eqref{SMeq:SMlem:SU2_sufficient_cond4}, we get $c_{n, k}\leq\binom{n}{\lfloor n/2\rfloor}-\binom{n}{\lfloor n/2\rfloor-1}$.

    When $n$ is odd, $n+1$ is even, and we have $n+1\geq 2^{2(\lfloor k/2\rfloor+1)}$. 
    Thus we can substitute $n\mapsto n+1$ in the result in the case when $n$ is even,  
    Then, we get 
    \begin{align}
        c_{n, k}
        \leq c_{n+1, k}
        \leq\binom{n+1}{\lfloor \frac{n+1}{2}\rfloor}-\binom{n+1}{\lfloor \frac{n+1}{2}\rfloor-1}. \label{SMeq:SMlem:SU2_sufficient_cond5}
    \end{align}
    We note that 
    \begin{align}
        \binom{n+1}{\lfloor \frac{n+1}{2}\rfloor}-\binom{n+1}{\lfloor \frac{n+1}{2}\rfloor-1}
        =&\binom{n+1}{\frac{n+1}{2}}-\binom{n+1}{\frac{n-1}{2}} \nonumber\\
        =&\left(\binom{n}{\frac{n-1}{2}}+\binom{n}{\frac{n+1}{2}}\right)-\left(\binom{n}{\frac{n-3}{2}}+\binom{n}{\frac{n-1}{2}}\right) \nonumber\\
        =&\binom{n}{\frac{n+1}{2}}-\binom{n}{\frac{n-3}{2}} \nonumber\\
        =&\binom{n}{\frac{n-1}{2}}-\binom{n}{\frac{n-3}{2}} \nonumber\\
        =&\binom{n}{\lfloor \frac{n}{2}\rfloor}-\binom{n}{\lfloor \frac{n}{2}\rfloor-1}. \label{SMeq:SMlem:SU2_sufficient_cond6}
    \end{align}
    By plugging Eq.~\eqref{SMeq:SMlem:SU2_sufficient_cond6} into Eq.~\eqref{SMeq:SMlem:SU2_sufficient_cond5}, we get $c_{n, k}\leq\binom{n}{\lfloor n/2\rfloor}-\binom{n}{\lfloor n/2\rfloor-1}$. 
\end{proof}

We are going to see two properties of the sequence $(a_{n, k, j})_{j=1}^\infty$ defined by Eq.~\eqref{SMeq:a_nkj_def}, i.e., 
\begin{align}
        a_{n, k, j}:=\sum_{p=0}^j \binom{n}{j-p}\binom{n-k+p-1}{p}. 
\end{align}
First, we prepare the property that we used to get the explicit expression of the result of Theorem~\ref{SMthm:SU2_general_locality}.

\begin{lemma} \label{SMlem:sequence_property1}
    Let $n, k\in\mathbb{Z}$ satisfy $0\leq k\leq n-1$ and the sequence $(a_{n, k, j})_{j=0}^\infty$ be defined by Eq.~\eqref{SMeq:a_nkj_def}. 
    Then, 
    \begin{align}
        a_{n, 2k, k}=\frac{2^k}{k!}\prod_{\alpha=1}^k (n-2\alpha+1). \label{SMeq:SMlem:sequence_property1_1}
    \end{align}
\end{lemma}

\begin{proof}
    We are going to prove that 
    \begin{align}
        \sum_{p=0}^k \left[\frac{1}{(k-p)!}\prod_{\beta=1}^{k-p}(z-\beta+1)\right]\left[\frac{1}{p!}\prod_{\beta=1}^p (z-2k+p-\beta)\right]
        =\frac{2^k}{k!}\prod_{\alpha=1}^k (z-2\alpha+1) \label{SMeq:SMlem:sequence_property1_2}
    \end{align}
    for all $z\in\mathbb{C}$, which gives Eq.~\eqref{SMeq:SMlem:sequence_property1_1} as a special case of $z=n$. 
    We define a polynomial $q(z)$ as the l.h.s. of Eq.~\eqref{SMeq:SMlem:sequence_property1_2}. 
    Since the both sides of Eq.~\eqref{SMeq:SMlem:sequence_property1_2} are polynomials of degree $k$, it is sufficient to show that Eq.~\eqref{SMeq:SMlem:sequence_property1_2} holds for $z=2k$ and $z=2\alpha-1$ with $\alpha\in\{1, 2, ..., k\}$. 
    First, when $z=2k$, we have 
    \begin{align}
        q(2k)
        =&\sum_{p=0}^k \left[\frac{1}{(k-p)!}\prod_{\beta=1}^{k-p} (2k-\beta+1)\right]\left[\frac{1}{p!}\prod_{\beta=1}^p (p-\beta)\right] \nonumber\\
        =&\frac{1}{k!}\prod_{\beta=1}^k (2k-\beta+1) \nonumber\\
        =&\frac{1}{k!}\cdot\frac{\displaystyle\prod_{\beta=1}^{2k} (2k-\beta+1)}{\displaystyle\prod_{\beta=k+1}^{2k} (2k-\beta+1)} \nonumber\\
        =&\frac{1}{k!}\cdot\frac{\left[\displaystyle\prod_{\beta=1}^k (2k-2\beta+1)\right]\left[\displaystyle\prod_{\beta=1}^k (2k-2\beta+2)\right]}{\displaystyle\prod_{\beta=1}^k (k-\beta+1)} \nonumber\\
        =&\frac{2^k}{k!}\prod_{\beta=1}^k (2k-2\beta+1), 
    \end{align}
    where we used $\prod_{\beta=1}^p (p-\beta)=\delta_{p, 0}$ in the second equality.  
    Next, we consider the case when $z=2\alpha-1$ with $\alpha\in\{1, 2, ..., k\}$. 
    By substitution, we have 
    \begin{align}
        q(2\alpha-1)
        =\sum_{p=0}^k \left[\frac{1}{(k-p)!}\prod_{\beta=1}^{k-p}(2\alpha-\beta)\right]\left[\frac{1}{p!}\prod_{\beta=1}^p (2\alpha-2k+p-\beta-1)\right]. \label{SMeq:SMlem:sequence_property1_3}
    \end{align}
    We note that 
    \begin{align}
        \prod_{\beta=1}^p (2\alpha-2k+p-\beta-1)
        =&\prod_{\beta=1}^p [2\alpha-2k+p-(p+1-\beta)-1] \nonumber\\
        =&\prod_{\beta=1}^p (2\alpha-2k+\beta-2) \nonumber\\
        =&(-1)^p \prod_{\beta=1}^p (2k-2\alpha-\beta+2). \label{SMeq:SMlem:sequence_property1_4}
    \end{align}
    By plugging Eq.~\eqref{SMeq:SMlem:sequence_property1_4} into Eq.~\eqref{SMeq:SMlem:sequence_property1_3}, we get 
    \begin{align}
        q(2\alpha-1) 
        =&\sum_{p=0}^k (-1)^p \left[\frac{1}{(k-p)!}\prod_{\beta=1}^{k-p}(2\alpha-\beta)\right]\left[\frac{1}{p!}\prod_{\beta=1}^p (2k-2\alpha-\beta+2)\right] \nonumber\\
        =&\sum_{p\in[0, k]\cap [k-2\alpha+1, 2k-2\alpha+1]\cap\mathbb{Z}} (-1)^p \left[\frac{1}{(k-p)!}\prod_{\beta=1}^{k-p}(2\alpha-\beta)\right]\left[\frac{1}{p!}\prod_{\beta=1}^p (2k-2\alpha-\beta+2)\right] \nonumber\\
        =&\sum_{p\in[0, k]\cap [k-2\alpha+1, 2k-2\alpha+1]\cap\mathbb{Z}} (-1)^p \binom{2\alpha-1}{k-p}\binom{2k-2\alpha+1}{p}, \label{SMeq:SMlem:sequence_property1_5}
    \end{align}
    where we used $\prod_{\beta=1}^{k-p} (2\alpha-\beta)=0$ if $p\leq k-2\alpha$, and $\prod_{\beta=1}^p (2k-2\alpha-\beta+2)=0$ if $p\geq 2k-2\alpha+2$ in the second equality. 
    By noting that the condition $p\in[0, k]\cap [k-2\alpha+1, 2k-2\alpha+1]\cap\mathbb{Z}$ is invariant under the transformation $p\mapsto 2k-2\alpha+1-p$, we have 
    \begin{align}
        q(2\alpha-1) 
        =&\sum_{p\in[0, k]\cap [k-2\alpha+1, 2k-2\alpha+1]\cap\mathbb{Z}} (-1)^{2k-2\alpha+1-p} \binom{2\alpha-1}{k-(2k-2\alpha+1-p)}\binom{2k-2\alpha+1}{2k-2\alpha+1-p} \nonumber\\
        =&-\sum_{p\in[0, k]\cap [k-2\alpha+1, 2k-2\alpha+1]\cap\mathbb{Z}} (-1)^p \binom{2\alpha-1}{2\alpha-k+p-1}\binom{2k-2\alpha+1}{2k-2\alpha+1-p} \nonumber\\
        =&-\sum_{p\in[0, k]\cap [k-2\alpha+1, 2k-2\alpha+1]\cap\mathbb{Z}} (-1)^p \binom{2\alpha-1}{k-p}\binom{2k-2\alpha+1}{p}. \label{SMeq:SMlem:sequence_property1_6}
    \end{align}
    By Eqs.~\eqref{SMeq:SMlem:sequence_property1_5} and \eqref{SMeq:SMlem:sequence_property1_6}, we get $q(2\alpha-1)=0$. 
\end{proof}

Next, by using the lemma above, we derive another property of the sequence $(a_{n, k, j})_{j=0}^\infty$ for the explicit expression of the result of Theorem~\ref{SMthm:U1_general_locality}.

\begin{lemma} \label{SMlem:sequence_property2}
    Let $n, k\in\mathbb{Z}$ satisfy $0\leq k\leq n-1$ and the sequence $(a_{n, k, j})_{j=0}^\infty$ be defined by Eq.~\eqref{SMeq:a_nkj_def}. 
    Then, 
    \begin{align}
        \frac{a_{n, k, \lceil k/2\rceil}+a_{n, k, \lfloor k/2\rfloor}}{2}
        =\frac{2^{\lfloor k/2\rfloor}}{\lceil \frac{k}{2}\rceil!}\prod_{\alpha=1}^{\lceil k/2\rceil} (n-k+2\alpha-1). 
    \end{align}
\end{lemma}

\begin{proof}
    First, when $k$ is even, by Lemma~\ref{SMlem:sequence_property1}, we have 
    \begin{align}
        \frac{a_{n, k, \lceil k/2\rceil}+a_{n, k, \lfloor k/2\rfloor}}{2}
        =&a_{n, k, k/2} \nonumber\\
        =&\frac{2^{k/2}}{(\frac{k}{2})!}\prod_{\alpha=1}^{k/2} (n-2\alpha+1) \nonumber\\
        =&\frac{2^{k/2}}{(\frac{k}{2})!}\prod_{\alpha=1}^{k/2} \left[n-2\left(\frac{k}{2}+1-\alpha\right)+1\right] \nonumber\\
        =&\frac{2^{\lfloor k/2\rfloor}}{\lceil \frac{k}{2}\rceil!}\prod_{\alpha=1}^{\lceil k/2\rceil} (n-k+2\alpha-1). 
    \end{align}
    Next, we consider the case when $k$ is odd. 
    We note that 
    \begin{align}
        a_{n, k, j+1}+a_{n, k, j}
        =&\sum_{p=0}^{j+1} \binom{n}{j+1-p}\binom{n-k+p-1}{p}+\sum_{p=0}^j \binom{n}{j-p}\binom{n-k+p-1}{p} \nonumber\\
        =&\sum_{p=0}^{j+1} \left(\binom{n}{j+1-p}+\binom{n}{j-p}\right)\binom{n-k+p-1}{p} \nonumber\\
        =&\sum_{p=0}^{j+1} \binom{n+1}{j+1-p}\binom{n-k+p-1}{p} \nonumber\\
        =&a_{n, k+1, j+1}, 
    \end{align}
    which implies that 
    \begin{align}
        \frac{a_{n, k, \lceil k/2\rceil}+a_{n, k, \lfloor k/2\rfloor}}{2}
        =\frac{a_{n, k, (k+1)/2}+a_{n, k, (k-1)/2}}{2}
        =\frac{a_{n+1, k+1, (k+1)2}}{2}. 
    \end{align}
    By applying Lemma~\ref{SMlem:sequence_property1} to the r.h.s. of this equation, we get 
    \begin{align}
        \frac{a_{n, k, \lceil k/2\rceil}+a_{n, k, \lfloor k/2\rfloor}}{2}
        =&\frac{1}{2}\frac{2^{(k+1)/2}}{(\frac{k+1}{2})!}\prod_{\alpha=1}^{(k+1)2} (n-2\alpha+2) \nonumber\\
        =&\frac{2^{(k-1)/2}}{(\frac{k+1}{2})!}\prod_{\alpha=1}^{(k+1)/2} \left[n-2\left(\frac{k+1}{2}+1-\alpha\right)+2\right] \nonumber\\
        =&\frac{2^{(k-1)/2}}{(\frac{k+1}{2})!}\prod_{\alpha=1}^{(k+1)/2} (n-k+2\alpha-1) \nonumber\\
        =&\frac{2^{\lfloor k/2\rfloor}}{(\lceil \frac{k}{2}\rceil)!}\prod_{\alpha=1}^{(k+1)/2} (n-k+2\alpha-1). 
    \end{align}
\end{proof}

In the following, we prepare two properties about binomial coefficients for the proof of Lemma~\ref{SMlem:SU2_cond}. 
First, we prove a property used for proving that the equations in Lemma~\ref{SMlem:SU2_cond} imply the ones in Theorem~\ref{SMthm:general}.

\begin{lemma} \label{SMlem:SU2_coefficient1}
    Let $n, k\in\mathbb{Z}_{\geq 0}$ satisfy $n\geq k$. 
    Then, for any $j\in\{0, 1, ..., \lfloor k/2\rfloor\}$, there exists $(v_{j, l})_{j, l\in\{0, 1, ..., \lfloor k/2\rfloor\}}\in\mathbb{R}^{(\lfloor k/2\rfloor+1)^2}$ such that for any $\alpha\in\mathbb{Z}$, 
    \begin{align}
        \binom{n-k}{\alpha-j}+\binom{n-k}{\alpha-k+j}
        =\sum_{l=0}^{\lfloor k/2\rfloor} v_{j, l}\binom{n-2l}{\alpha-l}. \label{SMeq:SMlem:SU2_coefficient1_1}
    \end{align}
\end{lemma}

\begin{proof}
    We prove this lemma by the mathematical induction about $k$. 
    The statement trivially holds for $k=0$, because when $j=0$ and $v_{0, 0}=2$, Eq.~\eqref{SMeq:SMlem:SU2_coefficient1_1} holds for all $\alpha\in\mathbb{Z}$. 
    We take arbitrary $K\in\mathbb{Z}_{\geq 0}$ and suppose that this statement holds for $k=K$, i.e., for any $j\in\{0, 1, ..., \lfloor k/2\rfloor\}$, we can take $v'_{j, l\in\{0, 1, ..., \lfloor k/2\rfloor\}}\in\mathbb{R}^{(\lfloor k/2\rfloor+1)^2}$ such that for any $\alpha\in\mathbb{Z}$, 
    \begin{align}
        \binom{n-K}{\alpha-j}+\binom{n-K}{\alpha-K+j}
        =\sum_{l=0}^{\left\lfloor K/2\right\rfloor} v'_{j, l}\binom{n-2l}{\alpha-l}. \label{SMeq:SMlem:SU2_coefficient1_2}
    \end{align}
    In the following, we are going to prove that the statement holds for $k=K+1$. 
    When $K$ is even, for any $j, l\in\{0, 1, ..., K/2\}$, we set 
    \begin{align}
        v_{j, l}=(-1)^{K/2-j}\delta_{K/2, l}+\sum_{p=j}^{K/2-1} (-1)^{p-j}v'_{p, l}, \label{SMeq:SMlem:SU2_coefficient1_3}
    \end{align}
    where we note that this means that $v_{K/2, l}=\delta_{K/2, l}$. 
    Then, for any $\alpha\in\mathbb{Z}$, we have 
    \begin{align}
        \sum_{l=0}^{K/2} v_{j, l}\binom{n-2l}{\alpha-l}
        =&(-1)^{K/2-j}\binom{n-K}{\alpha-K/2}+\sum_{p=j}^{K/2-1} \left[(-1)^{p-j}\sum_{l=0}^{K/2} v'_{p, l}\binom{n-2l}{\alpha-l}\right] \nonumber\\
        =&(-1)^{K/2-j}\binom{n-K}{\alpha-K/2}+\sum_{p=j}^{K/2-1} (-1)^{p-j}\left(\binom{n-K}{\alpha-p}+\binom{n-K}{\alpha-K+p}\right) \nonumber\\
        =&\sum_{p=j}^{K-j} (-1)^{p-j}\binom{n-K}{\alpha-p} \nonumber\\
        =&\sum_{p=j}^{K-j} \left[(-1)^{p-j}\binom{n-(K+1)}{\alpha-p}-(-1)^{(p+1)-j}\binom{n-(K+1)}{\alpha-(p+1)}\right] \nonumber\\
        =&\binom{n-(K+1)}{\alpha-j}+\binom{n-(K+1)}{\alpha-K+j}. \label{SMeq:SMlem:SU2_coefficient1_4}
    \end{align}
    When $K$ is odd, for any $j\in\{0, 1, ..., (K+1)/2\}$, we set 
    \begin{align}
        &v_{j, l}=\sum_{p=j}^{(K-1)/2} (-1)^{p-j}v'_{p, l}\ \forall l\in\{0, 1, ..., (K-1)/2\}, \label{SMeq:SMlem:SU2_coefficient1_5}\\
        &v_{j, (K+1)/2}=2(-1)^{(K+1)/2-j}. \label{SMeq:SMlem:SU2_coefficient1_6}
    \end{align}
    Then, for any $\alpha\in\mathbb{Z}$, we have 
    \begin{align}
        \sum_{l=0}^{(K+1)/2} v_{j, l}\binom{n-2l}{\alpha-l} 
        =&2(-1)^{(K+1)/2-j}\binom{n-(K+1)}{\alpha-\frac{K+1}{2}}+\sum_{p=j}^{(K-1)/2} \left[(-1)^{p-j}\sum_{l=0}^{(K-1)/2} v'_{p, l}\binom{n-2l}{\alpha-l}\right] \nonumber\\
        =&2(-1)^{(K+1)/2-j}\binom{n-(K+1)}{\alpha-\frac{K+1}{2}}+\sum_{p=j}^{(K-1)/2} (-1)^{p-j}\left(\binom{n-K}{\alpha-p}+\binom{n-K}{\alpha-K+p}\right). \label{SMeq:SMlem:SU2_coefficient1_7}
    \end{align}
    We note that 
    \begin{align}
        &\sum_{p=j}^{(K-1)/2} (-1)^{p-j}\left(\binom{n-K}{\alpha-p}+\binom{n-K}{\alpha-K+p}\right) \nonumber\\
        =&\sum_{p=j}^{(K-1)/2} (-1)^{p-j}\binom{n-K}{\alpha-p}-\sum_{p=(K+1)/2}^{K-j} (-1)^{p-j}\binom{n-K}{\alpha-p} \nonumber\\
        =&\sum_{p=j}^{(K-1)/2} \left[(-1)^{p-j}\binom{n-(K+1)}{\alpha-p}-(-1)^{(p+1)-j}\binom{n-(K+1)}{\alpha-(p+1)}\right] \nonumber\\
        &\hspace{5mm} -\sum_{p=(K+1)/2}^{K-j} \left[(-1)^{p-j}\binom{n-(K+1)}{\alpha-p}-(-1)^{(p+1)-j}\binom{n-(K+1)}{\alpha-(p+1)}\right] \nonumber\\
        =&\left[\binom{n-(K+1)}{\alpha-j}-(-1)^{(K+1)/2-j}\binom{n-(K+1)}{\alpha-\frac{K+1}{2}}\right]
        -\left[(-1)^{(K+1)/2-j}\binom{n-(K+1)}{\alpha-\frac{K+1}{2}}-\binom{n-(K+1)}{\alpha-(K+1)+j}\right] \nonumber\\
        =&\left(\binom{n-(K+1)}{\alpha-j}+\binom{n-(K+1)}{\alpha-(K+1)+j}\right)-2(-1)^{(K+1)/2-j}\binom{n-(K+1)}{\alpha-\frac{K+1}{2}}. \label{SMeq:SMlem:SU2_coefficient1_8}
    \end{align}
    By plugging Eq.~\eqref{SMeq:SMlem:SU2_coefficient1_8} into Eq.~\eqref{SMeq:SMlem:SU2_coefficient1_7}, we get 
    \begin{align}
        \sum_{l=0}^{(K+1)/2} v_{j, l}\binom{n-2l}{\alpha-l}
        =\binom{n-(K+1)}{\alpha-j}+\binom{n-(K+1)}{\alpha-(K+1)+j}. 
    \end{align}
    We have thus proven that the statement holds for $k=K+1$. 
\end{proof}

Next, we prove a property used for proving that the equations in Theorem~\ref{SMthm:general} imply the ones in Lemma~\ref{SMlem:SU2_cond}.

\begin{lemma} \label{SMlem:SU2_coefficient2}
    Let $n, k\in\mathbb{Z}_{\geq 0}$ satisfy $n\geq k$. 
    Then, for any $j\in\{0, 1, ..., \lfloor k/2\rfloor\}$, there exists $(w_{j, l})_{j, l\in\{0, 1, ..., \lfloor k/2\rfloor\}}\in\mathbb{R}^{(\lfloor k/2\rfloor+1)^2}$ such that for any $\alpha\in\mathbb{Z}$, 
    \begin{align}
        \binom{n-2j}{\alpha-j}
        =\sum_{l=0}^{\lfloor k/2\rfloor} w_{j, l}\left(\binom{n-2l}{\alpha}+\binom{n-2l}{\alpha-2l}\right). \label{SMeq:SMlem:SU2_coefficient2_1}
    \end{align}
\end{lemma}

\begin{proof}
    We prove this lemma by the mathematical induction about $k$. 
    We take arbitrary $K\in\mathbb{Z}_{\geq 0}$ and suppose that this lemma holds for $k=K$, i.e., for any $j\in\{0, 1, ..., \lfloor K/2\rfloor\}$, we can take $(w'_{j, l})_{j, l\in\{0, 1, ..., \lfloor k/2\rfloor\}}\in\mathbb{R}^{(\lfloor k/2\rfloor+1)^2}$ such that for any $\alpha\in\mathbb{Z}$, 
    \begin{align}
        \binom{n-2j}{\alpha-j}
        =\sum_{l=0}^{\lfloor k/2\rfloor} w'_{j, l}\left(\binom{n-2l}{\alpha}+\binom{n-2l}{\alpha-2l}\right). \label{SMeq:SMlem:SU2_coefficient2_2}
    \end{align}
    We are going to prove that this lemma holds for $k=K+1$. 
    When $K$ is even, the statement trivially holds for $k=K+1$, because Eq.~\eqref{SMeq:SMlem:SU2_coefficient2_1} is equivalent for $k=K$ and $k=K+1$. 
    In the following, we consider the case when $K$ is odd. 
    By Lemma~\ref{SMlem:SU2_coefficient1}, for any $p\in\{0, 1, ..., (K-1)/2\}$, we can take $(v_{p, q})_{p, q\in\{0, 1, ..., (K-1)/2\}}\in\mathbb{R}^{[(K+1)/2]^2}$ such that for any $\alpha\in\mathbb{Z}$, 
    \begin{align}
        \binom{n-K}{\alpha-p}+\binom{n-K}{\alpha-K+p}=\sum_{q=0}^{(K-1)/2} v_{p, q}\binom{n-2q}{\alpha-q}. \label{SMeq:SMlem:SU2_coefficient2_3}
    \end{align}
    For $j\in\{0, 1, ..., (K-1)/2\}$, we set $w_{j, l}=w'_{j, l}$ and $w_{j, (K+1)/2}=0$. 
    Then, we have Eq.~\eqref{SMeq:SMlem:SU2_coefficient2_1} for all $\alpha\in\mathbb{Z}$. 
    For $j=(K+1)/2$, we set 
    \begin{align}
        &w_{(K+1)/2, l}=-\frac{1}{2}\sum_{p=0}^{(K-1)/2}\sum_{q=0}^{(K-1)/2} (-1)^{(K+1)/2+p}v_{p, q}w'_{q, l}\ \forall l\in\left\{0, 1, ..., \frac{K-1}{2}\right\}, \label{SMeq:SMlem:SU2_coefficient2_4}\\
        &w_{(K+1)/2, (K+1)/2}=\frac{1}{2}(-1)^{(K+1)/2}. \label{SMeq:SMlem:SU2_coefficient2_5}
    \end{align} 
    For any $\alpha\in\mathbb{Z}$, by plugging $j=0$ into Eq.~\eqref{SMeq:SMlem:SU2_coefficient1_8} in Lemma~\ref{SMlem:SU2_coefficient1}, we get 
    \begin{align}
        \binom{n-(K+1)}{\alpha-\frac{K+1}{2}} 
        =&-\frac{1}{2}(-1)^{(K+1)/2}\sum_{p=0}^{(K-1)/2} (-1)^p\left(\binom{n-K}{\alpha-p}+\binom{n-K}{\alpha-K+p}\right) \nonumber\\
        &\hspace{5mm} +\frac{1}{2}(-1)^{(K+1)/2}\left(\binom{n-(K+1)}{\alpha}+\binom{n-(K+1)}{\alpha-(K+1)}\right). \label{SMeq:SMlem:SU2_coefficient2_6}
    \end{align}
    By plugging Eq.~\eqref{SMeq:SMlem:SU2_coefficient2_2} into Eq.~\eqref{SMeq:SMlem:SU2_coefficient2_3}, we have 
    \begin{align}
        \binom{n-K}{\alpha-p}+\binom{n-K}{\alpha-K+p} 
        =\sum_{q=0}^{(K-1)/2} \sum_{l=0}^{(K-1)/2} v_{p, q}w'_{q, l}\left(\binom{n-2l}{\alpha}+\binom{n-2l}{\alpha-2l}\right). \label{SMeq:SMlem:SU2_coefficient2_7}
    \end{align}
    By plugging Eq.~\eqref{SMeq:SMlem:SU2_coefficient2_7} into Eq.~\eqref{SMeq:SMlem:SU2_coefficient2_6}, we get 
    \begin{align}
        \binom{n-(K+1)}{\alpha-\frac{K+1}{2}} 
        =&-\frac{1}{2}\sum_{p=0}^{(K-1)/2} \sum_{q=0}^{(K-1)/2} \sum_{l=0}^{(K-1)/2} (-1)^{(K+1)/2+p}v_{p, q}w'_{q, l}\left(\binom{n-2l}{\alpha}+\binom{n-2l}{\alpha-2l}\right) \nonumber\\
        &\hspace{5mm} +\frac{1}{2}(-1)^{(K+1)/2}\left(\binom{n-(K+1)}{\alpha}+\binom{n-(K+1)}{\alpha-(K+1)}\right) \nonumber\\
        =&\sum_{l=0}^{(K+1)/2} w_{(K+1)/2, l}\left(\binom{n-2l}{\alpha}+\binom{n-2l}{\alpha-2l}\right), 
    \end{align}
    where we used Eqs.~\eqref{SMeq:SMlem:SU2_coefficient2_4} and \eqref{SMeq:SMlem:SU2_coefficient2_5}. 
    We have thus proven that the statement holds for $k=K+1$. 
\end{proof}

\end{document}